%% file: Kerr.tex
\documentclass{article}
\usepackage{wasysym}
\let\iint\relax
\let\iiint\relax
\usepackage{amsmath,amsthm}
\usepackage{amssymb}
\usepackage{upgreek}
\usepackage{graphics}
\usepackage{marvosym}
\usepackage{color}
\usepackage{bbm}

\usepackage{pifont}

\newtheorem{proposition}{Proposition}
\newtheorem{remark}{Remark}
\newtheorem{corollary}[proposition]{Corollary}
\newtheorem{theorem}{Theorem}

\newtheorem{maintheorem}{Main Theorem}
\newtheorem{definition}{Definition}

\usepackage[margin=1in,dvips]{geometry}

\newcommand{\nabb}{{\nabla} \mkern-13mu /\,}
\newcommand{\lapp}{{\Delta} \mkern-13mu /\,}
\newcommand{\lessflat}{\reflectbox{\mbox{$\flat$}}}

\title{A Vector Field Method Approach to Improved Decay for Solutions to the Wave Equation on a Slowly Rotating Kerr Black Hole}
\author{Jonathan Luk\thanks{Mathematics Department, Princeton University, Princeton, NJ 08540. E-mail: jluk@math.princeton.edu}}

\begin{document}

\maketitle

\begin{abstract}
We prove that sufficiently regular solutions to the wave equation $\Box_{g_K}\Phi=0$ on the exterior of a sufficiently slowly rotating Kerr black hole obey the estimates $|\Phi|\leq C (t^*)^{-\frac{3}{2}+\eta}$ on a compact region of $r$. This is proved with the help of a new vector field commutator that is analogous to the scaling vector field on Minkowski and Schwarzschild spacetime. This result improves the known robust decay rates that are proved using the vector field method in the region of finite $r$ and along the event horizon.
\end{abstract}

\tableofcontents

\section{Introduction}
A major open problem in general relativity is that of the nonlinear stability of Kerr spacetimes. These spacetimes are stationary axisymmetric asymptotically flat black hole solutions to the vacuum Einstein equations
$$R_{\mu\nu}=0$$
in $3+1$ dimensions. They are parametrized by two parameters $\left(M, a\right)$, representing respectively the mass and the specific angular momentum of a black hole. (See Section \ref{geometry}). It is conjectured that Kerr spacetimes are asymptotically stable. In the framework of the initial value problem, the stability of Kerr would mean that for any solution to the vacuum Einstein equations with initial data close to the initial data of a Kerr spacetime, its maximal Cauchy development has an exterior region that approaches a nearby, but possibly different, Kerr spacetime.

In order to study the stability of Kerr spacetimes, it is important to first understand the corresponding linear problem. One way to approach this is to study the linear scalar wave equation $\Box_{g_K}\Phi=0$, where $g_K$ is the metric on a fixed Kerr background and $\Box_{g_K}$ is the Laplace-Beltrami operator. This can be compared with the proofs of the nonlinear stability of the Minkowski spacetime in which a robust understanding of the quantitative decay properties of solutions to the linear wave equation plays a fundamental role \cite{CK}, \cite{LR}.

The Kerr family of spacetimes contains a one-parameter subfamily known as the Schwarzschild spacetimes for which $a=0$. It is natural when studying the wave equation on Kerr spacetimes to begin by focusing on the wave equation on Schwarzschild spacetimes. Pointwise boundedness and decay of the solutions to the wave equation on Schwarzschild spacetimes has been proved in \cite{W}, \cite{KW}, \cite{MS}, \cite{BSt}, \cite{DRS}, \cite{Kr}, \cite{BSo} \cite{DSS}, \cite{Ta}. In particular Dafermos and Rodnianski used the vector field method to show that on the exterior region of the Schwarzschild spacetimes, including along the event horizon, solutions to the linear wave equation satisfy $|\Phi|\leq C (t^*)^{-1}$, where $t^*$ is a regular coordinate (up to the event horizon) that approaches infinity towards null infinity. In an earlier work \cite{L}, we improved this decay rate. More precisely, we showed that sufficiently regular solutions to the wave equation $\Box_g\Phi=0$ on the Schwarzschild black hole obey the estimates $|\Phi|\leq C_\eta (t^*)^{-\frac{3}{2}+\eta}$ for any $\eta>0$ on a compact region of $r$, including along the event horizon and inside the black hole. 

This paper generalizes the above result to Kerr spacetimes where $a\ll M$. For Kerr spacetimes satisfying this condition, Dafermos and Rodnianski \cite{DRK}, and subsequently Andersson and Blue \cite{AB}, have proved a decay rate in the exterior region of the Kerr spacetime, including along the event horizon, of $|\Phi|\leq C(t^*)^{-1+\eta}$, where $t^*$ is a regular coordinate to be defined later, and with $t^*$ we will define a foliation of the exterior region of Kerr spacetime by the spacelike hypersurfaces $\Sigma_{t^*}$. Extending the methods in \cite{L}, we are able to improve this decay rate using the vector field method. In particular, we have

\begin{theorem}\label{maintheorem}
Suppose $\Box_{g_K}\Phi=0$. Then for all $\eta>0$ and all $M>0$ there exists $a_0$ such that the following estimates hold on Kerr spacetimes with $(M,a)$ for which $a\leq a_0$:
\begin{enumerate}
\item Improved Decay of Non-degenerate Energy
\begin{equation*}
\begin{split}
&\sum_{j=0}^M\int_{\Sigma_{t^*}\cap\{r\leq R\}} \left(D^j\Phi\right)^2\leq C_RE_M(t^*)^{-3+\eta}.
\end{split}
\end{equation*}
\item Improved Pointwise Decay\\
\begin{equation*}
\begin{split}
&\sum_{j=0}^M|D^j\Phi|\leq C_RE'_M(t^*)^{-\frac{3}{2}+\eta}\quad\mbox{for }r\leq R.
\end{split}
\end{equation*}
\end{enumerate}
Here, $D$ denotes derivatives in a regular coordinate system (See Section \ref{geometry}). $E_M$ and $E'_M$ depend only on $M$ and some weighted Sobolev norm of the initial data.
\end{theorem}

A more precise version of this theorem will be given in Section \ref{sectionmaintheorem}. Our proof relies on an analogue of the scaling vector field on Minkowski spacetime. Recall that in Minkowski spacetime the vector field $S=t\partial_t+r\partial_r$ is conformally Killing and satisfies $[\Box_{m},S]=2\Box_m$. Hence any estimates that hold for $\Phi$ a solution to $\Box_m\Phi=0$ would also hold for $S\Phi$. However, $S$ has a weight that is increasing with $t$. Hence one can hope to prove a better estimate for $\Phi$ using the estimates for $S\Phi$. (See, for example, \cite{KS}).

In \cite{L}, we introduced an analogue of the scaling vector field on Schwarzschild spacetimes. We defined, in the Regge-Wheeler tortoise coordinate (see Section \ref{geometry}), the vector field $S=t\partial_t+r^*\partial_{r^*}$. In \cite{L}, we studied the commutator $[\Box_{g_S},S]$ and showed that all the error terms can be controlled. Thus, up to a loss of $t^\eta$ (for $\eta$ arbitrarily small), $S\Phi$ obeys all the estimates of $\Phi$ that were proved in \cite{DRS}. In particular, we showed that $S\Phi$, like $\Phi$ itself, obeys a local integrated decay estimate
$$\int_{t'}^{t}\int_{r_1}^{r_2}(D^k\Phi)^2drdt \leq CE_k(t')^{-2}\quad\mbox{for }t'\leq t\leq (1.1)t'.,$$
$$\int_{t'}^{t}\int_{r_1}^{r_2}\left(SD^k\Phi\right)^2drdt \leq CE_k(t')^{-2+\eta}\quad\mbox{for }t'\leq t\leq (1.1)t'.$$
From this we proved an improved decay of the $L^2$ norm of $D^k\Phi$. We will explain the main idea in the case $k=0$. Firstly, the local integrated decay for $\Phi$ would already imply that on a sequence of $t_i$ slices, with $t_i\leq t_{i+1}\leq (1.1)^2 t_i$, $\Phi$ obeys a better decay rate, namely $\Phi(t_i) \leq Ct_i^{-\frac{3}{2}}$. We then introduced a new method to use the estimates for $S\Phi$, which can be explained heuristically as follows. Given any time $t$, we find the largest $t_i\leq t$ that has a better decay rate. Then we integrated from $t_i$ to $t$ using the vector field $S$. Notice at this point that $S$ has a weight that grows like $t$. Hence we have, at least schematically,
$$\int_{r_1}^{r_2}\Phi(t)^2dr\leq C\left(\int_{r_1}^{r_2}\Phi(t_i)^2dr+t^{-1}|\int_{t_i}^{t}\int_{r_1}^{r_2}S\left(\Phi^2\right)drdt|\right).$$ 
We then notice that the last term can be estimated by the local integrated decay estimates
$$|\int_{t_i}^{t}\int_{r_1}^{r_2}S\left(\Phi^2\right)drdt|\leq C\left(\int_{t_i}^{t}\int_{r_1}^{r_2}\Phi^2drdt+\int_{t_i}^{t}\int_{r_1}^{r_2}\left(S\Phi\right)^2drdt\right)\leq Ct^{-2+\eta}.$$
Putting these together, we would get 
$$\int_{r_1}^{r_2}\Phi(t)^2dr\leq Ct^{-3+\eta}.$$
Using this method, we also showed the improved decay for the $L^2$ norm of higher derivatives. Pointwise decay estimate thus followed from standard Sobolev embedding.

In this paper we would like to carry out a similar argument. We introduce a scaling vector field (which we again call $S$) which is the same as in \cite{L} at the asymptotically flat end, but is smooth up to and across the event horizon. We will prove a local integrated decay estimate for $S\Phi$ and use the argument in \cite{L} as outlined above to prove an improved decay rate. The most difficult part of the argument is to control the error terms coming from the commutation of $\Box_{g_K}$ and (the modified) $S$, i.e., the term $[\Box_{g_K},S]\Phi$. To control this, we need to use estimates for derivatives of $\Phi$, which in turn is provided by the energy estimates for the homogeneous equation $\Box_{g_K}\Phi=0$ proved in \cite{DRK} and \cite{DRL}. This term schematically looks like

\begin{equation}\label{commsch}
[\Box_{g_K},S]\Phi=O(1)\Box_{g_K}\Phi+O(r^{-2+\delta})(D^2\Phi+D\Phi+rD\nabb\Phi),
\end{equation}
where $\nabb$ is an angular derivative on the 2-sphere. The term $O(1)\Box_{g_K}\Phi$ vanishes since we are considering $\Box_{g_K}\Phi=0$. We note that the other terms have the following two desirable features. Firstly, although $S$ has a weight in $t^*$, the commutator is independent of $t^*$, which is a result of $\partial_{t^*}$ being a Killing vector field. Secondly, these terms decay as $r\to\infty$, which is a result of the asymptotic flatness of Kerr spacetimes. The last term would appear to have less decay in $r$, which is also the case in Schwarzschild spacetimes. In that case, we controlled this term in \cite{L} by commuting the equation with $\Omega$, the generators of the spherical symmetry of Schwarzschild spacetimes. The quantity $\Omega\Phi$ would then give us control over an extra power of $r$. One difficulty that arises in the case of Kerr spacetimes is that they are not spherically symmetric. Nevertheless, following \cite{DRL}, we can construct an analog of $\Omega$, call it $\tilde{\Omega}$, which is an asymptotic symmetry, i.e., the commutator $[\Box_{g_K},\tilde{\Omega}]$ would decay in $r$. The non-degenerate energy of $\tilde{\Omega}\Phi$ would then control the last term in the above expression. Moreover, it is sufficient to define $\tilde{\Omega}$ only when $r$ is very large since otherwise the factor in $r$ can be absorbed by constants. Notice however that in a finite region of $r$, the commutator $[\Box_{g_K},S]$ would in general be large. 

In order to understand what quantities of $S\Phi$ have to be controlled, we re-derive the energy estimates in \cite{DRL} in the slightly more general case of the inhomogeneous equation $\Box_{g_K}\Phi=G$. This would also immediately imply that for the linear inhomogeneous equation $\Box_{g_K}\Phi=G$ with sufficiently regular and sufficiently decaying (both in space and time) $G$, the solution $\Phi$, assuming that the initial data is sufficiently regular, would decay with a rate of $(t^*)^{-1+\eta}$, precisely as that in \cite{DRL}. We will then apply this to the equations for $\tilde{\Omega}\Phi$ and $S\Phi$. In order to derive these energy estimates, we will use the (non-Killing) vector field multipliers $N$ and $Z$. $N$ is a modification of $\partial_{t^*}$ so that it is timelike everywhere, including near the event horizon. The use of $N$ tackles the issue of superradiance, a difficulty that arises from the spacelike nature of $\partial_{t^*}$ near the event horizon. $Z$ is an analogue of the conformal vector field $u^2\partial_u+v^2\partial_v$ in Minkowski spacetime and is used to prove decay.

Since we will use vector field multipliers that has weights in $t^*$ and $r$, in order to prove the energy estimates at $t^*=\tau$ for the inhomogeneous equation would have to control the term (as well as other similar or more easily controlled terms):
$$\iint_{\mathcal R(\tau_0,\tau)} (t^*)^2r^{1+\delta}G^2,$$
where the integration over space and the $t^*$ interval $[\tau_0,\tau]$. In order to prove the energy estimates for $S\Phi$, we need to show that for $G$ as in (\ref{commsch}), this is bounded by $C(\tau)^\eta$. We split this into two parts: $r\leq \frac{t^*}{2}$ and $r\geq\frac{t^*}{2}$. For $r\leq\frac{t^*}{2}$, first we notice that since $G$ decays in $r$, we can replace $r^{1+\delta}$ by $r^{-3+2\delta}$. Then, we use the fact that 
$$\sum_{k=1}^N\iint_{\mathcal R((1.1)^{-1}\tau,\tau)\cap\{r\leq\frac{t^*}{2}\}} r^{-1+\delta}(D^k\Phi)^2\leq C\tau^{-2+\eta}.$$
Hence if we sum up the whole integral by integrating in $[\tau_0,(1.1)\tau_0]$, $[(1.1)\tau_0, (1.1)^2\tau_0]$ etc., we will get a bound of $$\sum_{i=0}^{\lfloor \log\tau\rfloor +1} (1.1)^i\tau_0 \sim_\eta \tau^\eta.$$
For $r\geq\frac{t^*}{2}$, we do not have a decay estimate for the integrated in time estimate. However, we would still have an almost boundedness estimate:
$$\sum_{k=1}^N\iint_{\mathcal R((1.1)^{-1}\tau,\tau)\cap\{r\geq\frac{t^*}{2}\}} r^{-1+\delta}(D^k\Phi)^2\leq C\tau^{\eta}.$$
Notice, moreover, that $G^2\sim r^{-3+\delta}(D^k\Phi)^2$ and this region we have $r^{-3+\delta}\leq (t^*)^{-2}r^{-1+\delta}$. Hence we again have
$$\sum_{k=1}^N\iint_{\mathcal R((1.1)^{-1}\tau,\tau)\cap\{r\geq\frac{t^*}{2}\}} G^2\leq C\tau^{-2+\eta},$$
and the required estimate followed in the same manner as in the case $r\leq\frac{t^*}{2}$.

With the modified $S$, which is smooth up to the event horizon (contrary to \cite{L}), we can prove the improved decay estimates for the $L^2$ norm of $\Phi$ and $D\Phi$ once these error terms are controlled. Using the commutation with the Killing vector $\partial_{t^*}$, we would also have control for $L^2$ norm of $D\partial_{t^*}^k\Phi$. Away from the event horizon, this is sufficient to control all other derivatives by elliptic estimates. However, since near the event horizon, $\partial_{t^*}$ is not Killing, we would not have control over other derivatives. Here, we follow \cite{DRK} and \cite{DRL} and commute the equation with a version of the red-shift vector field - $\hat{Y}$. Once we control $D\hat{Y}^k\partial_{t^*}^j\Phi$ we can use the wave equation to control (any derivatives of) $\lapp\Phi$, where $\lapp$ is the Laplace-Beltrami operator on the sphere, which is elliptic. We can thus control derivatives of $\Phi$ in any directions. We will show, moreover, that the commutator $[\Box_{g_K},\hat{Y}]$ has the property that the inhomogeneous terms can be controlled once we have controlled the $L^2$ norm of $D\partial_{t^*}^k\Phi$. This implies that $\hat{Y}\Phi$ would decay in the same rate as $\partial_{t^*}\Phi$ for which we have already derived an improved decay rate. 

We now turn to some history of this problem. We mention some results on Kerr spacetimes with $a>0$ here and refer the readers to \cite{DRL}, \cite{L} for references on the corresponding problem on Schwarzschild spacetimes. There has been a large literature on the mode stability and nonquantitative decay of azimuthal modes (See for example \cite{PT}, \cite{HW}, \cite{Wh}, \cite{FKSY}, \cite{FKSY2} and references in \cite{DRL}). The first global result for the Cauchy problem was obtained by Dafermos-Rodnianski in \cite{DRK}, in which they proved that for a class of small, axisymmetric, stationary perturbations of Schwarzschild spacetime, which include Kerr spacetimes that rotate sufficiently slowly, solutions to the wave equation are uniformly bounded. Similar results were obtained later using an integrated decay estimate on slowly rotating Kerr spacetimes by Tataru-Tohaneanu \cite{TT}. Using the integrated decay estimate, Tohaneanu also proved Strichartz estimates \cite{To}.

Decay for general solutions to the wave equation on sufficiently slowly rotating Kerr spacetimes was first proved by Dafermos-Rodnianski \cite{DRL} with a quantitative rate of $|\Phi|\leq C(t^*)^{-1+Ca}$. A similar result was later obtained by \cite{AB} using a physical space construction to obtain an integrated decay estimate. In all of \cite{TT}, \cite{DRL} and \cite{AB}, the integrated decay estimate is proved and plays an important role. All proofs of such estimates rely heavily on the separability of the wave equation, or equivalently, the existence of a Killing tensor on Kerr spacetime. In a recent work \cite{DRNPS}, Dafermos-Rodnianski shows that assuming the integrated decay estimate (non-degenerate up to the event horizon if it exists) and boundedness for the wave equation on an asymptotically flat spacetime, the decay rate $|\Phi|\leq C(t^*)^{-1}$ holds. This in particular improves the rates in \cite{DRL} and \cite{AB}. In a similar framework, but assuming in addition exact stationarity, Tataru \cite{Ta} proved a local decay rate of $(t^*)^{-3}$ using Fourier-analytic methods. This applies in particular to sufficiently slowly rotating Kerr spacetimes. Dafermos and Rodnianski have recently announced a proof for the decay of solutions to the wave equation on the full range of sub-extremal Kerr spacetimes $a< M$. 

In view of the nonlinear problem, it is important to understand decay in a robust manner. In particular, past experience shows that refined decay estimates might not be stable in nonlinear problems. The vector field method is known to be robust and culminated in the proof of the stability of the Minkowski spacetime \cite{CK}, \cite{LR}. We prove our decay result using the vector field method with the expectation that the method will be useful for studying nonlinear problems. As a model problem, we will study the semilinear equation with a null condition on a fixed slowly rotating Kerr background. In a forthcoming paper that we will show the global existence of solutions with small initial data for this class of equations. We will also study the asymptotic behavior of these solutions. The null condition, which is a special structure of the nonlinearity, has served as an important model for the proofs of the nonlinear stability of Minkowski spacetime and we hope that it will find relevance to the problem of the nonlinear stability of Kerr spacetime.

We end the introduction with an overview of the paper. In Section \ref{geometry}, we will introduce the Kerr geometry, including a few different coordinate systems that we will find useful in the rest of the paper. In Section \ref{commutators}, we introduce the (non-Killing) vector field commutators that will be used in this paper. These include the scaling vector field $S$, which is the main tool for obtaining improved decay rates in this paper. In Section \ref{currents}, we introduce the formalism for vector field multipliers. We then have all the notations necessary to state the precise form of our main theorem in Section \ref{sectionmaintheorem}. In Section \ref{sectionN}, \ref{sectionX} and \ref{sectionZ}, we prove the main energy estimates using the vector field multipliers $N, X$ and $Z$. We write down the energy estimates in the most general form, allowing for the possibility to control the inhomogeneous terms in different energy norms. Such generality is unnecessary for the result in this paper, but will be useful in studying the null condition. Starting from Section \ref{sectionhomo}, we return to the homogeneous equation. In Section \ref{sectionhomo}, we write down the energy estimates proved in \cite{DRL}. We then derive the energy estimates after commuting with $\hat{Y}$, $\tilde{\Omega}$ and $S$ in Sections \ref{sectioncommutatorY}, \ref{sectioncommutatoromega} and \ref{sectioncommutatorS} respectively. Finally, using the estimates for $S\Phi$, we prove the main theorem in Section \ref{sectionproof}.

\section{Geometry of Kerr Spacetime}\label{geometry}
\subsection{Kerr Coordinates}
The Kerr metric in the Boyer-Lindquist coordinates takes the following form:
\begin{equation}\label{kerrmetric}
\begin{split}
g_K=&-\left(1-\frac{2M}{r\left(1+\frac{a^2 \cos ^2{\theta}}{r^2}\right)}\right)dt^2+\frac{1+\frac{a^2\cos ^2 \theta}{r^2}}{1-\frac{2M}{r}+\frac{a^2}{r^2}}dr^2+r^2\left(1+\frac{a^2\cos ^2 \theta}{r^2}\right)d\theta^2 \\
&+r^2\left(1+\frac{a^2}{r^2}+\left(\frac{2M}{r}\right)\frac{a^2\sin ^2\theta}{r^2\left(1+\frac{a^2\cos ^2\theta}{r^2}\right)}\right)\sin ^2\theta d\phi^2
-4M \frac{a \sin ^2\theta}{r\left(1+\frac{a^2\cos ^2\theta}{r^2}\right)}dtd\phi.
\end{split}
\end{equation}
In this paper, we will consider Kerr spacetimes with $a$ small. It can then be thought of as a small perturbation of Schwarzschild spacetimes because by setting $a=0$, we retrieve the metric for Schwarzschild spacetimes:
\begin{equation*}
g_S=-\left(1-\frac{2M}{r_S}\right)dt_S^2+\left(1-\frac{2M}{r_S}\right)^{-1}dr_S^2+r_S^2 d\theta^2+r_S^2\sin^2 \theta d\phi^2.
\end{equation*}
The Cauchy development of Kerr spacetimes can be described schematically by taking a two dimensional slice as in the following diagram:

\begin{figure}[htbp]
\begin{center}
 
\input{penrose.pstex_t}
 
\caption{Kerr spacetime}
\end{center}
\end{figure}
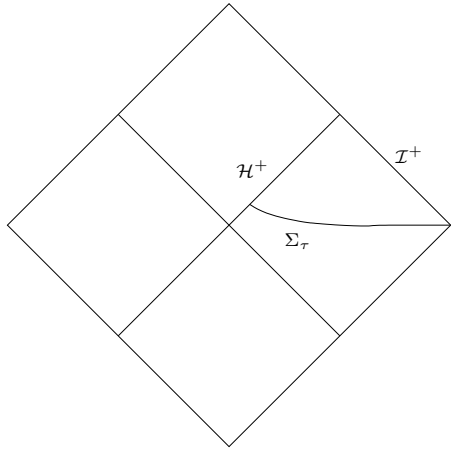

Notice that (\ref{kerrmetric}) represents the metric on the exterior region (the right hand side in the diagram). In the coordinate system above, this is the region $\{r\geq r_+\}$, where $r_+$ is the larger root of $\Delta=r^2-2Mr+a^2$. This is the region that we will study. We foliate the exterior region of the Kerr spacetime by hypersurfaces $\Sigma_\tau$ as depicted in the diagram. A precise definition of the hypersurface $\Sigma_\tau$ will be given in in Section \ref{integration}. The coordinates in (\ref{kerrmetric}) are not regular at the event horizon $\mathcal H^+=\{r=r_+\}$. It will be helpful in the sequel to use different coordinate systems on Kerr spacetimes. From now on we will call the coordinate system on which the metric (\ref{kerrmetric}) is defined the Kerr $(t,r,\theta,\phi)$ coordinates.
We define a new coordinate system - the Kerr $(t, r^*, \theta,\phi)$ coordinates by
$$\frac{dr^*}{dr}=\frac{r^2+a^2}{\Delta},$$
where $\Delta=r^2-2Mr+a^2$ is zero at the event horizon.
In this coordinate system, the metric looks like
\begin{equation*}
\begin{split}
g_K=&-\left(1-\frac{2M}{r(1+\frac{a^2 \cos ^2{\theta}}{r^2})}\right)dt^2+\Delta\left(r^2+a^2\right)^2\left(r^2+a^2\cos^2\theta\right){dr^*}^2+r^2\left(1+\frac{a^2\cos ^2 \theta}{r^2}\right)d\theta^2\\
&+r^2\left(1+\frac{a^2}{r^2}+(\frac{2M}{r})\frac{a^2\sin ^2\theta}{r^2(1+\frac{a^2\cos ^2\theta}{r^2})}\right)\sin ^2\theta d\phi^2-4M \frac{a \sin ^2\theta}{r(1+\frac{a^2\cos ^2\theta}{r^2})} dt d\phi.
\end{split}
\end{equation*}
Notice that since the definition of $r^*$ depends only on $r$, it would be unambiguous to talk about the vector $\partial_t$.\\
\subsection{Schwarzschild Coordinates}
In order to compare calculations on Kerr spacetimes to calculations on Schwarzschild spacetimes, it is helpful to exhibit a diffeomorphism between the two. We do so by defining an explicit map between the coordinate functions $\left(t, r, \theta,\phi\right)$ on a Kerr spacetime and the coordinate functions $\left(t_S, r_S, \theta_S,\phi_S\right)$ on a Schwarzschild spacetime with the same mass. These will be defined differently near and away from the event horizon. Take
$$\chi(r)=\left\{\begin{array}{clcr}1&r\le r^-_Y-\frac{r^-_Y-r_+}{2}\\0&r\ge r^-_Y-\frac{r^-_Y-r_+}{4}\end{array}\right.,$$
where $r_+$, as above, is the larger root of $\Delta=r^2-2Mr+a^2$ and $r^-_Y>r_+$ is a constant to be determined later.
With this $\chi(r)$, we can then define
$$r_S^2-2Mr_S=r^2-2Mr+a^2,$$
$$t_S+\chi(r_S)2M\log\left(r_S-2M\right)=t+\chi(r)h(r),\quad \mbox{where  }\frac{dh(r)}{dr}=\frac{2Mr}{r^2-2Mr+a^2},$$
$$\theta_S=\theta,$$
$$\phi_S=\phi+\chi(r)P(r),\quad \mbox{where  }\frac{dP(r)}{dr}=\frac{a}{r^2-2Mr+a^2}.$$
Then, by identifying $\left(t_S, r_S, \theta_S,\phi_S\right)$ with the corresponding coordinate functions on Schwarzschild spacetimes, we have a diffeomorphism between Kerr spacetimes and Schwarzschild spacetimes. This coordinate system will be used and will be called the Schwarzschild $\left(t_S, r_S, \theta_S,\phi_S\right)$ coordinates on Kerr spacetimes. Once we have this diffeomorphism, we can put any system of Schwarzschild coordinates on Kerr spacetimes. These include the Schwarzschild $(t^*_S, r_S, \theta_S, \phi_S)$ coordinates, where $t^*_S=t_S+\chi(r_S)2M\log\left(r_S-2M\right)$ and $r_S, \theta_S, \phi_S$ are defined as above. We also define $t^*=t^*_S=t_S+\chi(r_S)2M\log\left(r_S-2M\right)$  and use the Kerr $\left(t^*, r, \theta, \phi^*\right)$ coordinates. Notice that $\partial_{t^*}= \partial_{t^*_S}$.

It is common to denote on Schwarzschild spacetimes $\mu=\frac{2M}{r_S}$. We will take the same notation on Kerr spacetimes, with the understand that it is always defined with respect to the Schwarzschild $r_S$ coordinates. In particular $(1-\mu )$ approaches $0$ as $r\to r_+$ (the event horizon). 

Another system of Schwarzschild coordinates can be defined by considering two coordinate charts on the standard unit 2-sphere and introducing a system of coordinates $(x^A_S,x^B_S)$ on each of them. We then define the Schwarzschild $(t^*_S, r_S, x^A_S, x^B_S)$ coordinates in the obvious manner. Using this coordinate system and the diffeomorphism as above, we have, for small $a$: 
\begin{equation}\label{smallness}
|(g_K)_{\alpha\beta}-(g_S)_{\alpha\beta}|\leq \epsilon r^{-2}.
\end{equation}
This smallness assumption will be used throughout this paper.
\subsection{Null Frame near Event Horizon}
Some extra cancellations for the estimates near the event horizon are best captured using a null frame. Hence we define a null frame $\{\hat{V},\hat{Y},E_1, E_2\}$ in the region $r\leq r^-_Y$, where $r^-_Y$ is to be determined later. On the event horizon, $$V=\partial_{t^*}+\frac{a}{2Mr_+}\partial_{\phi^*}$$ is the Killing null generator. A direct computation shows that it satisfies
$$\nabla_V V=\kappa V,$$
where $\kappa$ is a strictly positive number on the event horizon. We want to extend $V$ to a null frame. On the event horizon, define $\hat{Y}$ first on a 2-sphere given by a fixed $t^*$ to be null, orthogonal to the 2-sphere and require that $g_K(V,\hat{Y})=-2$. Define also locally an orthonormal frame $\{E_1,E_2\}$ tangent to the fixed 2-sphere. Notice that in the sequel, we will only need to work with a local null frame. We then extend this null frame off the fixed 2-sphere on the event horizon (with $\hat{V}|_{\mathcal H}=V$) by requiring 
\begin{equation}\label{nulldef}
\nabla_{\hat{Y}}\hat{Y}=\nabla_{\hat{Y}}\hat{V}=\nabla_{\hat{Y}} E_A=0,
\end{equation}
where $A\in\{1,2\}$.
Then extend this null frame using the isomorphisms generated by $V$. We notice that the above equations hold everywhere. If we choose $r^-_Y$ close enough to $r_+$, we still have, by Taylor's Theorem, \begin{equation}\label{definekappa}\nabla_{\hat{V}} \hat{V}=\kappa \hat{V}+b^Y\hat{Y}+b^1E_1+b^2E_2,\end{equation}
where, in $r_+\leq r\leq r^-_Y$, $\kappa$ is a strictly positive function bounded away from $0$ and $|b^\alpha|\leq C\left(1-\mu \right)$. 

In Schwarzschild spacetime, consider the frame on $\mathcal S^2$ given by $\{\frac{1}{r_S^2\sin\theta_S}\partial_{\phi_S}, \frac{1}{r_S}\partial_{\theta_S}\}$. Then we get
$$\hat{V}=(1+\mu )\partial_{t^*_S}+(1-\mu )\partial_{r_S}, \hat{Y}=\partial_{t^*_S}-\partial_{r_S}, E_1=\frac{1}{r_S}\partial_{\theta_S}, E_2=\frac{1}{r_S\sin\theta_S}\partial_{\phi_S}$$

Since we consider Kerr spacetimes on which the metric is close to that on a Schwarzschild spacetime, we have that, in $(t^*,r,\theta,\phi^*)$ coordinates, the null frame can be expressed as
$$\hat{V}=(1+\mu )\partial_{t^*}+(1-\mu )\partial_r+O_1(\epsilon)\partial,$$
$$\hat{Y}=\partial_{t^*}-\partial_r+O_1(\epsilon)\partial,$$
$$E_1=\frac{1}{r}\partial_\theta+O_1(\epsilon)\partial,$$
$$E_2=\frac{1}{r\sin\theta}\partial_{\phi^*}+O_1(\epsilon)\partial.$$
Alternatively, if we write $E_\alpha$, where $\alpha=1,2,3,4$, for the null frame, we have

$$(1+\mu )\partial_{t^*}+(1-\mu )\partial_r=\hat{V}+O_1(\epsilon)E_\alpha,$$
$$\partial_{t^*}-\partial_r=\hat{Y}+O_1(\epsilon)E_\alpha,$$
$$\partial_\theta=rE_1+O_1(\epsilon)E_\alpha,$$
$$\partial_{\phi^*}=r\sin\theta E_2+O_1(\epsilon)E_\alpha.$$
We also define the vector fields $\hat{V}, \hat{Y}, E_1, E_2$ outside $\{r\leq r^-_Y\}$ by requiring them to be compactly supported in $\{r\leq r^+_Y\}$ (for some $r^+_Y$ to be determined) and invariant under the one parameter families of isometries generated by $\partial_{t^*}$ and $\partial_{\phi^*}$. Notice that in particular there is no requirement that the vector fields form a null frame in the region $\{r^-_Y\leq r \leq r^+_Y\}$. 

\section{Notations}
\subsection{Constants}
Throughout this paper, we will use $C$ to denote a large constant and $c$ to denote a small constant. They can be different from line to line. We will also use $A$ to denote bootstrap constants and we think of $A$ to be large, i.e., $A \gg C$. We also use the notation $O_i(1)$ and $O_i(\epsilon)$ to denote terms that are bounded up to a constant by $1$ and $\epsilon$, with bounds that improve by $r^{-1}$ for each derivative up to the $i$-th derivative. We will also use the notation $f\sim g$ to denote $cf\leq g\leq Cf$.

There are some constants that we will choose in the proof. The following are values of $r$ in the Kerr coordinates:
$$r_+<r^-_Y<r^+_Y<\frac{11M}{4}<R_\Omega.$$
$r^-_Y$ and $r^+_Y$ will be fixed in Remarks \ref{r-fix} and \ref{fixr+e} respectively.

There are also smallness parameters which can be thought of as obeying
$$0<\delta<\epsilon\ll\eta\ll e.$$
$\epsilon$ will always be used to denote the smallness of the specific angular momentum $a$ of the spacetimes that we are working on. $\eta\sim C\epsilon$ will denote the loss in the decay rate of the solutions to the wave equation as compared to that on Schwarzschild spacetimes. $e$ will be used to construct the non-degenerate energy. $\delta$ and $\delta'$ will simply be used as a small parameter whenever they are needed in the analysis. $\delta$ and $\delta'$ need not be fixed from line to line.

\subsection{Values of $t^*$}
We will adopt the following as much as possible:
$t^*$ will denote a general value of $t^*$. In particular, it will be used as integration variables. $\tau$ will denote the $t^*$ value for which we want an estimate. $\tau_0$ will denote the $t^*$ value where the initial data is posed. We will always assume $\tau_0 \geq 1$ and the reader can think of $\tau_0=1$. When integrating, we will often denote the endpoints by $\tau'$ and $\tau$. Finally, at a few places we will need to choose a particular value of $t^*$ in an interval. This is usually done to achieve the max or min of the energy quantities. Such choices will often be denoted as $\tilde{\tau}$.

\subsection{Integration}\label{integration}

\begin{definition}Define the following sets:
\begin{enumerate}
 \item $\Sigma_{\tau}=\{t^*=\tau\}$
 \item $\mathcal R(\tau',\tau)=\{\tau'\leq t^*\leq\tau\}$
 \item $\mathcal H(\tau',\tau)=\{r=r_+,\tau'\leq t^*\leq\tau\}$
\end{enumerate}
\end{definition}
When integrating on these sets, we will normally integrate with respect to the volume form which we suppress. On $\Sigma_{\tau}$ the volume form is $\sqrt{\det g_K|_{\Sigma_\tau}}.$ On $\mathcal R(\tau',\tau)$, the volume form is $\sqrt{\-det g_K}.$ However, on the event horizon $\mathcal H(\tau',\tau)$, the surface is null and the metric is degenerate. Nevertheless, on $\mathcal H(\tau',\tau)$, the integrand will always be of the form $J_\mu n^\mu_{\Sigma_{\mathcal H^+}}$, where $n^\mu_{\Sigma_{\mathcal H^+}}$ is the normal to $\mathcal H(\tau',\tau)$. We will hence take the volume form corresponding to the (arbitrarily) chosen normal. Occasionally, we will also integrate over the topological 2-spheres given by fixing $t$ and $r$. We will denote the volume form by $dA=\sqrt{\det g_K|_{\mathbb S^2}}.$

For some computations, however, it is more convenient to write down the volume form explicitly in coordinates. In our notations, the following two expressions denote the same integral:

$$\int_{\Sigma_\tau} f = \int_{\Sigma_\tau} f \sqrt{\det g_K|_{\Sigma_\tau}} dr d\theta d\phi.$$

When we write the integrals, we will often use $\iint$ to denote an integral over a spacetime region and use $\int$ denote an integral over a spacelike or null hypersurface. 

Notice that the volume form on $\Sigma_{t^*}$ can be compared with that on $\mathcal R(\tau',\tau)$. In particular, we have
$$\iint_{\mathcal R(\tau',\tau)} f \sim \int_{\tau'}^{\tau}\left(\int_{\Sigma_{t^*}} f\right)dt^*.$$

\section{Vector Field Commutators}\label{commutators}
In this section, we discuss the vector field commutators that we will use in this article. One obvious such vector field is the Killing vector field $\partial_{t^*}$ which satisfies
$$[\Box_{g_K},\partial_{t^*}]=0.$$
As we will see, this is not sufficient to control all the higher derivatives of the solution $\Phi$ near the event horizon. We will follow \cite{DRK} and \cite{DRL} to define the commutator $\hat Y$. The main innovation in this article is to define the vector field $S$ that would be used as a commutator to prove the improved decay rate for the solution $\Phi$. We would also need the vector fields $\Omega_i$ to control the error terms coming from the commutator $[\Box_{g_K},S]$.
\subsection{Vector Field Commutators under Metric Perturbations}
Some computations are easier in Schwarzschild spacetime than in Kerr spacetime. In the sequel, we will often consider consider a fixed vector field on the differentiable structure of the Schwarzschild exterior. We now show that for such vector fields, the commutators with $\Box_{g_S}$ and $\Box_{g_K}$ are close to each other as long as $a$ is chosen to be sufficiently small:
\begin{proposition}\label{commsta}
Consider either the Schwarzschild $(t^*_S,r_S,x^A_S,x^B_S)$ coordinates or $(t_S, r_S \geq r^-_Y, x^A_S,x^B_S)$ coordinates. Suppose $V$ is a vector field defined on either of these coordinates. Then $$|[\Box_{g_K}-\Box_{g_S},V]\Phi|\leq C\epsilon r^{-2}(\displaystyle\sum_{m=1}^{2} \sum_{k=1}^2\max_\alpha|\partial^m V^\alpha||\partial^k\Phi|),$$ where $\partial$ is the coordinate derivative for the coordinate system on which $V$ is defined.
\end{proposition}
\begin{proof}
We rewrite
$$\Box_{g_S}=g_S^{\alpha \beta}\partial_\alpha\partial_\beta+\eta_S^\alpha\partial_\alpha, $$
$$\Box_{g_K}=g_K^{\alpha \beta}\partial_\alpha\partial_\beta+\eta_K^\alpha\partial_\alpha. $$
Using $|(g_K)_{\alpha\beta}-(g_S)_{\alpha\beta}|\leq \epsilon r^{-2}$ and $|\partial_\gamma((g_K)_{\alpha\beta}-(g_S)_{\alpha\beta})|\leq \epsilon r^{-2}$, we have $|\sqrt{-\det g_K}-\sqrt{-\det g_S}|\leq \epsilon r^{-2}$ and $|\partial_\alpha(\sqrt{-\det g_K}-\sqrt{-\det g_S})|\leq \epsilon r^{-2}$.
Therefore,
$$\sup_{\alpha,\beta}|g_S^{\alpha \beta}-g_K^{\alpha \beta}|+\sup_\alpha|\eta_S^\alpha-\eta_K^\alpha|\leq C\epsilon r^{-2}.$$
Therefore,
\begin{equation*}
\begin{split}
|[\Box_{g_K}-\Box_{g_S},V]\Phi|\leq &|(g_K^{\alpha \beta}-g_S^{\alpha \beta})(\partial_\alpha\partial_\beta V^\gamma)\partial_\gamma\Phi|+2|(g_K^{\alpha \beta}-g_S^{\alpha \beta})\partial_\alpha V^\gamma \partial_\beta\partial_\gamma\Phi|+|(\eta_S^\alpha-\eta_K^\alpha)(\partial_\alpha V^\gamma)\partial_\gamma\Phi|\\
\leq& C\epsilon r^{-2}(\displaystyle\sum_{m=1}^{2} \sum_{k=1}^2\sup_\alpha|\partial^m V^\alpha||\partial^k\Phi|).
\end{split}
\end{equation*}
\end{proof}

\subsection{Commutator S}\label{commS}
We construct a commuting vector field $S$ on Schwarzschild that is different from \cite{L} and is stable under perturbation.\\
Define $S=t^*_S\partial_{t^*_S}+h(r_S)\partial_{r_S}$, where $h(r)= \left\{\begin{array}{clcr}r-2M &r\sim 2M\\(r+2M\log(r-2M)-3M-2M \log M )(1-\mu )&r\ge R\end{array}\right.$, for some large $R$ and is interpolated so that it is smooth and non-negative. Notice that for $r\geq R$, since $t^*=t$, this agrees with the definition in \cite{L}. Therefore we have 
\begin{equation}\label{commSinf}
[\Box_{g_S},S]=\left(2+\frac{r^*\mu}{r}\right)\Box_{g_S} +\frac{2}{r}\left(\frac{r^*}{r}-1-\frac{2r^*\mu}{r}\right)\partial_{r^*}+2\left(\left(\frac{r^*}{r}-1\right)-\frac{3r^*\mu}{2r}\right)\lapp,
\end{equation}
where $\lapp$ is the Laplace-Beltrami operator on the standard sphere.
In the coordinates $(t^*, r,\theta,\phi)$, 
$$\Box_{g_S}=-\alpha_1\left(r\right) \partial_t^2+\alpha_2\left(r\right)\partial_r^2+\alpha_3\left(r\right)\partial_r\partial_t+\alpha_4\left(r\right)\partial_t+\alpha_5\left(r\right)\partial_r+\lapp.$$
The crucial observation is that all $\alpha_i$ are smooth and bounded and depend only on $r$. Noting that $\alpha_i$ does not depend on $t$, we have
$$[\Box_{g_S},S]=\beta_1\left(r\right)\partial_t^2+\beta_2\left(r\right)\partial_r^2+\beta_3\left(r\right)\partial_r\partial_t+\beta_4\left(r\right)\partial_t+\beta_5\left(r\right)\partial_r+\beta_6\left(r\right)\lapp.$$
Again, it is important to note that all $\beta_i$ are smooth, bounded and depend only on $r$.
The form of $\beta_i$ for $r\geq R$ is given by (\ref{commSinf}).

We consider the same vector field $S$ on Kerr. Using Proposition \ref{commsta}, and noting that $\partial^m S^\alpha$ is bounded for $m\geq 1$, we have that for $r>R$,
$$|[\Box_{g_K},S]\Phi-\left(2+\frac{r^*\mu}{r}\right)\Box_{g_K} \Phi -\frac{2}{r}\left(\frac{r^*}{r}-1-\frac{2r^*\mu}{r}\right)\partial_{r^*}\Phi-2\left(\left(\frac{r^*}{r}-1\right)-\frac{3r^*\mu}{2r}\right)\lapp\Phi|\leq C\epsilon r^{-2}(\sum_{k=1}^2|\partial^k\Phi|),$$
and that for $r\leq R$,
$$|[\Box_{g_K},S]\Phi|\leq C\sum_{k=1}^2|D^k\Phi|.$$

\subsection{Commutator $\tilde{\Omega}_i$}\label{commOmega}
Let $\Omega_i$ be a basis of vector fields of rotations in Schwarzschild spacetimes. An explicit realization can be $\Omega=\partial_\phi, \sin\phi\partial_\theta+ \frac{\cos\phi\cos\theta}{\sin\theta}\partial_\phi, \cos\phi\partial_\theta-\frac{\sin\phi\cos\theta}{\sin\theta}\partial_{\phi}$. Define $\tilde{\Omega}_i=\chi(r)\Omega_i$ to be cutoff so that it is supported in $\{r>R_\Omega\}$ and equals $\Omega_i$ for $r>R_\Omega+1$ for some large $R$. On Schwarzschild spacetimes, $\Omega_i$ is Killing and therefore $\tilde{\Omega}_i$ is Killing for $r>R+1$. Therefore,
$$[\Box_{g_S}, \tilde{\Omega}_i]= \tilde{\chi}\left(r\right)\left(\partial^2+\partial\right),$$
where $\tilde\chi$ is some function depending only on $r$ and is supported in $\{R_\Omega<r<R_\Omega+1\}$.\\
Using Proposition \ref{commsta}, we have
$$|[\Box_{g_K}, \tilde{\Omega}_i]\Phi|\leq Cr^{-2}\left(|\partial^2\Phi|+|\partial\Phi|\right).$$
Moreover, since $\tilde{\Omega}_i$ vanishes for $r<R_\Omega$, we have trivially,
$$[\Box_{g_K}, \tilde{\Omega}_i]\Phi=0, \quad \mbox{for } r<R_\Omega.$$ 
From now on, we will write $\tilde{\Omega}$ to denote any one of the $\Omega_i$, while taking the norm to be $|\tilde{\Omega}\Phi|=\displaystyle\sum_i |\Omega_i\Phi|$. We would like to point out that this commutator is useful to gain powers of $r$ near spatial infinity. In particular we have
$$|\nabb\Phi|\leq \frac{C}{r}|\tilde{\Omega}\Phi|.$$
This extra power of $r$ is essential to control the error terms arising from the commutation of $\Box_{g_K}$ with $S$.

\subsection{Commutator $\hat{Y}$}
Let $\hat{Y}$, as in Section \ref{geometry}.3, be a vector field that is null near the event horizon, normalized with respect to another null vector $\hat{V}$ and is cut off to be compactly supported in $\{r\leq r^+_Y\}$.
\begin{proposition}
On Kerr spacetimes such that $\epsilon$ is small, we have for $r\leq r^-_Y$, $$|[\Box_{g_K},\hat{Y}]\Phi-\kappa\hat{Y}^2\Phi|\leq C\left(|D\partial_{t^*}\Phi|+\epsilon|D^2\Phi|+|D\Phi|\right),$$
where $\kappa >c >0$ is as in (\ref{definekappa}).
\end{proposition}
\begin{proof}
The principal term for the commutator $[\Box_{g_K},\hat{Y}]\Phi$ is $2^{(\hat Y)}\pi^{\mu\nu}D_\mu D_\nu \Phi$, where $^{(\hat Y)}\pi_{\mu\nu}$ is the deformation tensor defined by $^{(\hat Y)}\pi_{\mu\nu}=\frac{1}{2}(D_\mu\hat Y_\nu+D_\nu\hat Y_\mu)$. We look at three terms that are useful in deriving the estimates.
\begin{equation*}
\begin{split}
^{(\hat{Y})}\pi_{\hat{V}\hat{V}}=g(D_{\hat{V}}\hat{Y},\hat{V})=-g(\hat{Y},D_{\hat{V}}\hat{V})=2\kappa,\\
\end{split}
\end{equation*}
\begin{equation*}
\begin{split}
|^{(\hat{Y})}\pi_{\hat{V}E_A}|=&|\frac{1}{2}\left(g(D_{\hat{V}}\hat{Y},E_A)+g(D_{E_A}\hat{Y},\hat{V})\right)|
\leq C\epsilon,\\
\end{split}
\end{equation*}
\begin{equation*}
\begin{split}
|^{(\hat{Y})}\pi_{E_A E_B}|=&|\frac{1}{2}\left(g(D_{E_B}\hat{Y},E_A)+g(D_{E_A}\hat{Y},E_B)\right)|
\leq C\epsilon,\\
\end{split}
\end{equation*}
where the smallness in the the second and third line come from the assumption that we are close to Schwarzschild. Notice also that for $r\leq r^-_Y$, $\hat V$ is $C^0$ close to $\partial_{t^*}$.
Therefore, in the commutator, the main term is
$$\kappa \hat Y^2\Phi.$$ 
All the other second order terms either have a $\partial_{t^*}$ derivative or small.
\end{proof}

\section{The Basic Identities for Currents}\label{currents}
\subsection{Vector Field Multipliers}
We consider the conservation laws for $\Phi$ satisfying $\Box_{g}\Phi=0$.
Define the energy-momentum tensor $$T_{\mu\nu}=\partial_{\mu}\Phi\partial_{\nu}\Phi-\frac{1}{2}g_{\mu\nu}\partial^\alpha\Phi\partial_\alpha\Phi.$$
We note that $T_{\mu\nu}$ is symmetric and the wave equation implies that $$D^{\mu}T_{\mu\nu}=0.$$
Given a vector field $V^\mu$, we define the associated currents 
$$J^V_{\mu}\left(\Phi\right)=V^{\nu}T_{\mu\nu}\left(\Phi\right),$$
$$K^V\left(\Phi\right)=^{(V)}\pi_{\mu\nu}T^{\mu\nu}\left(\Phi\right),$$
where $^{(V)}\pi_{\mu\nu}$ is the deformation tensor defined by
$$^{(V)}\pi_{\mu\nu}=\frac{1}{2}\left(D_{\mu}V_{\nu}+D_{\nu}V_{\mu}\right).$$
In particular, $K^V\left(\Phi\right)=^{(V)}\pi_{\mu\nu}=0$ if $V$ is Killing.
Since the energy-momentum tensor is divergence-free, 
$$D^{\mu}J^{V}_{\mu}\left(\Phi\right)=K^V\left(\Phi\right).$$
We also define the modified currents $$J^{V,w}_{\mu}\left(\Phi\right)=J^{V}_{\mu}\left(\Phi\right)+\frac{1}{8}\left(w\partial_{\mu}\Phi^2-\partial_{\mu}w\Phi^2\right),$$
$$K^{V,w}\left(\Phi\right)=K^V\left(\Phi\right)+\frac{1}{4}w\partial^{\nu}\Phi\partial_{\nu}\Phi-\frac{1}{8}\Box_g w\Phi^2.$$
Then $$D^{\mu}J^{V,w}_{\mu}\left(\Phi\right)=K^{V,w}\left(\Phi\right).$$
We integrate by parts with this in the region bounded by $\Sigma_\tau$, $\Sigma_{\tau'}$ and $\mathcal H^+(\tau',\tau)$. We denote this region as $\mathcal R(\tau',\tau)$. Denoting the future-directed normal to $\Sigma_\tau$ by $n_{\Sigma_\tau}^{\mu}$, we have
\begin{proposition}
$$\int_{\Sigma_\tau} J^{V}_{\mu}\left(\Phi\right)n_{\Sigma_\tau}^{\mu}+\int_{\mathcal H(\tau',\tau)} J^{V}_{\mu}\left(\Phi\right)n_{\mathcal H^+(\tau',\tau)}^{\mu}+\iint_{\mathcal R(\tau',\tau)} K^V\left(\Phi\right) =\int_{\Sigma_{\tau'}} J^{V}_{\mu}\left(\Phi\right)n_{\Sigma_{\tau'}}^{\mu}.$$
$$\int_{\Sigma_\tau} J^{V,w^V}_{\mu}\left(\Phi\right)n_{\Sigma_\tau}^{\mu}+\int_{\mathcal H^+(\tau',\tau)} J^{V,w^V}_{\mu}\left(\Phi\right)n_{\mathcal H^+(\tau',\tau)}^{\mu}+\iint_{\mathcal R(\tau',\tau)} K^{V,w^V}\left(\Phi\right) =\int_{\Sigma_{\tau'}} J^{V,w^V}_{\mu}\left(\Phi\right)n_{\Sigma_{\tau'}}^{\mu}.$$
\end{proposition}
One can similarly define the above quantities for the inhomogeneous wave equation $\Box_{g}\Phi=F$. In this case, the energy-momentum is no longer divergence free. Instead, we have
$$D^{\mu}T_{\mu\nu}=F\partial_\nu\Phi.$$
In this case,
$$D^{\mu}J^{V}_{\mu}\left(\Phi\right)=K^V\left(\Phi\right)+FV^{\nu}\partial_{\nu}\Phi.$$
For the modified current, $$D^{\mu}J^{V,w}_{\mu}\left(\Phi\right)=K^{V,w}\left(\Phi\right)+\frac{1}{4}Fw\Phi+FV^{\nu}\partial_{\nu}\Phi.$$
\begin{proposition}
\begin{equation*}
\begin{split}
&\int_{\Sigma_\tau} J^{V}_{\mu}\left(\Phi\right)n_{\Sigma_\tau}^{\mu}+\int_{\mathcal H(\tau',\tau)} J^{V}_{\mu}\left(\Phi\right)n_{\mathcal H^+(\tau',\tau)}^{\mu}+\iint_{\mathcal R(\tau',\tau)} K^V\left(\Phi\right)\\
=&\int_{\Sigma_{\tau'}} J^{V}_{\mu}\left(\Phi\right)n_{\Sigma_{\tau'}}^{\mu}-\iint_{\mathcal R(\tau',\tau)} FV^{\nu}\partial_{\nu}\Phi.
\end{split}
\end{equation*}
\begin{equation*}
\begin{split}
&\int_{\Sigma_\tau} J^{V,w^V}_{\mu}\left(\Phi\right)n_{\Sigma_\tau}^{\mu}+\int_{\mathcal H(\tau',\tau)} J^{V,w^V}_{\mu}\left(\Phi\right)n_{\mathcal H^+(\tau',\tau)}^{\mu}+\iint_{\mathcal R(\tau',\tau)} K^{V,w^V}\left(\Phi\right)\\
=&\int_{\Sigma_{\tau'}} J^{V,w^V}_{\mu}\left(\Phi\right)n_{\Sigma_{\tau'}}^{\mu}+\iint_{\mathcal R(\tau',\tau)} \left(-\frac{1}{4}Fw\Phi-FV^{\nu}\partial_{\nu}\Phi\right).
\end{split}
\end{equation*}
\end{proposition}

\subsection{Vector Field Multipliers under Metric Perturbations}
If we consider Kerr spacetimes such that $\epsilon$ is small, vector fields multipliers that are defined in the Schwarzschild $(t^*,r,x^A,x^B)$ coordinates or Schwarzschild $(t, r \geq r^-_Y, x^A,x^B)$ coordinates are stable. We can consider a fixed vector field defined on the differentiable structure of a Schwarzschild exterior and compare the currents obtained using the Schwarzschild metric and the Kerr metric. 
\begin{proposition}
Consider either the Schwarzschild $(t^*_S,r_S,x^A_S,x^B_S)$ coordinates or $(t_S, r_S \geq r^-_Y, x^A_S,x^B_S)$ coordinates. Suppose $V$ is a vector field defined on either of these coordinates.
$$|(J^{V,w^V}_S)_{\mu}\left(\Phi\right)n_{\Sigma_\tau}^{\mu}-(J^{V,w^V}_K)_{\mu}\left(\Phi\right)n_{\Sigma_\tau}^{\mu}|\leq C\epsilon r^{-2}\max_\alpha |V^\alpha| (\partial\Phi)^2,$$
$$|K^{V,w^V}_S\left(\Phi\right)-K^{V,w^V}_K\left(\Phi\right)|\leq C\epsilon r^{-2} \left( (\sum_{k=0}^1\max_{\alpha}|\partial^k V^\alpha|+|w|) (\partial\Phi)^2+\sum_{m=1}^2 |\partial^m w|\Phi^2\right)$$
\end{proposition}

\section{Statement of Main Theorem}\label{sectionmaintheorem}
With the currents defined, we can state our main theorem.
\begin{maintheorem}\label{mainmaintheorem}
Suppose $\Box_{g_K}\Phi=0$. Then for all $\eta>0$, $R>r_+$ and all $M>0$ there exists $a_0$ such that the following estimates hold in the region $\{r_+\leq r\leq R\}$ on Kerr spacetimes with $(M,a)$ for which $a\leq a_0$:
\begin{enumerate}
\item Improved Decay of Non-degenerate Energy
\begin{equation*}
\begin{split}
&\sum_{j=0}^\ell\int_{\Sigma_\tau\cap\{r\leq R\}} \left(D^j\Phi\right)^2\\
\leq &C_R\tau^{-3+\eta}\left(\sum_{m=0}^{\ell+2}\int_{\Sigma_{\tau_0}} J^{Z+CN,w^Z}_\mu\left(\partial_{t^*}^mS\Phi\right)n^\mu_{\Sigma_{{\tau_0}}}+\sum_{m+k+j\leq \ell+5}\int_{\Sigma_{\tau_0}} J^{Z+CN,w^Z}_\mu\left(\partial_{t^*}^m\hat{Y}^k\tilde{\Omega}^j\Phi\right)n^\mu_{\Sigma_{{\tau_0}}}\right).\\
\end{split}
\end{equation*}
\item Improved Pointwise Decay\\
\begin{equation*}
\begin{split}
&\sum_{j=0}^\ell|D^j\Phi|\\\leq &C_R\tau^{-\frac{3}{2}+\eta}\left(\sum_{m=0}^{\ell+4}\int_{\Sigma_{\tau_0}} J^{Z+CN,w^Z}_\mu\left(\partial_{t^*}^mS\Phi\right)n^\mu_{\Sigma_{{\tau_0}}}+\sum_{m+k+j\leq \ell+7}\int_{\Sigma_{\tau_0}} J^{Z+CN,w^Z}_\mu\left(\partial_{t^*}^m\hat{Y^k}\tilde{\Omega}^j\Phi\right)n^\mu_{\Sigma_{{\tau_0}}}\right)^{\frac{1}{2}}.\\
\end{split}
\end{equation*}
Here the vector field $N$ will be defined in Section \ref{sectionN}; the vector field $Z$ with the modifying function $w^Z$ will be defined in Section \ref{sectionZ}.
\end{enumerate}

\end{maintheorem}
\begin{remark}
We will show that although $J^{Z,w^Z}_\mu\left(\Phi\right)n^\mu_{\Sigma_{{t^*}}}$ is not always nonnegative, $J^{Z+CN,w^Z}_\mu\left(\Phi\right)n^\mu_{\Sigma_{{t^*}}}$ is nonnegative for sufficiently large $C$. Hence all the energy quantities in the Theorem are nonnegative.
\end{remark}

\begin{remark}
Since we have the improved decay of the non-degenerate energy, the above Theorem can be extended beyond the event horizon. More precisely, for any $r_b\in (r_-, r_+)$, where $r_-$ is the smaller root of $\Delta=r^2-2Mr-a^2$, the theorem holds up to $r \geq r_b$ for $D$ understood as a regular derivative inside the black hole, and with the constant depending also on $r_b$. The proof would be similar to that in \cite{L}. 
\end{remark}

\section{Vector Field Multiplier $N_e$ and Mild Growth of Non-degenerate Energy}\label{sectionN}
Kerr spacetime has a Killing vector field $\partial_t$. The conservation law gives that
$$\int_{\Sigma_\tau} J^T_\mu\left(\Phi\right)n^\mu_{\Sigma_\tau} +\int_{\mathcal H(\tau_0,\tau)} J^T_\mu\left(\Phi\right)n^\mu_{\Sigma_\tau} =\int_{\Sigma_{\tau_0}} J^T_\mu\left(\Phi\right)n^\mu_{\Sigma_{\tau_0}} +\iint_{\mathcal R(\tau_0,\tau)} \partial_t\Phi G.$$
We add to the Killing vector field $\partial_t$ a red-shift vector field. Here, we use the ``non-regular'' red-shift vector field as in \cite{DRL}. Under this construction, $N_e$ is $C^0$ but not $C^1$ at the event horizon $\mathcal H^+$. Compared to the smooth construction in \cite{DRK}, this construction would provide extra control for some derivatives near $\mathcal H^+$.

Define $$Y=y_1\left(r\right)\hat{Y}+y_2\left(r\right)\hat{V},$$
where
$$y_1\left(r\right)=1-\frac{1}{(\log(r-r_+))^3},$$
$$y_2\left(r\right)=-\frac{1}{(\log(r-r_+))^3}.$$
Notice that by this definition $Y$ is compactly supported in $\{r\leq r^+_Y\}$ and is invariant under the isomorphisms generated by $\partial_{t^*}$ and $\partial_{\phi^*}$.
\begin{proposition}
Let $N_e=\partial_{t^*}+eY$. For any $e$, there is a corresponding choice of $\epsilon\ll e$ and $r^-_Y$ such that for every integer $p$, there exists $c_p>0$ such that
\begin{equation*}
J^{N_e}_\mu\left(\Phi\right)n^\mu_{\mathcal H^+}\sim \left(D_{\hat V} \Phi\right)^2+e\sum_{E_A\in\{E_1,E_2\}}\left(D_{E_{A}} \Phi\right)^2,\quad\mbox{on the event horizon },
\end{equation*}
\begin{equation*}
J^{N_e}_\mu\left(\Phi\right)n^\mu_{\Sigma_\tau}\sim \sum_{E_\alpha\in\{E_1,E_2,\hat V\}}\left(D_{E_\alpha} \Phi\right)^2+e\left(D_{E_{\hat Y}} \Phi\right)^2,\quad\mbox{for }r\leq r^-_Y,
\end{equation*}
\begin{equation*}
J^{N_e}_\mu\left(\Phi\right)n^\mu_{\Sigma_\tau}\sim \sum\left(\partial \Phi\right)^2,\quad\mbox{for }r\geq r^-_Y\mbox{ in the }(t^*,r,x^A,x^B)\mbox{ coordinates},
\end{equation*}
\begin{equation*}
K^{N_e}\left(\Phi\right)\geq c_p e\left(|\log(r-r_+)|^p\left(\left(D_{\hat{V}} \Phi\right)^2+\sum_A\left(D_{E_A}\Phi\right)^2\right)+\left(D_{\hat{Y}}\Phi\right)^2\right), \quad \mbox{for }r\leq r^-_Y,
\end{equation*}
\begin{equation*}
K^{N_e}\left(\Phi\right) \leq C e J^{N_e}_\mu\left(\Phi\right)n^\mu_{\Sigma_\tau}, \quad \mbox{for }r^-_Y\leq r \leq r^+_Y.
\end{equation*}
\end{proposition}
\begin{proof}
It is obvious that  $Y$ is timelike and future oriented for $r\leq r^-_Y$. Since $\partial_{t^*}$ is casual in the exterior region of Schwarzschild spacetime and is null only on the event horizon, for every small $e>0$, there exists sufficiently small $\epsilon>0$ such that $N_e$ is timelike and future directed on Kerr spacetimes up to the event horizon. The first two estimates hold since in Kerr spacetime, $\partial_{t^*}$ is $\epsilon$-close to $\hat V$ on the event horizon. The third estimate holds because outside a small (depending on $\epsilon$) neighborhood of the event horizon, $\partial_{t^*}$ is timelike.

To show that $K^{N_e}\left(\Phi\right)$ has the required positivity near the event horizon, we compute the deformation tensor. First, notice that
$$D_{\hat Y} y_1=D_{\hat Y} y_2=\frac{3 D_{\hat Y}r}{(r-r_+)(\log (r-r_+))^4}.$$
Using this we have
$$^{(Y)}\pi_{\hat V \hat V}=g_K(D_{\hat V}(y_1 \hat Y+y_2 \hat V),\hat V)=-g_K(y_1\hat Y+y_2\hat V,D_{\hat V}\hat V)=2y_1\kappa+b^Y y_2,$$
$$^{(Y)}\pi_{\hat Y \hat Y}=g_K(D_{\hat Y}(y_1 \hat Y+y_2 \hat V),\hat Y)=-\frac{6 D_{\hat Y}r}{(r-r_+)(\log (r-r_+))^4},$$
$$^{(Y)}\pi_{\hat V \hat Y}=\frac{1}{2}g_K(D_{\hat V}(y_1 \hat Y+y_2 \hat V),\hat Y)+\frac{1}{2}g_K(D_{\hat Y}(y_1 \hat Y+y_2 \hat V),\hat V)=-\frac{3 D_{\hat Y}r}{(r-r_+)(\log (r-r_+))^4}+y_1\kappa+y_2 b^Y,$$
Moreover, we have
$$^{(Y)}\pi_{\hat V E_A}, ^{(Y)}\pi_{\hat Y E_A}, ^{(Y)}\pi_{E_A E_B}=O(1).$$
Notice that 
$$T_{\hat Y \hat Y}\sim (D_{\hat Y}\Phi)^2, T_{\hat V \hat V}\sim (D_{\hat V}\Phi)^2, T_{\hat Y \hat V}\sim |\nabb\Phi|^2,$$
and that $^{(Y)}\pi_{\hat V E_A}, ^{(Y)}\pi_{\hat Y E_A}, ^{(Y)}\pi_{E_A E_B}$ have no terms of the form $(D_{\hat Y}\Phi)^2$.
Hence we can choose $r^-_Y$ sufficiently close to $r_+$ such that for $r_+\leq r\leq r^-_Y$,
\begin{equation*}
\begin{split}
K^{Y}\left(\Phi\right)
\geq&c\kappa \left(D_{\hat{Y}}\Phi\right)^2+\frac{c}{(r-r_+)|\log(r-r_+)|^4}\left(\left(D_{\hat{V}} \Phi\right)^2+\sum_A\left(D_{E_A}\Phi\right)^2\right).
\end{split}
\end{equation*}
Since $\partial_{t^*}$ is Killing, $K^{N_e}\left(\Phi\right)=eK^{Y}\left(\Phi\right)$, we have 
\begin{equation*}
K^{N_e}\left(\Phi\right)\geq ce\left(\kappa \left(D_{\hat{Y}}\Phi\right)^2+\frac{1}{(r-r_+)|\log(r-r_+)|^4}\left(\left(D_{\hat{V}} \Phi\right)^2+\sum_A\left(D_{E_A}\Phi\right)^2\right)\right), \quad \mbox{for }r\leq r^-_Y,
\end{equation*}
Finally, since in the region $r^-_Y\leq r \leq r^+_Y$, $J^{\partial_{t^*}}_\mu n^\mu_{\Sigma_{t^*}}$ controls all derivatives, we have
\begin{equation*}
K^{N_e}\left(\Phi\right) \leq C e J^{N_e}_\mu\left(\Phi\right)n^\mu_{\Sigma_\tau}, \quad \mbox{for }r^-_Y\leq r \leq r^+_Y.
\end{equation*}
\end{proof}
\begin{definition}We call the positive quantity
$\int_{\Sigma_\tau} J^{N_e}_\mu\left(\Phi\right)n^\mu_{\Sigma_\tau}$ the non-degenerate energy.
\end{definition}
The following identity determines how the non-degenerate energy changes with $\tau$.

\begin{proposition}\label{Nid}
Let $\Phi$ satisfy $\Box_{g_K}\Phi=G$. Then
\begin{equation*}
\begin{split}
&\int_{\Sigma_\tau} J^{N_e}_\mu\left(\Phi\right)n^\mu_{\Sigma_\tau} +\int_{\mathcal H(\tau_0,\tau)} J^{N_e}_\mu\left(\Phi\right)n^\mu_{\Sigma_\tau} +\iint_{\mathcal R(\tau_0,\tau)\cap\{r\leq r^-_Y\}}K^{N_e}_\mu\left(\Phi\right)\\
=&\int_{\Sigma_{\tau_0}} J^{N_e}_\mu\left(\Phi\right)n^\mu_{\Sigma_{\tau_0}}+e\iint_{\mathcal R(\tau_0,\tau)\cap\{r^-_Y\leq r\leq r^+_Y\}} K^Y\left(\Phi\right)+\int_{\mathcal R(\tau_0,\tau)} \left(\partial_{t^*}\Phi+eY\Phi\right) G.
\end{split}
\end{equation*}
\end{proposition}
The estimates given by the vector field $N$ is sufficient to show that modulo inhomogeneous terms, the quantity $\int_{\Sigma_\tau} J^{N_e}_\mu\left(\Phi\right)n^\mu_{\Sigma_\tau}$ cannot grow too much in a short time interval:
\begin{proposition}\label{mg}
Let $\Phi$ satisfy $\Box_{g_K}\Phi=G$. For $e$ sufficiently small, $\epsilon \ll e$ and $0\leq\tau-\tau'\leq1$, we have 
$$\int_{\Sigma_{\tau}} J^{N_e}_\mu\left(\Phi\right)n^\mu_{\Sigma_{\tau}}+\int_{\mathcal H(\tau',\tau)}J^{N_e}_\mu\left(\Phi\right)n^\mu_{\mathcal H^+}\leq 4\int_{\Sigma_{\tau'}} J^{N_e}_\mu\left(\Phi\right)n^\mu_{\Sigma_{\tau'}}+C\iint_{\mathcal R(\tau',\tau)} G^2.$$
\end{proposition}
\begin{proof}
We first note that $\iint_{\mathcal R(\tau',\tau)\cap\{r^-_Y\leq r\leq r^+_Y\}} K^{Y}\left(\Phi\right)\leq C\int_{\tau'}^{\tau}\int_{\Sigma_{\bar\tau}} J^{N_e}_\mu\left(\Phi\right)n^\mu_{\Sigma_{\bar\tau}}d{\bar\tau}$, with $C$ independent of $e$ and $\epsilon$ whenever $\epsilon\ll e<1$.
Then, by Proposition \ref{Nid},
\begin{equation*}
\begin{split}
&\int_{\Sigma_{\tau}} J^{N_e}_\mu\left(\Phi\right)n^\mu_{\Sigma_{\tau}} +\int_{\mathcal H(\tau',\tau)} J^{N_e}_\mu\left(\Phi\right)n^\mu_{\mathcal H^+} +\iint_{\mathcal R(\tau',\tau)\cap\{r\leq r^-_Y\}}K^{N_e}_\mu\left(\Phi\right)\\
=&\int_{\Sigma_{\tau'}} J^{N_e}_\mu\left(\Phi\right)n^\mu_{\Sigma_{\tau'}}+e\iint_{\mathcal R(\tau',\tau)\cap\{r^-_Y\leq r\leq r^+_Y\}} K^Y\left(\Phi\right)+\int_{\mathcal R(\tau',\tau)} \left(\partial_{t^*}\Phi+eY\Phi\right) G\\
\leq &\int_{\Sigma_{\tau'}} J^{N_e}_\mu\left(\Phi\right)n^\mu_{\Sigma_{\tau'}}+Ce\int_{\tau'}^{\tau}\int_{\Sigma_{\bar\tau}} J^{N_e}_\mu\left(\Phi\right)n^\mu_{\Sigma_{\bar\tau}}d{\bar\tau}+\delta'\iint_{\mathcal R(\tau',\tau)} \left(\left(\partial_{t^*}\Phi+eY\Phi\right)\right)^2 +(\delta')^{-1}\iint_{\mathcal R(\tau',\tau)}G^2\\
\leq &\int_{\Sigma_{\tau'}} J^{N_e}_\mu\left(\Phi\right)n^\mu_{\Sigma_{\tau'}}+\left(C\delta'+2Ce\right)\int_{\tau'}^{\tau}\int_{\Sigma_{\bar\tau}} J^{N_e}_\mu\left(\Phi\right)n^\mu_{\Sigma_{\bar\tau}}d{\bar\tau}+(\delta')^{-1}\iint_{\mathcal R(\tau',\tau)} G^2.
\end{split}
\end{equation*}
By Gronwall's inequality and absorbing $(\delta')^{-1}$ into the constant $C$, we have
$$\int_{\Sigma_{\tau}} J^{N_e}_\mu\left(\Phi\right)n^\mu_{\Sigma_{\tau}}\leq 2\int_{\Sigma_{\tau'}} J^{N_e}_\mu\left(\Phi\right)n^\mu_{\Sigma_{\tau'}}+C\iint_{\mathcal R(\tau',\tau)} G^2.$$
Now the estimate for the term horizon follows from Proposition \ref{Nid}.
\end{proof}
\section{Integrated Decay Estimates and Boundedness of Non-degenerate Energy}\label{sectionX}
In this section we would like to show an integrated decay estimate. We first follow \cite{L} to construct a vector field and prove an integrated decay estimates for the terms near spatial infinity. That construction is in turn inspired by \cite{StR}. In \cite{L}, the decay rate in $r$ of this integrated decay estimate is crucial to control the error terms arising from the vector field commutator $S$. In the sequel, such estimate would also facilitate many computations as we prove the full integrated decay estimate. 

In view of the red-shift, all derivatives of $\Phi$ can be controlled near the event horizon. However, we would also like to prove an integrated decay estimate that controls $\Phi$ itself near the event horizon. This is in contrast to the integrated decay estimate in \cite{DRL} which degenerates near the event horizon. This extra control is useful as we are considering the inhomogeneous problem. 

The proof of the integrated decay estimate for a finite region of $r$ away from the horizon follows that in \cite{DRL}. We would like to remark that one difference here is that we do not assume the boundedness of $\int_{\Sigma_{\tau}} J^{N_e}_\mu\left(\Phi\right)n^\mu_{\Sigma_{\tau}}$ (even after ignoring inhomogeneous terms). We would instead like to prove the boundedness of $\int_{\Sigma_{\tau}} J^{N_e}_\mu\left(\Phi\right)n^\mu_{\Sigma_{\tau}}$ using the integrated decay estimates. We will, however, use Proposition \ref{mg}. 

The reader should think of this integrated decay estimates as analogous to the estimates associated to the vector field $X$ in \cite{DRS},\cite{DRK},\cite{L}. However, it is impossible to obtain such estimates using a vector field in Kerr spacetimes and we therefore resort to a phase space analysis (see \cite{A}).

To perform the phase space analysis, we will take the Fourier transform in the variable $t^*$, take the Fourier series in the variable $\phi^*$ and express the dependence on the $\theta$ variable in oblate spheroidal harmonics. Carter \cite{C} discovered that with this decomposition, the wave equation can be separated. However, in order to take the Fourier transform in the variable $t^*$, we need $\Phi$ to be at least in $L^2$. To this end, we perform a cutoff in the variable $t^*$. 

\subsection{Estimates near Spatial Infinity}
In this subsection, we follow \cite{L} to construct a vector field $\tilde{X}=\tilde{f}\left(r^*\right)\partial_{r^*}$ so that the spacetime integral that can be controlled with a good weight in $r$. We refer the reader to (\cite{L}, Proposition 8) for the following:
\begin{proposition}
In Schwarzschild spacetimes, using $(t,r^*,x^A,x^B)$ coordinates, there exists $\tilde{X}_S=\tilde{f}\left(r^*\right)\partial_{r^*}$ and $w_S^{\tilde{X}}$ supported in $r\geq \frac{13M}{4}$ such that
$$K^{\tilde{X},w^{\tilde{X}}}\left(\Phi \right)\geq c\left(r^{-1-\delta}\left(\partial_{r^*}\Phi\right)^2+r^{-1}|\nabb\Phi|^2+r^{-3-\delta}\Phi^2\right),$$
for $r^*\ge \max\{100,100M\}$ and 
$$|\int_{\Sigma_\tau} J^{\tilde{X},w^{\tilde{X}}}_{\mu}\left(\Phi\right)n^\mu_{\Sigma_\tau}|\leq C\int_{\Sigma_\tau}J^{N_e}_{\mu}\left(\Phi\right)n^\mu_{\Sigma_\tau}.$$
\end{proposition}
This implies via stability (since the vector field is supported away from the event horizon) the following:
\begin{proposition}
In Kerr spacetimes, using $(t^*,r,x^A,x^B)$ coordinates, there exists $\tilde{X}$ and $w^{\tilde{X}}$ supported in $r\geq \frac{25M}{8}$ such that for some large $R$,
$$K^{\tilde{X},w^{\tilde{X}}}\left(\Phi \right)\geq c_{\tilde{X}}\left(r^{-1-\delta}\left(\partial_{r^*}\Phi\right)^2+r^{-1}|\nabb\Phi|^2+r^{-3-\delta}\Phi^2\right)-C_{\tilde{X}}\epsilon r^{-2}\left(\partial_{t^*}\Phi\right)^2,$$
for $r^*\ge R$ and 
$$|\int_{\Sigma_\tau}J^{\tilde{X},w^{\tilde{X}}}_{\mu}\left(\Phi\right)n^\mu_{\Sigma_\tau}|\leq C\int_{\Sigma_\tau}J^{N_e}_{\mu}\left(\Phi\right)n^\mu_{\Sigma_\tau}.$$
\end{proposition}
Now it is easy to construct the following vector field on Schwarzschild spacetimes:
\begin{proposition}
In Schwarzschild spacetimes, using $(t,r^*,x^A,x^B)$ coordinates, there exists $\tilde{\tilde{X}}_S=\tilde{\tilde{f}}\left(r^*\right)\partial_{r^*}$ supported in $r\geq \frac{13M}{4}$ such that
$$K^{\tilde{\tilde{X}}}\left(\Phi \right)\geq cr^{-1-\delta}\left(\partial_{t^*}\Phi\right)^2- C\left(r^{-1-\delta}\left(\partial_{r^*}\Phi\right)^2+r^{-1}|\nabb\Phi|^2+r^{-3-\delta}\Phi^2\right),$$
for $r^*\ge \max\{100,100M\}$ and 
$$|\int_{\Sigma_\tau}J^{\tilde{\tilde{X}}}_{\mu}\left(\Phi\right)n^\mu_{\Sigma_\tau}|\leq C\int_{\Sigma_\tau}J^{N_e}_{\mu}\left(\Phi\right)n^\mu_{\Sigma_\tau}.$$
\end{proposition}
\begin{proof}
Let $\tilde{\tilde{f}}$ be supported appropriately and $\tilde{\tilde{f}}\left(r^*\right)=\frac{1}{\left(1+r^*\right)^\delta}$ whenever $r^*$ is large.
\end{proof}
As before, a stability argument would give:
\begin{proposition}
In Kerr spacetimes, using $(t^*,r,x^A,x^B)$ coordinates, there exists $\tilde{\tilde{X}}$ supported in $r\geq \frac{25M}{8}$ such that for some large $R$,
$$K^{\tilde{\tilde{X}}}\left(\Phi \right)\geq cr^{-1-\delta}\left(\partial_{t^*}\Phi\right)^2- C_{\tilde{\tilde{X}}}\left(r^{-1-\delta}\left(\partial_{r^*}\Phi\right)^2+r^{-1}|\nabb\Phi|^2+r^{-3-\delta}\Phi^2\right),$$
for $r^*\ge R$ and
$$|\int_{\Sigma_\tau}J^{\tilde{\tilde{X}}}_{\mu}\left(\Phi\right)n^\mu_{\Sigma_\tau}|\leq C\int_{\Sigma_\tau}J^{N_e}_{\mu}\left(\Phi\right)n^\mu_{\Sigma_\tau}.$$ 
\end{proposition}
Now using the vector field $\tilde{X}+\frac{1}{2}\frac{c_{\tilde{X}}}{C_{\tilde{\tilde{X}}}}\tilde{\tilde{X}}$ and modifying function $w^{\tilde{X}}$, we get the following estimate for $\epsilon$ sufficiently small:
\begin{proposition}\label{Xinf}
\begin{equation*}
\begin{split}
&\iint_{\mathcal R(\tau_0,\tau)\cap\{r\geq R\}}\left(r^{-1-\delta}J^{N_e}_\mu\left(\Phi\right)n^\mu_{\Sigma_{t^*}}+r^{-3-\delta}\Phi^2\right)\\
\leq&C\left(\int_{\Sigma_\tau} J^{N_e}_\mu\left(\Phi\right)n^\mu_{\Sigma_\tau}+\int_{\Sigma_{\tau_0}} J^N_\mu\left(\Phi\right)n^\mu_{\Sigma_{\tau_0}}+\iint_{\mathcal R(\tau_0,\tau)\cap\{\frac{25M}{8}\leq r\leq R\}}\left(J^{N_e}_\mu\left(\Phi\right)n^\mu_{\Sigma_{t^*}}+\Phi^2\right)\right.\\
&\left.+\iint_{\mathcal R(\tau_0,\tau)}\left(|\partial_r\Phi|+r^{-1}|\Phi|\right)|G|\right).
\end{split}
\end{equation*}
\end{proposition}

\subsection{Estimates near the Event Horizon}
The integrated decay estimates shown in \cite{DRL} is degenerate around the event horizon. Here we will prove the corresponding estimates near the event horizon. In view of the availability of the red-shift estimate $K^{N_e}$, we will focus on the zeroth order term $\Phi^2$. It turns out that we can use a construction in \cite{L}.
\begin{proposition}
In Schwarzschild spacetimes, using $(t,r^*,x^A,x^B)$ coordinates, there exists $X_h=f_h\left(r^*\right)\partial_{r^*}$ and $w^{X_h}$ supported in $r\leq \frac{23M}{8}$ such that
$$K^{X_h,w^{X_h}}\left(\Phi \right)\geq c\left((\partial_{r^*}\Phi)^2+|\nabb\Phi|^2+\Phi^2\right),$$
for $r\le r^-_Y$ and 
$$|\int_{\Sigma_\tau}J^{X_h,w^{X_h}}_{\mu}\left(\Phi\right)n^\mu_{\Sigma_\tau}|\leq C\int_{\Sigma_\tau}J^{N_e}_{\mu}\left(\Phi\right)n^\mu_{\Sigma_\tau}$$
and 
$$|J^{X_h,w^{X_h}}_{\mu}\left(\Phi\right)n^\mu_{\mathcal H^+}|\leq CJ^{N_e}_{\mu}\left(\Phi\right)n^\mu_{\mathcal H^+}.$$

\end{proposition}
\begin{proof}
Let $$X_h=f_h\left(r^*_S\right)\partial_{r^*_S}=-\chi(r)\frac{M^3}{\left(1+4\mu^{-2}\right)}\partial_{r^*_S}=-\chi(r)\frac{\mu^3r^3}{8\left(1+4\mu^{-2}\right)}\partial_{r^*_S},$$ where $\chi(r)$ is a cutoff function that is compactly supported in $r\leq \frac{23M}{8}$ and is identically $1$ for $r\leq r^-_Y$. Also, let
$$w^{X_h}=2f_h'\left(r^*\right)+\frac{4\left(1-\mu \right)}{r}f_h\left(r^*\right).$$
From now on, we will focus on the behavior when $ r\leq r^-_Y$ and treat the terms in $\{r^-_Y\leq r\leq\frac{23M}{8}\}$ as errors. 
Recall that on Schwarzschild spacetime:
\begin{equation*}
\begin{split}
K^{X_h,w^{X_h}}\left(\Phi\right)=&\frac{f'\left(r^*\right)}{1-\mu}\left(\partial_{r^*}\Phi \right)^2+\frac{\left(2-3\mu\right)f\left(r^*\right)}{2r} |\nabb\Phi|^2\\
&-\frac{1}{4}\left(\frac{1}{1-\mu}f'''\left(r^*\right)+\frac{4}{r}f''\left(r^*\right)+\frac{\mu}{r^2}f'\left(r^*\right)-\frac{2\mu}{r^3}\left(3-4\mu \right) f\left(r^*\right)\right)\Phi^2\\
\end{split}
\end{equation*}
We now look at the sign of this expression for $r\leq r^-_Y$. It is easy to see that the coefficient for $\left(\partial_{r^*}\Phi\right)^2$ is positive:
\begin{equation*}
\begin{split}
f'\left(r^*\right)=&\left(1-\mu \right)\partial_{r}f_0\left(r\right)\\
=&\frac{\mu r^2\left(1-\mu \right)}{\left(1+4\mu^{-2}\right)^2}\geq \frac{c(1-\mu )}{r^3},\\
\end{split}
\end{equation*}
The coefficient of $|\nabb\Phi|^2$ is also clearly positive.
A computation shows that
\begin{equation*}
\begin{split}
&\frac{1}{1-\mu}f'''+\frac{4}{r}f''+\frac{\mu}{r^2}f'-\frac{2\mu}{r^3}\left(3-4\mu\right)f\\
=&-\frac{\mu^6\left(192+\mu \left(128+\mu \left(-784+\mu \left(464+\mu \left(-28+\mu \left(52+\mu \left(-3+4\mu \right)\right)\right)\right)\right)\right)\right)}{4\left(4+\mu^{2}\right)^4}
\end{split}
\end{equation*}
We want to show that $P(\mu)=192+\mu \left(128+\mu \left(-784+\mu \left(464+\mu \left(-28+\mu \left(52+\mu \left(-3+4\mu \right)\right)\right)\right)\right)\right)\geq \frac{1}{7}$ for $\frac{16}{23}\leq \mu \leq 1$.\\
$192+128 \mu - 784 \mu^2 +464\mu^3 =16\left(-12-20 \mu +29 \mu^2 \right)\left(\mu -1\right)\geq 0.$\\
$52-3\mu +4 \mu^2$ reaches its minimum at $\frac{3}{8}$. Hence, $52-3\mu +4 \mu^2 \geq \frac{823}{16}.$\\
$-28+\mu \left(52-3\mu +4 \mu^2\right) \geq-28+\frac{11}{20}\frac{823}{16}\geq \frac{93}{320}.$\\
Therefore, $P(\mu)\geq \frac{1023}{6400}\geq\frac{1}{7}$ for $\frac{16}{23} \leq \mu \leq 1$.
Therefore, for $r\leq r^-_Y$,
$$K^{X_h,w^{X_h}}\left(\Phi\right)\geq c\left((\partial_{r^*}\Phi)^2+|\nabb\Phi|^2+\Phi^2\right).$$
The second statement
$$|\int_{\Sigma_\tau}J^{X_h,w^{X_h}}_{\mu}\left(\Phi\right)n^\mu_{\Sigma_\tau}|\leq C\int_{\Sigma_\tau}J^{N_e}_{\mu}\left(\Phi\right)n^\mu_{\Sigma_\tau}$$
and the third statement
$$|J^{X_h,w^{X_h}}_{\mu}\left(\Phi\right)n^\mu_{\mathcal H^+}|\leq CJ^{N_e}_{\mu}\left(\Phi\right)n^\mu_{\mathcal H^+}$$
follow from the boundedness of $f_h$ and $w^{X_h}$ and that on the Schwarzschild horizon $\partial_t=\partial_{r^*}$. Hence in both estimates, the constants are independent of $e$ for $e$ small.
\end{proof}
Noticing the $X_h$ and $w^{X_h}$ are actually smooth up to the event horizon, this implies via a stability argument:
\begin{proposition}
In Kerr spacetimes, using $(t_S,r_S,x^1_S,x^2_S)$ coordinates, there exists $X_h$ and $w^{X_h}$ supported in $r\leq \frac{23M}{8}$ such that
$$K^{X_h,w^{X_h}}\left(\Phi \right)\geq c\Phi^2-C\epsilon(\partial_{t^*}\Phi)^2-C\epsilon(\partial_r\Phi)^2,$$
for $r\le r^-_Y$ and 
$$|\int_{\Sigma_\tau}J^{X_h,w^{X_h}}_{\mu}\left(\Phi\right)n^\mu_{\Sigma_\tau}|\leq C\int_{\Sigma_\tau}J^N_{\mu}\left(\Phi\right)n^\mu_{\Sigma_\tau}$$
and
$$|J^{X_h,w^{X_h}}_{\mu}\left(\Phi\right)n^\mu_{\mathcal H^+}|\leq CJ^{N_e}_{\mu}\left(\Phi\right)n^\mu_{\mathcal H^+}.$$
\end{proposition}
Together with the red-shift, we therefore have the following integrated decay estimate near the event horizon. 
\begin{proposition}\label{Xhor}
\begin{equation*}
\begin{split}
&\iint_{\mathcal R(\tau_0,\tau)\cap\{r\leq r^-_Y\}}(\Phi^2+K^{N_e}\left(\Phi\right))\\
\leq&C\left(\int_{\Sigma_\tau} J^{N_e}_\mu\left(\Phi\right)n^\mu_{\Sigma_\tau}+\int_{\Sigma_{\tau_0}} J^{N_e}_\mu\left(\Phi\right)n^\mu_{\Sigma_{\tau_0}}+\iint_{\mathcal R(\tau_0,\tau)\cap\{r^-_Y\leq r\leq \frac{23M}{8}\}}(\Phi^2+J^{N_e}_\mu\left(\Phi\right)n^\mu_{\Sigma_{t^*}})\right.\\
&\left.+\iint_{\mathcal R(\tau_0,\tau)\cap\{r\leq \frac{23M}{8}\}}\left(|\partial_{r^*}\Phi|+r^{-1}|\Phi|\right)|G|+|\iint_{\mathcal R(\tau_0,\tau)\cap\{r\leq \frac{23M}{8}\}}(\partial_{t^*}\Phi+eY\Phi)G|\right).
\end{split}
\end{equation*}
\end{proposition}

\subsection{Cutoff, Decomposition and Separation}
Following \cite{DRL}, we define the cutoff
$$\Phi^\tau_{\tau'} =\xi\Phi, $$
where $$\xi=\chi(t^*-1-\tau)\chi(-t^*-1+\tau'),$$
for some smooth cutoff function $\chi(x)$ that is identically 1 for $x\leq -1$ and support on $\{x\leq 0\}$.
Then
$$\Box_g\Phi^\tau_{\tau'}=\xi G+2D^\alpha\Phi D_\alpha\xi+\Phi\Box_g\xi=:F.$$
We then decompose in frequency. We decompose the Fourier transform in $t$ of $\Phi$ into Fourier series in $\phi$ and oblate spheroidal harmonics:
$$\hat{\Phi}^\tau_{\tau'}=\displaystyle\sum_{m,\ell} R^\omega_{m\ell}(r)S_{m\ell}(a\omega,\cos\theta)e^{im\phi}.$$
We also decompose the inhomogeneous term $F$ (which comes both from the original inhomogeneous term $G$ and the cutoff):
$$\hat{F}=\displaystyle\sum_{m,\ell} F^\omega_{m\ell}(r)S_{m\ell}(a\omega,\cos\theta)e^{im\phi}.$$
Letting $\zeta$ be a sharp cutoff with such that $\zeta=1$ for $|x|\leq 1$ and $\zeta=0$ for $|x|>1$, we define 
$$\Phi_{\mbox{$\flat$}}=\int_{-\infty}^{\infty}\zeta(\frac{\omega}{\omega_1})\sum_{m,l:\lambda_{ml}(\omega)\leq\lambda_1}R^\omega_{ml}(r)S_{ml}(a\omega,\cos \theta )e^{im \phi}e^{i\omega t}d\omega$$
$$\Phi_{\lessflat}=\int_{-\infty}^{\infty}\zeta(\frac{\omega}{\omega_1})\sum_{m,l:\lambda_{ml}(\omega)>\lambda_1}R^\omega_{ml}(r)S_{ml}(a\omega,\cos \theta )e^{im \phi}e^{i\omega t}d\omega$$
$$\Phi_{\mbox{$\natural$}}=\int_{-\infty}^{\infty}(1-\zeta(\frac{\omega}{\omega_1}))\sum_{m,l:\lambda_{ml}(\omega)\geq\lambda_2\omega^2}R^\omega_{ml}(r)S_{ml}(a\omega,\cos \theta )e^{im \phi}e^{i\omega t}d\omega$$
$$\Phi_{\mbox{$\sharp$}}=\int_{-\infty}^{\infty}(1-\zeta(\frac{\omega}{\omega_1}))\sum_{m,l:\lambda_{ml}(\omega)<\lambda_2\omega^2}R^\omega_{ml}(r)S_{ml}(a\omega,\cos \theta )e^{im \phi}e^{i\omega t}d\omega.$$
In this decomposition, we think of $\omega_1$ as large and $\lambda_2$ as small.
\subsection{The Trapped Frequencies}
Trapping occurs for $\Phi_{\mbox{$\natural$}}$. An integrated decay estimate is proved in detail in \cite{DRL}. We refer the readers to Section 5.3.3 of \cite{DRL}. Notice that the first term on the right hand side in the following Proposition is different from that in \cite{DRL}, but the inequality still holds as a result of the proof of the corresponding inequality in \cite{DRL}.
\begin{proposition}\label{trapped}
\begin{equation*}
\begin{split}
&\iint_{\mathcal R(-\infty,\infty)}\left(\chi\Phi_{\mbox{$\natural$}}^2+\chi(\partial_r\Phi_{\mbox{$\natural$}})^2+\chi\mathbbm{1}_{\{|r- 3M|\geq \frac{M}{8}\}} J^N_\mu\left(\Phi_{\mbox{$\natural$}}\right)n^\mu_{\Sigma_{t^*}}\right)\\
\leq &C \int_{\mathcal H(-\infty,\infty)} \left(\partial_{t^*}\Phi_{\tau'}^\tau\right)^2+C\epsilon\int_{\mathcal H(-\infty,\infty)} \left(\partial_{\phi^*}\Phi_{\tau'}^{\tau}\right)^2+ \int_{-\infty}^\infty dt^*\int_{r\ge R}
\left(2f (r^2+a^2)^{1/2} F_{\mbox{$\natural$}} \partial_{r^*} ((r^2+a^2)^{1/2} \Phi_{\tau'}^\tau )
\right.\\
&\hskip2pc  \left.
+ f'(r^2+a^2)F_{\mbox{$\natural$}} 
\Phi_{\tau'}^\tau  \right)\frac{\Delta}{r^2+a^2}
 \sin \theta\,  d\phi \, d\theta \,  dr^* +\delta'\iint_{\mathcal R\cap\{r\leq R\}}(\Phi^\tau_{\tau'})^2+(\partial_{r^*}\Phi_{\tau'}^\tau)^2
 \\&+C(\delta')^{-1}\iint_{\mathcal R\cap\{r\leq R\}} F^2,
\end{split}
\end{equation*}
where $\chi$ is a weight that degenerates at infinity and near the event horizon and $f$ is increasing and $f=\tan^{-1}\frac{r^*-\alpha-\sqrt{\alpha}}{\alpha}-\tan^{-1}(-1-\alpha)^{-\frac{1}{2}},$ for $r>R$ for some fixed $\alpha$.
\end{proposition}

\subsection{The Untrapped Frequencies}
For each of the pieces that are untrapped, i.e., $\Phi_{\eighthnote}$ for $\hbox{\eighthnote}={\mbox{$\flat$}}, {\lessflat} \mbox{ or } {\mbox{$\sharp$}}$, a vector field $X_{\eighthnote}$ is constructed in \cite{DRL} so that
$$\iint_{\mathcal R(-\infty,\infty)}\chi\left( J^{N_e}_\mu\left(\Phi_{\eighthnote}\right)n^\mu_{\Sigma_{t^*}}+\Phi_{\eighthnote}^2\right)\leq C\iint_{\mathcal R(-\infty,\infty)} K^{X_{\eighthnote}}\left(\Phi_{\eighthnote}\right),$$
where $\chi$ is a weight function that both degenerates at infinity and vanishes around the event horizon. Using this vector field and the conservation identity, it is shown in Section 5.3.4 in \cite{DRL} that
\begin{proposition}\label{untrapped}
\begin{equation*}
\begin{split}
&\iint_{\mathcal{R}(-\infty,\infty)} \chi\left(
(J^{N_e}_\mu(\Phi_{\mbox {$\flat$}})+  J^{N_e}_\mu(\Phi_{\lessflat})+J^N_\mu(\Phi_{\mbox {$\sharp$}})
 )n^\mu_{\Sigma_\tau}+\left(\Phi_{\mbox {$\flat$}}^2+\Phi_{\lessflat}^2+\Phi_{\mbox {$\sharp$}}^2\right)\right)\\
 \le& 
C \int_{\mathcal H(-\infty,\infty)}J^{N_e}_\mu(\Phi) n_{\mathcal H^+}^\mu  +  C (\delta')^{-1} \iint_{\mathcal{R}(-\infty,\infty)\cap \{r\le R\}}F^2
    \\
&
 +   C \delta' \iint_{\mathcal{R}(-\infty,\infty)\cap \{r\le R\}}
(\Phi_{\tau'}^\tau)^2+(\partial_{r^*}\Phi_{\tau'}^\tau)^2 +\mathbbm 1_{\{r\leq \frac{23M}{8}\}} J^N_\mu(\Phi_{\tau'}^\tau) n_{\Sigma_{t^*}}^\mu\\
&
+ \int_{-\infty}^\infty dt^*\int_{\{r\ge R\}}
\left(2f (r^2+a^2)^{1/2} (F_{\mbox {$\flat$}}+F_{\lessflat}+F_{\mbox {$\sharp$}})
\partial_{r^*} ((r^2+a^2)^{1/2} \Phi_{\tau'}^\tau)
\right.\\
&  \left.
+ f'(r^2+a^2)(F_{\mbox {$\flat$}}+F_{\lessflat}+F_{\mbox {$\sharp$}})
\Phi_{\tau'}^\tau  \right)\frac{\Delta}{r^2+a^2}
 \sin \theta\,  d\phi \, d\theta \,  dr^* ,
\end{split}
\end{equation*}
where $\chi$ and $f$ are exactly as in Proposition \ref{trapped}.
\end{proposition}
\begin{proof}
This inequality is essentially borrowed from Section 5.3.4 in \cite{DRL}. The only difference is the first term on the right hand side of the inequality. In \cite{DRL}, the estimate
$$\int_{\mathcal H(-\infty,\infty)} J^{X_{\eighthnote}}_{\mu}(\Phi_{\eighthnote}) n_{\mathcal H^+}^\mu \leq C\int_{\Sigma_{\tau'}}J^N_\mu(\Phi^\tau_{\tau'}) n^\mu_{\Sigma_{\tau'}}$$
is used. Here, we have not proved boundedness of the solution and hence we are content with the estimate
$$\int_{\mathcal H(-\infty,\infty)} J^{X_{\eighthnote}}_{\mu}(\Phi_{\eighthnote}) n_{\mathcal H^+}^\mu \leq C\int_{\mathcal H(-\infty,\infty)}J^{N_e}_\mu(\Phi^\tau_{\tau'}) n_{\mathcal H^+}^\mu.$$
Notice that this estimate holds for $C$ independent of $e$ because $X_{\eighthnote}$ is constructed as $f\partial_{r^*}$ and on the event horizon, $\partial_{r^*}=O(1)\hat V+O(\epsilon)E_A$.
\end{proof}

\subsection{The Integrated Decay Estimates}
In order to add up the estimates in the previous sections, we need a Hardy-type inequality:
\begin{proposition}\label{hi}
For $R'<R$,
$$\int_{\Sigma_\tau\cap\{r\geq R\}} r^{\alpha-2}\Phi^2 \leq C\int_{\Sigma_\tau\cap\{r\geq R'\}} r^{\alpha}J^{N_e}_\mu\left(\Phi\right)n^\mu_{\Sigma_\tau}.$$
\end{proposition}
\begin{proof}
Let $k(r)$ be defined by solving $$k'(r,\theta,\phi)=r^{\alpha-2}vol,$$ where $vol=vol\left(r,\theta,\phi\right)$ is the volume density on $\Sigma_\tau$ with $r, \theta,\phi$ coordinates, with boundary condition $k(R',\theta,\phi)=0$.
Now
\begin{equation*}
\begin{split}
\int_{\Sigma_\tau} r^{\alpha-2}\Phi^2=& \iiint_{r_+}^{\infty} k'(r)\Phi^2 dr d\theta d\phi\\
=& -2\iiint k(r)\Phi\partial_r\Phi dr d\theta d\phi\\
\leq &2\left(\iiint \frac{k(r)^2}{k'(r)}\left(\partial_r \Phi\right)^2 dr d\theta d\phi\right)^{\frac{1}{2}}\left(\iiint k'(r)\Phi^2 dr d\theta d\phi\right)^{\frac{1}{2}}
\end{split}
\end{equation*}
Notice that $vol \sim r^2$, $k(r)\sim r^{\alpha+1}$ and $k'(r)\sim r^{\alpha}$. Hence $\frac{1+k(r)^2}{1+k'(r)}\sim r^{\alpha}vol$. The lemma follows.
\end{proof}
We now add up the estimates for $\Phi_{\mbox {$\flat$}}$, $\Phi_{\lessflat}$, $\Phi_{\mbox {$\natural$}}$ and $\Phi_{\mbox {$\sharp$}}$.
\begin{proposition}\label{ide}
\begin{equation*}
\begin{split}
&\iint_{\mathcal R(\tau',\tau)}\left(r^{-1-\delta}\mathbbm{1}_{\{|r- 3M|\geq \frac{M}{8}\}}J^{N_e}_\mu\left(\Phi\right)n^\mu_{\Sigma_\tau}+r^{-1-\delta}\left(\partial_r\Phi\right)^2+r^{-3-\delta}\Phi^2\right)\\
\leq& C\left(\int_{\Sigma_\tau}J^{N_e}_\mu\left(\Phi\right)n^\mu_{\Sigma_\tau}+\int_{\Sigma_{\tau'}}J^{N_e}_\mu\left(\Phi\right)n^\mu_{\Sigma_{\tau'}}+\int_{\mathcal H(\tau',\tau)}J^{N_e}_\mu\left(\Phi\right)n^\mu_{\mathcal H^+}\right.\\
&\left.+C\iint_{\mathcal R(\tau'-1,\tau+1)}\left(|\partial_{r^*}\Phi|+r^{-1}|\Phi|\right)|G|+C|\iint_{\mathcal R(\tau'-1,\tau+1)\cap\{r\leq \frac{23M}{8}\}}(\partial_{t^*}\Phi+eY\Phi)G|\right.\\
&\left.
 +C\iint_{\mathcal R(\tau'-1,\tau+1)}G^2\right).
\end{split}
\end{equation*}
\end{proposition}
\begin{proof}
Since the function $f$ appears identically in Propositions \ref{trapped} and \ref{untrapped}, we can add up the estimates to obtain:
\begin{equation*}
\begin{split}
&\iint_{\mathcal{R}(-\infty,\infty)} \left(\chi
(J^{N_e}_\mu(\Phi_{\mbox {$\flat$}})+  J^{N_e}_\mu(\Phi_{\lessflat})+\mathbbm 1_{\{|r-3M|\geq\frac{M}{8}\}}J^{N_e}_\mu(\Phi_{\mbox {$\natural$}})+J^{N_e}_\mu(\Phi_{\mbox {$\sharp$}})
 )n^\mu_{\Sigma_\tau}+{\chi}\left(\Phi_{\mbox {$\flat$}}^2+\Phi_{\lessflat}^2+\Phi_{\mbox {$\natural$}}^2+\Phi_{\mbox {$\sharp$}}^2\right)\right)\\
 \le& 
C \int_{\mathcal H(\tau'-1,\tau+1)}J^{N_e}_\mu(\Phi) n^\mu_{\mathcal H^+}   +  C (\delta')^{-1} \iint_{\mathcal{R}(-\infty,\infty)\cap \{r\le R\}}F^2
    \\
&
 +   C \delta' \iint_{\mathcal{R}(-\infty,\infty)\cap \{r\le R\}}
(\Phi_{\tau'}^\tau)^2+(\partial_{r^*}\Phi_{\tau'}^\tau)^2 +\mathbbm 1_{\{r\leq\frac{23M}{8}\}} J^{N_e}_\mu(\Phi_{\tau'}^\tau) n^\mu_{\Sigma_{t^*}}\\
&
+ \int_{-\infty}^\infty dt^*\int_{\{r\ge R\}}
\left(2f (r^2+a^2)^{1/2} (F_{\mbox {$\flat$}}+F_{\lessflat}+F_{\mbox {$\natural$}}+F_{\mbox {$\sharp$}})
\partial_{r^*} ((r^2+a^2)^{1/2} \Phi_{\tau'}^\tau)
\right.\\
&  \left.+ f'(r^2+a^2)(F_{\mbox {$\flat$}}+F_{\lessflat}+F_{\mbox {$\natural$}}+F_{\mbox {$\sharp$}})
\Phi_{\tau'}^\tau  \right)\frac{\Delta}{r^2+a^2}
 \sin \theta\,  d\phi \, d\theta \,  dr^* .
\end{split}
\end{equation*}
By the definition of the cutoff, we have the pointwise equalities
$$F=F_{\mbox {$\flat$}}+F_{\lessflat}+F_{\mbox {$\natural$}}+F_{\mbox {$\sharp$}}.$$
Therefore, we have
\begin{equation*}
\begin{split}
&\iint_{\mathcal{R}(-\infty,\infty)} \chi \mathbbm 1_{\{|r-3M|\geq\frac{M}{8}\}}
J^{N_e}_\mu(\Phi_{\tau'}^{\tau} ) n^\mu_{\Sigma_{t^*}}+{\chi}\left(\Phi_{\tau'}^\tau \right)^2\\
\le& 
C \int_{\mathcal H(\tau'-1,\tau+1)}J^{N_e}_\mu(\Phi) n^\mu_{\mathcal H^+}  +  C (\delta')^{-1} \iint_{\mathcal{R}(-\infty,\infty)\cap \{r\le R\}}F^2
    \\
&
 +   C \delta' \iint_{\mathcal{R}(-\infty,\infty)\cap \{r\le R\}}
(\Phi_{\tau'}^\tau)^2+(\partial_{r^*}\Phi_{\tau'}^\tau)^2 +\mathbbm 1_{\{r\leq\frac{23M}{8}\}} J^{N_e}_\mu(\Phi_{\tau'}^\tau) n^\mu_{\Sigma_{t^*}}\\
&
+ \int_{-\infty}^\infty dt^*\int_{\{r\ge R\}}
\left(2f (r^2+a^2)^{1/2} F
\partial_{r^*} ((r^2+a^2)^{1/2} \Phi_{\tau'}^\tau) + f'(r^2+a^2)F
\Phi_{\tau'}^\tau  \right)\frac{\Delta}{r^2+a^2}
 \sin \theta\,  d\phi \, d\theta \,  dr^* .
\end{split}
\end{equation*}
First, by Proposition \ref{mg}, we have 
$$\int_{\mathcal H(\tau'-1,\tau+1)}J^{N_e}_\mu(\Phi) n^\mu_{\mathcal H^+}\leq C \int_{\Sigma_{\tau'}}J^{N_e}_\mu(\Phi) n^\mu_{\Sigma_{\tau'}}  +C\int_{\Sigma_\tau} J^{N_e}_\mu\left(\Phi\right)n^\mu_{\Sigma_{\tau}}+\int_{\mathcal H(\tau',\tau)}J^{N_e}_\mu(\Phi) n^\mu_{\mathcal H^+}.$$
Recall that
$$F=\xi G+2D^\alpha\Phi D_\alpha\xi+\Phi\Box_{g_K}\xi.$$
Notice that by the definition of $\xi$, the last two terms are supported in the $t^*$ range $(\tau'-1,\tau')\cup (\tau,\tau+1)$. Moreover, since $\xi$ depends only on $t^*$, the only terms involving $D\Phi$ are $\partial_{t^*}\Phi$ and $O(\epsilon)\partial_{\phi^*}\Phi$. Using this, we immediately have the following with $C$ independent of $e$ as long as $\epsilon\ll e$:
\begin{equation*}
\begin{split}
&C (\delta')^{-1} \iint_{\mathcal{R}(-\infty,\infty)\cap \{r\le R\}}F^2\\
\leq &C (\delta')^{-1}\left( \iint_{\mathcal{R}(\tau'-1,\tau+1)\cap \{r\le R\}}G^2+\iint_{\mathcal{R}(\tau'-1,\tau')\cup\mathcal{R}(\tau,\tau+1)}\left(r^{-2}\Phi^2+J^{N_e}_\mu\left(\Phi\right)n^\mu_{\Sigma_{t^*}}\right)\right).
\end{split}
\end{equation*}
Similarly, we have
\begin{equation*}
\begin{split}
&\int_{-\infty}^\infty dt^*\int_{\{r\ge R\}}
f'(r^2+a^2)F
\Phi_{\tau'}^\tau \frac{\Delta}{r^2+a^2}
\sin \theta\,  d\phi \, d\theta \,  dr^*\\
\leq &C \left(\iint_{\mathcal R(\tau'-1,\tau+1)} r^{-1}|\Phi||G|+\iint_{\mathcal R(\tau'-1,\tau')\cup\mathcal R(\tau,\tau+1)}\left(r^{-2}\Phi^2+J^{N_e}_\mu\left(\Phi\right)n^\mu_{\Sigma_{t^*}}\right)\right).
\end{split}
\end{equation*}
The other term with $F$ is more delicate to estimate. One of the terms in the expansion does not have sufficient decay in $r$:
\begin{equation*}
\begin{split}
&\int_{-\infty}^\infty dt^*\int_{\{r\ge R\}}
2f (r^2+a^2)^{1/2} F
\partial_{r^*} ((r^2+a^2)^{1/2} \Phi_{\tau'}^\tau) \frac{\Delta}{r^2+a^2}
 \sin \theta\,  d\phi \, d\theta \,  dr^*\\
\leq &C \left(\iint_{\mathcal R(\tau'-1,\tau+1)} r^{-1}|\Phi||G|+\iint_{\mathcal R(\tau'-1,\tau')\cup\mathcal R(\tau,\tau+1)}\left(r^{-2}\Phi^2+J^{N_e}_\mu\left(\Phi\right)n^\mu_{\Sigma_{t^*}}\right)\right)\\
&+\int_{-\infty}^\infty dt^*\int_{\{r\ge R\}}
2f (r^2+a^2)^{1/2} \Phi\Box_{g_K}\xi
\partial_{r^*} ((r^2+a^2)^{1/2}\Phi) \xi \frac{\Delta}{r^2+a^2}
 \sin \theta\,  d\phi \, d\theta \,  dr^*\\
\end{split}
\end{equation*}
Nevertheless, noting that $\xi$ is independent of $r^*$, an integration by parts in $r^*$ would give
\begin{equation*}
\begin{split}
&\int_{-\infty}^\infty dt^*\int_{\{r\ge R\}}
2f (r^2+a^2)^{1/2} \Phi\Box_{g_K}\xi
\partial_{r^*} ((r^2+a^2)^{1/2}\Phi) \xi \frac{\Delta}{r^2+a^2}
 \sin \theta\,  d\phi \, d\theta \,  dr^*\\
=&-\int_{-\infty}^\infty dt^*\int_{\{r\ge R\}}
(r^2+a^2) \Phi^2\xi\Box_{g_K}\xi
 \partial_{r^*}\left(f\frac{\Delta}{r^2+a^2}\right)
 \sin \theta\,  d\phi \, d\theta \,  dr^*+\mbox{boundary terms}\\
\leq&C\iint_{\mathcal R(\tau'-1,\tau')\cup\mathcal R(\tau,\tau+1)}r^{-2}\Phi^2,
\end{split}
\end{equation*}
where the boundary terms can be controlled (after possibly changing $R$) by pigeonholing in $r\in [R,R+1]$.
By the mild growth estimate of Proposition \ref{mg}, the estimate near the event horizon from Proposition \ref{Xhor} and the Hardy inequality of Proposition \ref{hi},
\begin{equation*}
\begin{split}
&\iint_{\mathcal R(\tau'-1,\tau')\cup\mathcal R(\tau,\tau+1)}\left(r^{-2}\Phi^2+J^{N_e}_\mu\left(\Phi\right)n^\mu_{\Sigma_{t^*}}\right)\\
\leq &C\left(\int_{\Sigma_{\tau'}} J^{N_e}_\mu\left(\Phi\right)n^\mu_{\Sigma_{\tau'}}+\int_{\Sigma_\tau} J^{N_e}_\mu\left(\Phi\right)n^\mu_{\Sigma_{\tau'}}+\iint_{\mathcal R(\tau'-1,\tau')\cup\mathcal R(\tau,\tau+1)}G^2\right).
\end{split}
\end{equation*}
Therefore, using all the above estimates and noticing the support of $\xi$, we have
\begin{equation*}
\begin{split}
&\iint_{\mathcal{R}(\tau',\tau)} \chi \left(\mathbbm 1_{\{|r-3M|\geq\frac{M}{8}\}}
J^{N_e}_\mu(\Phi ) n^\mu_{\Sigma_{t^*}}+\Phi^2\right)\\
\le& 
C(\delta')^{-1} \int_{\Sigma_{\tau'}}J^{N_e}_\mu(\Phi) n^\mu_{\Sigma_{\tau'}}  +C(\delta')^{-1}\int_{\Sigma_\tau} J^{N_e}_\mu\left(\Phi\right)n^\mu_{\Sigma_{\tau}}+C\int_{\mathcal H(\tau',\tau)} J^{N_e}_\mu\left(\Phi\right)n^\mu_{\mathcal H^+} 
    \\
&+C\iint_{\mathcal R(\tau'-1,\tau+1)}r^{-1}|\Phi||G|+C(\delta')^{-1}\iint_{\mathcal R(\tau'-1,\tau+1)}G^2\\
&
 +   C \delta' \iint_{\mathcal{R}(\tau',\tau)\cap \{r\le R\}}
\Phi^2+(\partial_{r^*}\Phi)^2 +\mathbbm 1_{\{r\leq\frac{23M}{8}\}} J^{N_e}_\mu(\Phi) n^\mu_{\Sigma_{t^*}}.
\end{split}
\end{equation*}
We add to this the estimates near spatial infinity and the event horizon, i.e., Propositions \ref{Xinf} and \ref{Xhor}, to get
\begin{equation*}
\begin{split}
&\iint_{\mathcal{R}(\tau',\tau)} r^{-1-\delta}\mathbbm 1_{\{|r-3M|\geq\frac{M}{8}\}}
J^{N_e}_\mu(\Phi ) n^\mu_{\Sigma_{t^*}}+r^{-1-\delta}\left(\partial_r\Phi\right)^2+r^{-3-\delta}\Phi^2\\
\le& 
C (\delta')^{-1}\int_{\Sigma_{\tau'}}J^{N}_\mu(\Phi) n^\mu_{\Sigma_{\tau'}}  +C(\delta')^{-1}\int_{\Sigma_\tau} J^{N_e}_\mu\left(\Phi\right)n^\mu_{\Sigma_{\tau}}+C\int_{\mathcal H(\tau',\tau)} J^{N_e}_\mu\left(\Phi\right)n^\mu_{\mathcal H^+}   \\
&+C\iint_{\mathcal R(\tau'-1,\tau+1)}\left(|\partial_{r^*}\Phi|+r^{-1}|\Phi|\right)|G|+C|\iint_{\mathcal R(\tau'-1,\tau+1)\cap\{r\leq \frac{23M}{8}\}}(\partial_{t^*}\Phi+eY\Phi)G|\\
&
 +C(\delta')^{-1}\iint_{\mathcal R(\tau'-1,\tau+1)}G^2+   C \delta' \iint_{\mathcal{R}(\tau'-1,\tau+1)\cap \{r\le R\}}
\Phi^2+(\partial_{r^*}\Phi)^2 +\mathbbm 1_{\{r\leq\frac{23M}{8}\}} J^N_\mu(\Phi) n^\mu_{\Sigma_{t^*}}.
\end{split}
\end{equation*}
By choosing $\delta'$ sufficiently small and absorbing $(\delta')^{-1}$ into the constant $C$, we can absorb the last term:
\begin{equation*}
\begin{split}
&\iint_{\mathcal{R}(\tau',\tau)} r^{-1-\delta}\mathbbm 1_{\{|r-3M|\geq\frac{M}{8}\}}
J^{N_e}_\mu(\Phi ) n^\mu_{\Sigma_{t^*}}+r^{-1-\delta}\left(\partial_r\Phi\right)^2+r^{-3-\delta}\Phi^2\\
\le& 
C \int_{\Sigma_{\tau'}}J^{N_e}_\mu(\Phi) n^\mu_{\Sigma_{\tau'}}  +C\int_{\Sigma_\tau} J^{N_e}_\mu\left(\Phi\right)n^\mu_{\Sigma_{\tau}}+C\int_{\mathcal H(\tau',\tau)} J^{N_e}_\mu\left(\Phi\right)n^\mu_{\mathcal H^+}   \\
&+C\iint_{\mathcal R(\tau'-1,\tau+1)\cap\{r\leq \frac{23M}{8}\}}\left(|\partial_{r^*}\Phi|+r^{-1}|\Phi|\right)|G|+C|\iint_{\mathcal R(\tau'-1,\tau+1)\cap\{r\leq \frac{23M}{8}\}}(\partial_{t^*}\Phi+eY\Phi)G|\\
&
 +C\iint_{\mathcal R(\tau'-1,\tau+1)}G^2.
\end{split}
\end{equation*}
using Proposition \ref{mg} and \ref{hi} at the last step.
\end{proof}

\begin{definition}
From now on, denote
$$K^{X_0}\left(\Phi\right)=r^{-1-\delta}\mathbbm{1}_{\{|r- 3M|\geq \frac{M}{8}\}}J^N_\mu\left(\Phi\right)n^\mu_{\Sigma_\tau}+r^{-1-\delta}\left(\partial_r\Phi\right)^2+r^{-3-\delta}\Phi^2,\quad\mbox{and}$$
$$K^{X_1}\left(\Phi\right)=r^{-1-\delta}J^N_\mu\left(\Phi\right)n^\mu_{\Sigma_\tau}+r^{-3-\delta}\Phi^2.$$
We remark that this is a slight abuse of notation because these ``currents'' do not arise directly from a vector field.
\end{definition}

\subsection{Boundedness of the Non-degenerate Energy}
\begin{proposition}\label{bdd}
Let $\Phi$ satisfy $\Box_{g_K}\Phi=G$. For $e$ sufficiently small and $\epsilon \ll e$, we have 
\begin{equation*}
\begin{split}
&\int_{\Sigma_{\tau}} J^{N_e}_\mu\left(\Phi\right)n^\mu_{\Sigma_{\tau}} +\int_{\mathcal H(\tau',\tau)} J^{N_e}_\mu\left(\Phi\right)n^\mu_{\mathcal H^+} +\iint_{\mathcal R(\tau',\tau)\cap\{r\leq r^-_Y\}}K^{N_e}_\mu\left(\Phi\right)\\
\leq &C\left(\int_{\Sigma_{\tau'}} J^{N_e}_\mu\left(\Phi\right)n^\mu_{\Sigma_{\tau'}}+|\iint_{\mathcal R(\tau'-1,\tau+1)} \partial_{t^*}\Phi G|+|\iint_{\mathcal R(\tau'-1,\tau+1)}eY\Phi G|\right.\\
&\left.+\iint_{\mathcal R(\tau'-1,\tau+1)}\left(|\partial_r\Phi|+r^{-1}|\Phi|\right) |G|+\iint_{\mathcal R(\tau'-1,\tau+1)}G^2\right).\\
\end{split}
\end{equation*}
\end{proposition}
\begin{proof}
We recall that $$\iint_{\mathcal R(\tau',\tau)\cap\{r^-_Y\leq r\leq r^+_Y\}} K^{Y}\left(\Phi\right)\leq C\int_{\tau'}^{\tau}\int_{\Sigma_{\bar\tau}} J^{N_e}_\mu\left(\Phi\right)n^\mu_{\Sigma_{t^*}}d{t^*},$$ with $C$ independent of $e$ and $\epsilon$ whenever $\epsilon\ll e<1$. At this point, we choose $r^+_Y < \frac{11M}{4}<\frac{23M}{8}$. Hence this term can be controlled by the integrated decay estimates.
Then, by Proposition \ref{Nid},
\begin{equation*}
\begin{split}
&\int_{\Sigma_{\tau}} J^{N_e}_\mu\left(\Phi\right)n^\mu_{\Sigma_{\tau}} +\int_{\mathcal H(\tau',\tau)} J^{N_e}_\mu\left(\Phi\right)n^\mu_{\mathcal H^+} +\iint_{\mathcal R(\tau',\tau)\cap\{r\leq r^-_Y\}}K^{N_e}_\mu\left(\Phi\right)\\
=&\int_{\Sigma_{\tau'}} J^{N_e}_\mu\left(\Phi\right)n^\mu_{\Sigma_{\tau'}}+e\iint_{\mathcal R(\tau',\tau)\cap\{r^-_Y\leq r\leq r^+_Y\}} K^Y\left(\Phi\right)+\iint_{\mathcal R(\tau',\tau)} \left(\partial_{t^*}\Phi+eY\Phi\right) G\\
\leq &\int_{\Sigma_{\tau'}} J^{N_e}_\mu\left(\Phi\right)n^\mu_{\Sigma_{\tau'}}+Ce\int_{\tau'}^{\tau}\int_{\Sigma_{\bar\tau}\cap\{r^-_Y\leq r\leq r^+_Y\}} J^{N_e}_\mu\left(\Phi\right)n^\mu_{\Sigma_{\bar\tau}}d{\bar\tau}+|\iint_{\mathcal R(\tau',\tau)} \left(\partial_{t^*}\Phi+eY\Phi\right) G|\\
\leq &\int_{\Sigma_{\tau'}} J^{N_e}_\mu\left(\Phi\right)n^\mu_{\Sigma_{\tau'}}+Ce\left(\int_{\Sigma_\tau}J^{N_e}_\mu\left(\Phi\right)n^\mu_{\Sigma_\tau}+\int_{\Sigma_{\tau'}}J^{N_e}_\mu\left(\Phi\right)n^\mu_{\Sigma_{\tau'}}+\int_{\mathcal H(\tau',\tau)}J^{N_e}_\mu(\Phi)n^\mu_{\mathcal H^+}\right.\\
&\left.+\iint_{\mathcal R(\tau'-1,\tau+1)}\left(|\partial_r\Phi|+r^{-1}|\Phi|\right) |G|+|\iint_{\mathcal R(\tau'-1,\tau+1)}G^2\right)+|\iint_{\mathcal R(\tau',\tau)} \left(\partial_{t^*}\Phi+eY\Phi\right) G|.
\end{split}
\end{equation*}
Hence, the Proposition holds if $e$ is chosen to be sufficiently small.
\end{proof}
\begin{remark}\label{fixr+e}
From this point on, we will consider $r^+_Y$ and $e$ to be fixed. After $e$ is fixed, the vector field $N_e$ will be written simply as $N$.
\end{remark}
We now estimate the inhomogeneous terms in Proposition \ref{bdd}:
\begin{proposition}\label{Nbdd}
\begin{equation*}
\begin{split}
&\int_{\Sigma_{\tau}} J^{N}_\mu\left(\Phi\right)n^\mu_{\Sigma_{\tau}} +\int_{\mathcal H(\tau',\tau)} J^{N}_\mu\left(\Phi\right)n^\mu_{\mathcal H^+} +\iint_{\mathcal R(\tau',\tau)\cap\{r\leq r^-_Y\}}K^{N}\left(\Phi\right)+\iint_{\mathcal R(\tau',\tau)}K^{X_0}\left(\Phi\right)\\
\leq &C\left(\int_{\Sigma_{\tau'}} J^{N}_\mu\left(\Phi\right)n^\mu_{\Sigma_{\tau'}}+\left(\int_{\tau'-1}^{\tau+1}\left(\int_{\Sigma_{t^*}} G^2\right)^{\frac{1}{2}}dt^*\right)^2+\iint_{\mathcal R(\tau',\tau)}G^2\right).
\end{split}
\end{equation*}

\end{proposition}
\begin{proof}
Adding the estimates in Propositions \ref{ide} and $\delta$ times the estimates in Proposition \ref{bdd},
\begin{equation*}
\begin{split}
&\int_{\Sigma_{\tau}} J^{N}_\mu\left(\Phi\right)n^\mu_{\Sigma_{\tau}} +\int_{\mathcal H(\tau',\tau)} J^{N}_\mu\left(\Phi\right)n^\mu_{\mathcal H^+} +\iint_{\mathcal R(\tau',\tau)\cap\{r\leq r^-_Y\}}K^{N}\left(\Phi\right)+\delta\iint_{\mathcal R(\tau',\tau)}K^{X_0}\left(\Phi\right)\\
\leq &C\left(\int_{\Sigma_{\tau'}} J^{N}_\mu\left(\Phi\right)n^\mu_{\Sigma_{\tau'}}+\iint_{\mathcal R(\tau'-1,\tau+1)} \left(|\partial\Phi|+r^{-1}|\Phi|\right) |G|+\iint_{\mathcal R(\tau'-1,\tau+1)}G^2\right)\\
&+C\delta\left(\int_{\Sigma_\tau}J^N_\mu\left(\Phi\right)n^\mu_{\Sigma_\tau}+\int_{\mathcal H(\tau',\tau)}J^N_\mu\left(\Phi\right)n^\mu_{\mathcal H^+}\right)\\
\leq &C\left(\int_{\Sigma_{\tau'}} J^{N}_\mu\left(\Phi\right)n^\mu_{\Sigma_{\tau'}}+\sup_{t^*\in [\tau'-1,\tau+1]}\left(\int_{\Sigma_{t^*}}J^{N}_\mu\left(\Phi\right)n^\mu_{\Sigma_{t^*}}\right)^{\frac{1}{2}}\int_{\tau'-1}^{\tau+1}\left(\int_{\Sigma_{t^*}} G^2\right)^{\frac{1}{2}}dt^*\right.\\
&\left.+\iint_{\mathcal R(\tau'-1,\tau+1)}G^2\right)+C\delta\left(\int_{\Sigma_\tau}J^N_\mu\left(\Phi\right)n^\mu_{\Sigma_\tau}+\int_{\mathcal H(\tau',\tau)}J^N_\mu\left(\Phi\right)n^\mu_{\mathcal H^+}\right),\\
\end{split}
\end{equation*}
where at the last step we have used Proposition \ref{hi}. Choosing $C\delta\leq \frac{1}{2}$, we can absorb the last term to the left hand side to get
\begin{equation}\label{bddtstep1}
\begin{split}
&\int_{\Sigma_{\tau}} J^{N}_\mu\left(\Phi\right)n^\mu_{\Sigma_{\tau}} +\int_{\mathcal H(\tau',\tau)} J^{N}_\mu\left(\Phi\right)n^\mu_{\mathcal H^+} +\iint_{\mathcal R(\tau',\tau)\cap\{r\leq r^-_Y\}}K^{N}\left(\Phi\right)+\delta\iint_{\mathcal R(\tau',\tau)}K^{X_0}\left(\Phi\right)\\
\leq &C\left(\int_{\Sigma_{\tau'}} J^{N}_\mu\left(\Phi\right)n^\mu_{\Sigma_{\tau'}}+\sup_{t^*\in [\tau'-1,\tau+1]}\left(\int_{\Sigma_{t^*}}J^{N}_\mu\left(\Phi\right)n^\mu_{\Sigma_{t^*}}\right)^{\frac{1}{2}}\int_{\tau'-1}^{\tau+1}\left(\int_{\Sigma_{t^*}} G^2\right)^{\frac{1}{2}}dt^*\right.\\
&\left.+\iint_{\mathcal R(\tau'-1,\tau+1)}G^2\right).\\
\end{split}
\end{equation}
By considering the above estimate on $[\tau',\tilde{\tau}]$, where $\tilde{\tau}$ is when the supremum on the right hand side is achieved, and using Proposition \ref{mg}, we get
\begin{equation*}
\begin{split}
&\int_{\Sigma_{\tilde{\tau}}} J^{N}_\mu\left(\Phi\right)n^\mu_{\Sigma_{\tilde{\tau}}} +\int_{\mathcal H(\tau',\tilde{\tau})} J^{N}_\mu\left(\Phi\right)n^\mu_{\mathcal H^+} +\iint_{\mathcal R(\tau',\tilde{\tau})\cap\{r\leq r^-_Y\}}K^{N}\left(\Phi\right)+\delta\iint_{\mathcal R(\tau',\tilde{\tau})}K^{X_0}\left(\Phi\right)\\
\leq &C\left(\int_{\Sigma_{\tau'}} J^{N}_\mu\left(\Phi\right)n^\mu_{\Sigma_{\tau'}}+\left(\int_{\tau'-1}^{\tau+1}\left(\int_{\Sigma_{t^*}} G^2\right)^{\frac{1}{2}}dt^*\right)^2\right).
\end{split}
\end{equation*}
We plug this into (\ref{bddtstep1}) and apply Cauchy-Schwarz to prove the proposition.
\end{proof}

We can also estimate the inhomogeneous terms not in $L^1L^2$ but in $L^2L^2$, provided that we allow some extra factors of $r$ and some loss of derivatives in $G$. This is especially useful for estimating the commutator terms from $S$, which do not have sufficient decay in $t^*$ in the interior to be estimated in $L^1L^2$. More precisely, we have 
\begin{proposition}\label{bddloss}
\begin{equation*}
\begin{split}
&\int_{\Sigma_{\tau}}J^N_\mu\left(\Phi\right)n^\mu_{\Sigma_{\tau}}+\int_{\mathcal H(\tau',\tau)}J^N_\mu\left(\Phi\right)n^\mu_{\mathcal H^+}+\iint_{\mathcal R(\tau',\tau)\cap\{r\leq r^-_Y\}}K^{N}\left(\Phi\right)+\iint_{\mathcal R(\tau',\tau)}K^{X_0}\left(\Phi\right)\\
\leq &C\left(\int_{\Sigma_{\tau'}} J^{N}_\mu\left(\Phi\right)n^\mu_{\Sigma_{\tau'}}+\sum_{m=0}^{1}\iint_{\mathcal R(\tau'-1,\tau+1)}r^{1+\delta}\left(\partial_{t^*}^{m}G\right)^2+\sup_{t^*\in [\tau'-1,\tau+1]}\int_{\Sigma_{t^*}\cap\{|r-3M|\leq\frac{M}{8}\}} {G}^2\right).
\end{split}
\end{equation*}
\end{proposition}
\begin{proof}
By Propositions \ref{ide} and \ref{bdd},
\begin{equation*}
\begin{split}
&\int_{\Sigma_{\tau}}J^N_\mu\left({\Phi}\right)n^\mu_{\Sigma_{\tau}}+\int_{\mathcal H(\tau',\tau)}J^N_\mu\left(\Phi\right)n^\mu_{\mathcal H^+}+\iint_{\mathcal R(\tau',\tau)\cap\{r\leq r^-_Y\}}K^{N}\left(\Phi\right)+\delta\iint_{\mathcal R(\tau',\tau)}K^{X_0}\left(\Phi\right)\\
\leq &C\left(\int_{\Sigma_{\tau'}} J^{N_e}_\mu\left(\Phi\right)n^\mu_{\Sigma_{\tau'}}+|\iint_{\mathcal R(\tau'-1,\tau+1)} \partial_{t^*}\Phi G|+|\iint_{\mathcal R(\tau'-1,\tau+1)}eY\Phi G|\right.\\
&\left.+\iint_{\mathcal R(\tau'-1,\tau+1)}\left(|\partial_r\Phi|+r^{-1}|\Phi|\right) |G|+\iint_{\mathcal R(\tau'-1,\tau+1)}G^2\right)\\
&+C\delta'\left(\int_{\Sigma_\tau}J^N_\mu\left(\Phi\right)n^\mu_{\Sigma_\tau}+\int_{\mathcal H(\tau',\tau)}J^N_\mu\left(\Phi\right)n^\mu_{\mathcal H^+}\right).
\end{split}
\end{equation*}
Choosing $C\delta'\leq\frac{1}{2}$, we can absorb the last term to the left hand side to get
\begin{equation*}
\begin{split}
&\int_{\Sigma_{\tau}}J^N_\mu\left({\Phi}\right)n^\mu_{\Sigma_{\tau}}+\int_{\mathcal H(\tau',\tau)}J^N_\mu\left(\Phi\right)n^\mu_{\mathcal H^+}+\iint_{\mathcal R(\tau',\tau)\cap\{r\leq r^-_Y\}}K^{N}\left(\Phi\right)+\iint_{\mathcal R(\tau',\tau)}K^{X_0}\left(\Phi\right)\\
\leq &C\left(\int_{\Sigma_{\tau'}} J^{N}_\mu\left(\Phi\right)n^\mu_{\Sigma_{\tau'}}+|\iint_{\mathcal R(\tau'-1,\tau+1)} \partial_{t^*}\Phi G|+|\iint_{\mathcal R(\tau'-1,\tau+1)}eY\Phi G|\right.\\
&\left.+\iint_{\mathcal R(\tau'-1,\tau+1)}\left(|\partial_r\Phi|+r^{-1}|\Phi|\right) |G|+\iint_{\mathcal R(\tau'-1,\tau+1)}G^2\right).\\
\end{split}
\end{equation*}
For the bulk error term, we focus at the region $\{|r-3M|\leq \frac{M}{8}\}$ and integrate by parts. 
\begin{equation*}
\begin{split}
&|\iint_{\mathcal R(\tau'-1,\tau+1)\cap\{|r-3M|\leq \frac{M}{8}\}} \partial_{t^*}\Phi G|\\
\leq &\delta'\iint_{\mathcal R(\tau'-1,\tau+1)\cap\{|r-3M|\leq \frac{M}{8}\}} \Phi^2+C(\delta')^{-1}\iint_{\mathcal R(\tau'-1,\tau+1)\cap\{|r-3M|\leq \frac{M}{8}\}}(\partial_{t^*} G)^2\\
&+|\int_{\Sigma_{\tau+1}\cap\{|r-3M|\leq \frac{M}{8}\}} \Phi G|+|\int_{\Sigma_{\tau'-1}\cap\{|r-3M|\leq \frac{M}{8}\}} \Phi G|\\
\leq &\delta'\iint_{\mathcal R(\tau'-1,\tau+1)\cap\{|r-3M|\leq \frac{M}{8}\}} r^{-3-\delta}\Phi^2+C(\delta')^{-1}\iint_{\mathcal R(\tau'-1,\tau+1)\cap\{|r-3M|\leq \frac{M}{8}\}}(\partial_{t^*} G)^2\\
&+\sup_{t^*\in [\tau'-1,\tau+1] }\left(\delta\int_{\Sigma_{t^*}\cap\{|r-3M|\leq \frac{M}{8}\}} r^{-2}\Phi^2+ C(\delta')^{-1}\int_{\Sigma_{t^*}\cap\{|r-3M|\leq \frac{M}{8}\}} G^2\right)\\
\leq &\delta'\iint_{\mathcal R(\tau'-1,\tau+1)\cap\{|r-3M|\leq \frac{M}{8}\}} r^{-3-\delta}\Phi^2+C(\delta')^{-1}\iint_{\mathcal R(\tau'-1,\tau+1)\cap\{|r-3M|\leq \frac{M}{8}\}}(\partial_{t^*} G)^2\\
&+\sup_{t^*\in [\tau'-1,\tau+1] }\left(\delta'\int_{\Sigma_{t^*}} J^{N}_\mu\left({\Phi}\right)n^\mu_{\Sigma_{t^*}}+ C(\delta')^{-1}\int_{\Sigma_{t^*}\cap\{|r-3M|\leq \frac{M}{8}\}} G^2\right),
\end{split}
\end{equation*}
where at the last step we used Proposition \ref{hi}.\\
Therefore,
\begin{equation}\label{bddX1}
\begin{split}
&\int_{\Sigma_{\tau}}J^N_\mu\left({\Phi}\right)n^\mu_{\Sigma_{\tau}}+\int_{\mathcal H(\tau',\tau)}J^N_\mu\left(\Phi\right)n^\mu_{\mathcal H^+}+\iint_{\mathcal R(\tau',\tau)\cap\{r\leq r^-_Y\}}K^{N}\left(\Phi\right)+\iint_{\mathcal R(\tau',\tau)}K^{X_0}\left(\Phi\right)\\
\leq &C\left(\int_{\Sigma_{\tau'}} J^{N}_\mu\left(\Phi\right)n^\mu_{\Sigma_{\tau'}}+|\iint_{\mathcal R(\tau'-1,\tau+1)} \partial_{t^*}\Phi G|+|\iint_{\mathcal R(\tau'-1,\tau+1)}eY\Phi G|\right.\\
&\left.+\iint_{\mathcal R(\tau'-1,\tau+1)}\left(|\partial_r\Phi|+r^{-1}|\Phi|\right) |G|+\iint_{\mathcal R(\tau'-1,\tau+1)}G^2\right)\\
\leq &C\left(\int_{\Sigma_{\tau'}} J^{N}_\mu\left({\Phi}\right)n^\mu_{\Sigma_{\tau'}}+|\iint_{\mathcal R(\tau'-1,\tau+1)\cap\{|r-3M|\leq\frac{M}{8}\}} \partial_{t^*}{\Phi} G|\right)+C(\delta')^{-1}\iint_{\mathcal R(\tau'-1,\tau+1)}r^{1+\delta}{G}^2\\
&+\delta'\iint_{\mathcal R(\tau'-1,\tau+1)} \left(r^{-3-\delta}{\Phi}^2 +r^{-1-\delta}\left(\partial_r{\Phi}\right)^2+\mathbbm 1_{\{r\leq r^+_Y\}}J^{N}_\mu\left({\Phi}\right)n^\mu_{\Sigma_{t^*}}\right)\\
\leq &C\int_{\Sigma_{\tau'}} J^{N}_\mu\left({\Phi}\right)n^\mu_{\Sigma_{\tau'}}+C(\delta')^{-1}\sum_{m=0}^{1}\int_{\mathcal R(\tau'-1,\tau+1)}r^{1+\delta}\left(\partial_{t^*}^{m}{G}\right)^2+\delta'\iint_{\mathcal R(\tau'-1,\tau+1)} K^{X_0}\left(\Phi\right)\\
&+\sup_{t^*\in [\tau'-1,\tau+1] }\left(\delta'\int_{\Sigma_{t^*}} J^{N}_\mu\left({\Phi}\right)n^\mu_{\Sigma_{t^*}}+ C(\delta')^{-1}\int_{\Sigma_{t^*}\cap\{|r-3M|\leq \frac{M}{8}\}} G^2\right).
\end{split}
\end{equation}
where at the last step we have used Propositions \ref{mg}, \ref{hi}.
Suppose $\displaystyle\sup_{t^*\in [\tau'-1,\tau+1] }\delta'\int_{\Sigma_{t^*}} J^{N}_\mu\left({\Phi}\right)n^\mu_{\Sigma_{t^*}}$ is achieved by $t^*=\tilde{\tau}$. Apply (\ref{bddX1}) on $[\tau',\tilde{\tau}]$, we get
\begin{equation*}
\begin{split}
&\int_{\Sigma_{\tilde{\tau}}}J^N_\mu\left({\Phi}\right)n^\mu_{\Sigma_{\tilde{\tau}}}\\
\leq &C\int_{\Sigma_{\tau'}} J^{N}_\mu\left({\Phi}\right)n^\mu_{\Sigma_{\tau'}}+C(\delta')^{-1}\sum_{m=0}^{1}\iint_{\mathcal R(\tau'-1,\tau+1)}r^{1+\delta}\left(\partial_{t^*}^{m}{G}\right)^2+\delta'\iint_{\mathcal R(\tau'-1,\tau+1)} K^{X_0}\left(\Phi\right)\\
&+\delta'\int_{\Sigma_{\tilde{\tau}}} J^{N}_\mu\left({\Phi}\right)n^\mu_{\Sigma_{\tilde{\tau}}}+ C\sup_{t^*\in [\tau'-1,\tau+1] }\int_{\Sigma_{t^*}\cap\{|r-3M|\leq \frac{M}{8}\}} G^2,
\end{split}
\end{equation*}
which, upon choosing $\delta'\leq\frac{1}{2}$ and subtracting the small term on both sides, gives
\begin{equation*}
\begin{split}
&\int_{\Sigma_{\tilde{\tau}}}J^N_\mu\left({\Phi}\right)n^\mu_{\Sigma_{\tilde{\tau}}}\\
\leq &C\int_{\Sigma_{\tau'}} J^{N}_\mu\left({\Phi}\right)n^\mu_{\Sigma_{\tau'}}+C(\delta')^{-1}\sum_{m=0}^{1}\iint_{\mathcal R(\tau'-1,\tau+1)}r^{1+\delta}\left(\partial_{t^*}^{m}{G}\right)^2+\delta'\iint_{\mathcal R(\tau'-1,\tau+1)} K^{X_0}\left(\Phi\right)\\
&+C\sup_{t^*\in [\tau'-1,\tau+1] }\int_{\Sigma_{t^*}\cap\{|r-3M|\leq \frac{M}{8}\}} G^2,
\end{split}
\end{equation*}
Therefore, plugging this back into (\ref{bddX1})
\begin{equation*}
\begin{split}
&\int_{\Sigma_{\tau}}J^N_\mu\left({\Phi}\right)n^\mu_{\Sigma_{\tau}}+\int_{\mathcal H(\tau',\tau)}J^N_\mu\left(\Phi\right)n^\mu_{\mathcal H^+}+\iint_{\mathcal R(\tau',\tau)\cap\{r\leq r^-_Y\}}K^{N}\left(\Phi\right)+\iint_{\mathcal R(\tau',\tau)}K^{X_0}\left(\Phi\right)\\
\leq &C\int_{\Sigma_{\tau'}} J^{N}_\mu\left({\Phi}\right)n^\mu_{\Sigma_{\tau'}}+\delta'\iint_{\mathcal R(\tau',\tau)} K^{X_0}\left(\Phi\right)+\delta'\iint_{\mathcal R(\tau'-1,\tau')\cup\mathcal R(\tau,\tau+1)} K^{X_0}\left(\Phi\right)\\
&+C\sum_{m=0}^{1}\iint_{\mathcal R(\tau'-1,\tau+1)}r^{1+\delta}\left(\partial_{t^*}^{m}{G}\right)^2+C\sup_{t^*\in [\tau'-1,\tau+1] }\int_{\Sigma_{t^*}\cap\{|r-3M|\leq \frac{M}{8}\}} G^2\\
\leq &C\int_{\Sigma_{\tau'}} J^{N}_\mu\left({\Phi}\right)n^\mu_{\Sigma_{\tau'}}+\delta'\iint_{\mathcal R(\tau',\tau)} K^{X_0}\left(\Phi\right)+C\delta'\int_{\Sigma_{\tau}}J^N_\mu\left({\Phi}\right)n^\mu_{\Sigma_{\tau}}\\
&+C\sum_{m=0}^{1}\iint_{\mathcal R(\tau'-1,\tau+1)}r^{1+\delta}\left(\partial_{t^*}^{m}{G}\right)^2+C\sup_{t^*\in [\tau'-1,\tau+1] }\int_{\Sigma_{t^*}\cap\{|r-3M|\leq \frac{M}{8}\}} G^2,\\
\end{split}
\end{equation*}
where at the last step we have used Proposition \ref{mg}. Finally, by choosing $C\delta\leq\frac{1}{2}$, we can absorb the small terms to the left hand side and achieve the conclusion of the Proposition.
\end{proof}
Notice that in the proof of Proposition \ref{bddloss}, there is a loss in derivative for $G$ because we have to integrate by parts in the region $\{|r-3M|\leq\frac{M}{8}\}$. Therefore, if $G$ is supported away from this region, we can repeat the proof without this loss. In other words, we have
\begin{proposition}\label{bddnoloss}
Suppose $G$ is supported away from $\{|r-3M|\leq\frac{M}{8}\}$. Then
\begin{equation*}
\begin{split}
&\int_{\Sigma_{\tau}}J^N_\mu\left(\Phi\right)n^\mu_{\Sigma_{\tau}}+\int_{\mathcal H(\tau',\tau)}J^N_\mu\left(\Phi\right)n^\mu_{\mathcal H^+}+\iint_{\mathcal R(\tau',\tau)\cap\{r\leq r^-_Y\}}K^{N}\left(\Phi\right)+\iint_{\mathcal R(\tau',\tau)}K^{X_0}\left(\Phi\right)\\
\leq &C\left(\int_{\Sigma_{\tau'}} J^{N}_\mu\left(\Phi\right)n^\mu_{\Sigma_{\tau'}}+\sum_{m=0}^{1}\iint_{\mathcal R(\tau'-1,\tau+1)}r^{1+\delta}G^2+\sup_{t^*\in [\tau'-1,\tau+1]}\int_{\Sigma_{t^*}\cap\{|r-3M|\leq\frac{M}{8}\}} {G}^2\right).
\end{split}
\end{equation*}
\end{proposition}
This will be useful in Section \ref{sectioncommutatorS}.

In applications, it is useful to have both ways of estimating $G$.
\begin{proposition}\label{bddcom}
Let $G=G_1+G_2$ be any way to decompose the function $G$. Then
\begin{equation*}
\begin{split}
&\int_{\Sigma_{\tau}} J^{N}_\mu\left(\Phi\right)n^\mu_{\Sigma_{\tau}} +\int_{\mathcal H(\tau',\tau)} J^{N}_\mu\left(\Phi\right)n^\mu_{\mathcal H^+} +\iint_{\mathcal R(\tau',\tau)\cap\{r\leq r^-_Y\}}K^{N}\left(\Phi\right)+\iint_{\mathcal R(\tau',\tau)}K^{X_0}\left(\Phi\right)\\
\leq &C\left(\int_{\Sigma_{\tau'}} J^{N}_\mu\left(\Phi\right)n^\mu_{\Sigma_{\tau'}}+\left(\int_{\tau'-1}^{\tau+1}\left(\int_{\Sigma_{t^*}} G_1^2\right)^{\frac{1}{2}}dt^*\right)^2+\iint_{\mathcal R(\tau'-1,\tau+1)}G_1^2\right.\\
&\left.+\sum_{m=0}^{1}\iint_{\mathcal R(\tau'-1,\tau+1)}r^{1+\delta}\left(\partial_{t^*}^{m}G_2\right)^2+\sup_{t^*\in [\tau'-1,\tau+1]}\int_{\Sigma_{t^*}\cap\{|r-3M|\leq\frac{M}{8}\}} G_2^2\right).
\end{split}
\end{equation*}
\end{proposition}

In the above estimates, only the function $\Phi$ and its $\partial_r$ derivative can be estimated without a loss around the trapped set. To estimate the other derivatives, we need to commute with the Killing vector field $\partial_{t^*}$.
\begin{proposition}\label{bddcom1}
\begin{equation*}
\begin{split}
&\iint_{\mathcal R(\tau',\tau)}K^{X_1}\left(\Phi\right)\\
\leq& C\left(\sum_{m=0}^1\int_{\Sigma_{\tau'}} J^{N}_\mu\left(\partial_{t^*}^{m}\Phi\right)n^\mu_{\Sigma_{\tau'}}+\sum_{m=0}^1\left(\int_{\tau'-1}^{\tau+1}\left(\int_{\Sigma_{t^*}} \left(\partial_{t^*}^{m}G_1\right)^2\right)^{\frac{1}{2}}dt^*\right)^2+\sum_{m=0}^1\iint_{\mathcal R(\tau'-1,\tau+1)}\left(\partial_{t^*}^{m}G_1\right)^2\right.\\
&\left.+\sum_{m=0}^{2}\iint_{\mathcal R(\tau'-1,\tau+1)}r^{1+\delta}\left(\partial_{t^*}^{m}G_2\right)^2+\sup_{t^*\in [\tau'-1,\tau+1]}\sum_{m=0}^1\int_{\Sigma_{t^*}\cap\{|r-3M|\leq\frac{M}{8}\}} \left(\partial_{t^*}^{m}G_2\right)^2\right).
\end{split}
\end{equation*}
\end{proposition}

\begin{proof}
Using Proposition \ref{bddcom} and the fact that $\partial_{t^*}$ is Killing, we immediately have the following estimate for $\partial_{t^*}\Phi$:
\begin{equation*}
\begin{split}
&\iint_{\mathcal R(\tau',\tau)}r^{-3-\delta}\left(\partial_{t^*}\Phi\right)^2\\
\leq& C\left(\int_{\Sigma_{\tau'}} J^{N}_\mu\left(\partial_{t^*}\Phi\right)n^\mu_{\Sigma_{\tau'}}+\left(\int_{\tau'-1}^{\tau+1}\left(\int_{\Sigma_{t^*}} \left(\partial_{t^*}G_1\right)^2\right)^{\frac{1}{2}}dt^*\right)^2+\iint_{\mathcal R(\tau'-1,\tau+1)}\left(\partial_{t^*}G_1\right)^2\right.\\
&\left.+\sum_{m=1}^{2}\iint_{\mathcal R(\tau'-1,\tau+1)}r^{1+\delta}\left(\partial_{t^*}^{m}G_2\right)^2+\sup_{t^*\in [\tau'-1,\tau+1]}\int_{\Sigma_{t^*}\cap\{|r-3M|\leq\frac{M}{8}\}} \left(\partial_{t^*}G_2\right)^2\right).
\end{split}
\end{equation*}
This would allow us to estimate all derivatives of $\Phi$ except for the fact that the estimates for the angular derivatives of $\Phi$ degenerate around $r=3M$:
\begin{equation*}
\begin{split}
&\iint_{\mathcal R(\tau',\tau)}\left(r^{-1-\delta}\mathbbm{1}_{\{|r- 3M|\geq \frac{M}{8}\}}|\nabb\Phi|^2+r^{-1-\delta}\left(\partial_r\Phi\right)^2+r^{-1-\delta}\left(\partial_{t^*}\Phi\right)^2+r^{-3-\delta}\Phi^2\right)\\
\leq& C\left(\sum_{m=0}^1\int_{\Sigma_{\tau'}} J^{N}_\mu\left(\partial_{t^*}^{m}\Phi\right)n^\mu_{\Sigma_{\tau'}}+\sum_{m=0}^1\left(\int_{\tau'-1}^{\tau+1}\left(\int_{\Sigma_{t^*}} \left(\partial_{t^*}^{m}G_1\right)^2\right)^{\frac{1}{2}}dt^*\right)^2+\sum_{m=0}^1\iint_{\mathcal R(\tau'-1,\tau+1)}\left(\partial_{t^*}^{m}G_1\right)^2\right.\\
&\left.+\sum_{m=0}^{2}\iint_{\mathcal R(\tau'-1,\tau+1)}r^{1+\delta}\left(\partial_{t^*}^{m}G_2\right)^2+\sup_{t^*\in [\tau'-1,\tau+1]}\sum_{m=0}^1\int_{\Sigma_{t^*}\cap\{|r-3M|\leq\frac{M}{8}\}} \left(\partial_{t^*}^{m}G_2\right)^2\right).
\end{split}
\end{equation*}
We now use this known estimate and construct another vector field to control the angular derivatives in the region $r\sim 3M$. The argument is simple because the estimate is only local. Take $f_{an}(r)$ to be compactly support in $3M-\frac{M}{4}\leq r \leq 3M+\frac{M}{4}$ and identically equal to $-1$ in $3M-\frac{M}{8}\leq r \leq 3M+\frac{M}{8}$. If we consider $X_{an}=f_{an}(r)\partial_{r^*}$ in Schwarzschild spacetime, we get that the coefficient in front of the terms with angular derivatives is $\frac{\mu }{2r}$, which is bounded below in $3M-\frac{M}{8}\leq r \leq 3M+\frac{M}{8}$. In other words, one gets an estimate of the following form:
\begin{equation}\label{XanS}
\begin{split}
&\iint_{\mathcal R(\tau',\tau)}r^{-1-\delta}\mathbbm{1}_{\{|r- 3M|\leq \frac{M}{8}\}}|\nabb\Phi|^2\\
\leq &C\left(\iint_{\mathcal R(\tau',\tau)}\left(r^{-1-\delta}\mathbbm{1}_{\{|r- 3M|\geq \frac{M}{8}\}}|\nabb\Phi|^2+r^{-1-\delta}\left(\partial_r\Phi\right)^2+r^{-1-\delta}\left(\partial_{t^*}\Phi\right)^2+r^{-3-\delta}\Phi^2\right)\right.\\
&\left.+\int_{\Sigma_\tau}J^N_\mu\left(\Phi\right)n^\mu_{\Sigma_\tau}+\int_{\Sigma_{\tau'}}J^N_\mu\left(\Phi\right)n^\mu_{\Sigma_{\tau'}}+\iint_{\mathcal R(\tau',\tau)}\left(|\partial_r\Phi|+|r^{-1}\Phi|\right)|G|\right.\\
&\left.+\iint_{\mathcal R(\tau',\tau)}G^2\right).
\end{split}
\end{equation}
Using a stability argument, (\ref{XanS}) would hold also on Kerr spacetimes. One easily checks that the terms with $G$ on the right hand side can be estimated in the same manner as before. Hence, the Proposition can be proved by applying Proposition \ref{bddcom}.
\end{proof}

\section{Vector Field Multiplier $Z$ and Decay of Non-degenerate Energy}\label{sectionZ}
We follow the definition of $Z$ in \cite{DRL}. Let $Z=u^2\underline{L}+v^2 L$, where $u$, $v$ are the Schwarzschild coordinates $u=\frac{1}{2}\left(t-r^*_S\right)$, $v=\frac{1}{2}\left(t+r^*_S\right)$, $\underline{L}=\partial_u$ and $L=2V-\underline{L}$, where $V=\partial_{t^*}+\chi(r)\frac{a}{2Mr_+}\partial_{\phi^*}$ with $\chi$ being a cutoff function that is identically $1$ for $r\leq r^-_Y-\frac{r^-_Y-r_+}{2}$ and is compactly supported in $\{r\leq r^-_Y-\frac{r^-_Y-r_+}{4}\}$. Notice that with this definition, $V$ is Killing except in the set $\{r^-_Y-\frac{r^-_Y-r_+}{2}\leq r\leq r^-_Y-\frac{r^-_Y-r_+}{4}\}$. Let $w^Z=\frac{4tr^*_S\left(1-\mu \right)}{r}$. Notice also that while $u\to\infty$ as one approaches the event horizon, $Z$ is continuous up to the event horizon due to the following (However, $Z$ is not $C^1$ and hence its deformation tensor is not continuous up to the event horizon):

\begin{proposition}\label{Znearhorizon}
In the Kerr $(t^*,r,\theta,\phi^*)$ coordinates,
$$\underline{L}=\left(1-\mu\right)\partial_{t^*}-\left(1-\mu\right)\left(\frac{2r_s-2M}{2r-2M}\right)\partial_r.$$
In the null frame near the event horizon in Section \ref{geometry}.3, we can write
$$\underline{L}=\underline{L}^{\hat{V}}\hat{V}+\underline{L}^{\hat{Y}}\hat{Y}+\underline{L}^A E_A,\quad\mbox{ where }|\underline{L}^\alpha|\leq C(1-\mu ).$$
\end{proposition}

Heuristically, we want to show that in the region $\{r\geq r^-_Y\}$,
$$\int_{\Sigma_\tau\cap\{r\geq r^-_Y\}} J^{Z,w^Z}_\mu\left(\Phi\right)n^\mu_{\Sigma_\tau}\geq 0.$$
Moreover, we would like to have
$$\int_{\Sigma_\tau\cap\{r\geq r^-_Y\}} J^{Z,w^Z}_\mu\left(\Phi\right)n^\mu_{\Sigma_\tau}\geq \int_{\Sigma_\tau\cap\{r\geq r^-_Y\}} u^2\left(\underline L\Phi\right)^2+v^2\left( L\Phi\right)^2+\left(u^2+v^2\right)|\nabb\Phi|^2+\left(\frac{u^2+v^2}{r^2}\right)\Phi^2$$
These are true modulo some error terms that can be controlled:
\begin{proposition}\label{Zlowerbound}
\begin{equation*}
\begin{split}
&\int_{\Sigma_\tau\cap\{r\geq r^-_Y\}} u^2\left(\underline{L}\Phi\right)^2+v^2\left(L\Phi\right)^2+\left(u^2+v^2\right)|\nabb\Phi|^2+\left(\frac{u^2+v^2}{r^2}\right)\Phi^2\\
\leq &C\int_{\Sigma_\tau} J^{Z,w^Z}_\mu\left(\Phi\right)n^\mu_{\Sigma_\tau}+C \int_{\Sigma_\tau} J^{N}_\mu\left(\Phi\right)n^\mu_{\Sigma_\tau}+C^2 \tau^2\int_{\Sigma_\tau\cap\{r\leq r^-_Y\}} J^{N}_\mu\left(\Phi\right)n^\mu_{\Sigma_\tau}.
\end{split}
\end{equation*}
\end{proposition}
\begin{proof}
The proof is analogous to that in Minkowski spacetime (see \cite{Morawetz}) and Schwarzschild spacetime (see \cite{DRS}).
Recall that on Schwarzschild spacetime, on a $t$ slice \cite{DRS}:
\begin{equation*}
\begin{split}
&\left(J^{Z,w^Z}_S\right)_\mu\left(\Phi\right)n^\mu_{\Sigma_{t}}\\
=&\frac{1}{\sqrt{1-\mu }}\left(v^2\left(L\Phi\right)^2+u^2\left(\underline L\Phi\right)^2+\left(u^2+v^2\right)|\nabb\Phi|^2+\frac{2tr^*\left(1-\mu\right)}{r}\Phi\partial_{t}\Phi-\frac{r^*\left(1-\mu\right)}{r}\Phi^2\right).
\end{split}
\end{equation*}
Now, since $t$, $r^*$ are stable under perturbation on $\{r\geq r^-_Y-\left(r^-_Y-r_+\right)/4\}$, we have, on this set:
\begin{equation*}
\begin{split}
&\left(J^{Z,w^Z}_K\right)_\mu\left(\Phi\right)n^\mu_{\Sigma_{\tau}}\\
\geq &\frac{1}{\sqrt{1-\mu }}\left(v^2\left(L\Phi\right)^2+u^2\left(\underline L\Phi\right)^2+\left(u^2+v^2\right)|\nabb\Phi|^2+\frac{2tr^*\left(1-\mu\right)}{r}\Phi\partial_{t^*}\Phi-\frac{r^*\left(1-\mu\right)}{r}\Phi^2\right)\\
&-C\epsilon r^{-2}\left(\left(u^2+v^2\right)\left(\nabla\Phi\right)^2+t^*\Phi^2\right).
\end{split}
\end{equation*}

We now cutoff $\Phi$: define $\hat\Phi$ so that it is supported in $\{r\geq r^-_Y-\left(r^-_Y-r_+\right)/4\}$ and equals $\Phi$ in $\{r\geq r^-_Y\}$. All the error terms arising from the cutoff will be controlled using the red-shift vector field.
\begin{equation*}
\begin{split}
&\int_{\Sigma_\tau} J^{Z,w^Z}_\mu\left(\hat\Phi\right)n^\mu_{\Sigma_\tau}\\
\geq &\int_{\Sigma_\tau\cap\{r\geq r^-_Y-\left(r^-_Y-r_+\right)/4\}}\frac{1}{\sqrt{1-\mu }}\left(v^2\left(L\hat\Phi\right)^2+u^2\left(\underline L\hat\Phi\right)^2+\left(u^2+v^2\right)|\nabb\hat\Phi|^2\right)\\
&+\frac{2tr^*_S\left(1-\mu\right)^\frac{1}{2}}{r}\hat\Phi\partial_{t^*}\hat\Phi-\frac{r^*_S\left(1-\mu\right)^\frac{1}{2}}{r}\hat\Phi^2-C\epsilon r^{-2}\left(\left(u^2+v^2\right)\left(\nabla\hat\Phi\right)^2+t^*\hat\Phi^2\right).
\end{split}
\end{equation*}
The term $$\int_{\Sigma_\tau\cap\{r\geq r^-_Y-\left(r^-_Y-r_+\right)/4\}} \frac{2tr^*_S\left(1-\mu\right)^\frac{1}{2}}{r}\hat\Phi\partial_{t^*}\hat\Phi$$ is to be handled by two different integration by parts. Recall that \cite{DRS} on Schwarzschild spacetimes we have
$$t\partial_t\hat\Phi=vL\hat\Phi+u\underline{L}\hat\Phi-r^*_S\partial_{r^*_S}\hat\Phi, \quad\mbox{and}$$
$$t\partial_t\hat\Phi=\frac{t}{r^*_S}\left(vL\hat\Phi-u\underline{L}\hat\Phi\right)-\frac{t^2}{r^*_S}\partial_{r^*_S}\hat\Phi.$$
Therefore, upon integrating by parts, we have on Schwarzschild spacetimes that
\begin{equation*}
\begin{split}
&\int_{\Sigma_\tau\cap\{r\geq r^-_Y-\left(r^-_Y-r_+\right)/4\}}\frac{tr^*_S\left(1-\mu\right)^{\frac{1}{2}}}{r}\hat\Phi\partial_{t^*}\hat\Phi\\
=&\int_{\Sigma_\tau\cap\{r\geq r^-_Y-\left(r^-_Y-r_+\right)/4\}}\left(\left(1-\mu\right)r^2\frac{r^*_S}{r}\left(vL\hat\Phi+u\underline L\hat\Phi\right)\hat\Phi+\frac{1}{2}\partial_{r^*_S}\left(\left(1-\mu\right)r\left(r^*\right)^2\right)\hat\Phi^2\right)d\theta \, d\phi\, dr^*\\
=&\int_{\Sigma_\tau\cap\{r\geq r^-_Y-\left(r^-_Y-r_+\right)/4\}}\left(1-\mu\right)r^2\left(\frac{r^*_S}{r}\left(vL\hat\Phi+u\underline L\hat\Phi\right)\hat\Phi+\left(\frac{1}{2}\frac{\left(r^*_S\right)^2}{r^2}+\frac{r^*_S}{r}\right)\hat\Phi^2\right)d\theta \, d\phi\, dr^*.
\end{split}
\end{equation*}
Notice that in the above equation, we suppressed the volume form in our notations in the first line, while when we write in coordinates as in the second and the third line, we write out the volume form explicitly. Alternatively, we have
\begin{equation*}
\begin{split}
&\int_{\Sigma_\tau\cap\{r\geq r^-_Y-\left(r^-_Y-r_+\right)/4\}}\frac{tr^*_S\left(1-\mu\right)^{\frac{1}{2}}}{r}\hat\Phi\partial_{t^*}\hat\Phi\\
=&\int_{\Sigma_\tau\cap\{r\geq r^-_Y-\left(r^-_Y-r_+\right)/4\}}\left(\left(1-\mu\right)r^2\frac{t^*}{r}\left(vL\hat\Phi-u\underline L\hat\Phi\right)\hat\Phi+\frac{1}{2}\partial_{r^*}\left((1-\mu )r(t^*)^2\right)\hat\Phi^2\right)d\theta \, d\phi\, dr^*\\
=&\int_{\Sigma_\tau\cap\{r\geq r^-_Y-\left(r^-_Y-r_+\right)/4\}}\left(1-\mu\right)r^2\left(\frac{t^*}{r}\left(vL\hat\Phi-u\underline L\hat\Phi\right)\hat\Phi+\frac{1}{2}\frac{\left(t^*\right)^2}{r^2}\hat\Phi^2\right)d\theta \, d\phi\, dr^*.
\end{split}
\end{equation*}
We would like to imitate this integration by parts on Kerr spacetimes. We notice that on the domain of integration we have
\begin{equation}\label{byparts1}
t\partial_t\hat\Phi=vL\hat\Phi+u\underline{L}\hat\Phi-r^*_S\partial_{r^*_S}\hat\Phi, \quad\mbox{and}
\end{equation}
\begin{equation}\label{byparts2}
t\partial_t\hat\Phi=\frac{t}{r^*_S}\left(vL\hat\Phi-u\underline{L}\hat\Phi\right)-\frac{t^2}{r^*_S}\partial_{r^*_S}\hat\Phi.
\end{equation}
The volume form on a constant $t^*$ slice on a Kerr spacetime is close to that on a Schwarzschild spacetime, including in the region being considered. In other words, for $r\geq r^-_Y-\left(r^-_Y-r_+\right)/4$,
$$dVol_{\Sigma_\tau}=\left(r^2 \left(1-\mu\right)^{-\frac{1}{2}} + O_1(\epsilon)\right) dr dx^A dx^B .$$
Moreover, for $r\geq r^-_Y-\left(r^-_Y-r_+\right)/4$,
$$\partial_{r^*_S}=\left((1-\mu )+O_1(\epsilon r^{-2})\right)\partial_r.$$
Therefore, using (\ref{byparts1}), we have
\begin{equation*}
\begin{split}
&\int_{\Sigma_\tau\cap\{r\geq r^-_Y-\left(r^-_Y-r_+\right)/4\}}\frac{tr^*_S\left(1-\mu\right)^\frac{1}{2}}{r}\hat\Phi\partial_{t^*}\hat\Phi\\
=&\int_{\Sigma_\tau\cap\{r\geq r^-_Y-\left(r^-_Y-r_+\right)/4\}}\left(\left(rr^*_S+O(\epsilon )\right)\left(vL\hat\Phi+u\underline L\hat\Phi\right)\hat\Phi+\left(\frac{1}{2}\partial_{r}\left((1-\mu )r\left(r^*_S\right)^2\right)+O(\epsilon )\right)\hat\Phi^2\right)dr dx^Adx^B\\
=&\int_{\Sigma_\tau\cap\{r\geq r^-_Y-\left(r^-_Y-r_+\right)/4\}}\left(1-\mu\right)^\frac{1}{2}\left(\left(\frac{r^*_S}{r}+O(\epsilon r^{-2})\right)\left(vL\hat\Phi+u\underline L\hat\Phi\right)\hat\Phi+\left(\frac{1}{2}\frac{\left(r^*_S\right)^2}{r^2}+\frac{r^*_S}{r}+O(\epsilon r^{-2})\right)\hat\Phi^2\right).
\end{split}
\end{equation*}
Alternatively, we can integrate by parts after using (\ref{byparts2}):
\begin{equation*}
\begin{split}
&\int_{\Sigma_\tau\cap\{r\geq r^-_Y-\left(r^-_Y-r_+\right)/4\}}\frac{tr^*_S\left(1-\mu\right)^\frac{1}{2}}{r}\hat\Phi\partial_{t^*}\hat\Phi\\
=&\int_{\Sigma_\tau\cap\{r\geq r^-_Y-\left(r^-_Y-r_+\right)/4\}}\left(\left(rt^*+O(\epsilon )\right)\left(vL\hat\Phi-u\underline L\hat\Phi\right)\hat\Phi+\left(\frac{1}{2}\partial_{r}\left((1-\mu )r(t^*)^2\right)+O(\epsilon )\right)\hat\Phi^2\right)\\
=&\int_{\Sigma_\tau\cap\{r\geq r^-_Y-\left(r^-_Y-r_+\right)/4\}}\left(1-\mu\right)^\frac{1}{2}\left(\left(\frac{t^*}{r}+O(\epsilon r^{-2})\right)\left(vL\hat\Phi-u\underline L\hat\Phi\right)\hat\Phi+\left(\frac{1}{2}\frac{\left(t^*\right)^2}{r^2}+O(\epsilon r^{-2})\right)\hat\Phi^2\right).
\end{split}
\end{equation*}
Therefore, we have
\begin{equation}\label{Zcontrol}
\begin{split}
&\int_{\Sigma_\tau} J^{Z,w^Z}_\mu\left(\hat\Phi\right)n^\mu_{\Sigma_\tau}\\
\geq &\int_{\Sigma_\tau\cap\{r\geq r^-_Y-\left(r^-_Y-r_+\right)/4\}}\frac{1}{\sqrt{1-\mu }}\left(v^2\left(L\hat\Phi\right)^2+u^2\left(\underline L\hat\Phi\right)^2+\left(u^2+v^2\right)|\nabb\hat\Phi|^2\right)\\
&+\frac{2tr^*_S\left(1-\mu\right)^\frac{1}{2}}{r}\hat\Phi\partial_{t^*}\hat\Phi-\frac{r^*_S\left(1-\mu\right)^\frac{1}{2}}{r}\hat\Phi^2-C\epsilon r^{-2}\left(\left(u^2+v^2\right)\left(D\hat\Phi\right)^2+t^*\hat\Phi^2\right)\\
\geq &\int_{\Sigma_\tau\cap\{r\geq r^-_Y-\left(r^-_Y-r_+\right)/4\}}\frac{1}{\sqrt{1-\mu }}\left(v^2\left(L\hat\Phi\right)^2+u^2\left(\underline L\hat\Phi\right)^2+\left(u^2+v^2\right)|\nabb\hat\Phi|^2\right)\\
&+\frac{r^*_S\left(1-\mu\right)^\frac{1}{2}}{r}\left(vL\hat\Phi+u\underline L\hat\Phi\right)\hat\Phi+\frac{1}{2}\frac{\left(r^*_S\right)^2\left(1-\mu\right)^\frac{1}{2}}{r^2}\hat\Phi^2\\
&+\frac{t^*\left(1-\mu\right)^\frac{1}{2}}{r}\left(vL\hat\Phi-u\underline L\hat\Phi\right)\hat\Phi+\frac{1}{2}\frac{\left(t^*\right)^2}{r^2}\hat\Phi^2-C\epsilon r^{-2}\left(\left(u^2+v^2\right)\left(D\hat\Phi\right)^2+t^*\hat\Phi^2\right)\\
\geq &c\left(\int_{\Sigma_\tau\cap\{r\geq r^-_Y-\left(r^-_Y-r_+\right)/4\}}\mu\left(\left(vL\hat\Phi+u\underline L\hat\Phi\right)^2+\left(vL\hat\Phi-u\underline L\hat\Phi\right)^2\right)\right.\\
&+\left(1-\mu\right)\left(\left(vL\hat\Phi+u\underline L\hat\Phi+\frac{r^*_S}{r}\hat\Phi\right)^2+\left(vL\hat\Phi-u\underline L\hat\Phi+\frac{t^*}{r}\hat\Phi\right)^2+2\left(u^2+v^2\right)|\nabb\hat\Phi|^2\right)\\
&\left.-C\epsilon r^{-2}\left(\left(u^2+v^2\right)\left(D\hat\Phi\right)^2+t^*\hat\Phi^2\right)\right),
\end{split}
\end{equation}
where the last line is obtained by first completing the square and using $c\leq1-\mu\leq C$ in this region of $r$.
Let us for now ignore the error term and look at the other terms (which are manifestly positive). With exactly the same argument as in \cite{DRS}, we have that these positive terms provide good estimates:
\begin{equation*}
\begin{split}
&\left(\int_{\Sigma_\tau\cap\{r\geq r^-_Y-\left(r^-_Y-r_+\right)/4\}}\mu\left(\left(vL\hat\Phi+u\underline L\hat\Phi\right)^2+\left(vL\hat\Phi-u\underline L\hat\Phi\right)^2\right)\right.\\
&\left.+\left(1-\mu\right)\left(\left(vL\hat\Phi+u\underline L\hat\Phi+\frac{r^*_S}{r}\hat\Phi\right)^2+\left(vL\hat\Phi-u\underline L\hat\Phi+\frac{t^*}{r}\hat\Phi\right)^2+2\left(u^2+v^2\right)|\nabb\hat\Phi|^2\right)\right)\\
\geq &c\int_{\Sigma_\tau\cap\{r\geq r^-_Y-\left(r^-_Y-r_+\right)/4\}}\left(v^2\left(L\hat\Phi\right)^2+u^2\left(\underline L\hat\Phi\right)^2+\left(u^2+v^2\right)|\nabb\hat\Phi|^2+\frac{u^2+v^2}{r^2}\hat\Phi^2\right)\\
&+C\epsilon\int_{\Sigma_\tau\cap\{r\geq r^-_Y-\left(r^-_Y-r_+\right)/4\}} r^{-2}\left(\left(u^2+v^2\right)\left(D\hat\Phi\right)^2+t^*\hat\Phi^2\right).
\end{split}
\end{equation*}
We refer the reader to \cite{DRS} for the proof.
This together with $J^N_\mu\left(\hat\Phi\right)n^\mu_{\Sigma_\tau}$ bound the error term in (\ref{Zcontrol}):
\begin{equation*}
\begin{split}
&\int_{\Sigma_\tau\cap\{r\geq r^-_Y-\left(r^-_Y-r_+\right)/4\}} r^{-2}\left(\left(u^2+v^2\right)\left(D\hat\Phi\right)^2+\tau\Phi^2\right)\\
\leq &C\int_{\Sigma_\tau\cap\{r\geq r^-_Y-\left(r^-_Y-r_+\right)/4\}}\left(v^2\left(L\hat\Phi\right)^2+u^2\left(\underline L\hat\Phi\right)^2+\left(u^2+v^2\right)|\nabb\hat\Phi|^2+\frac{u^2+v^2}{r^2}\hat\Phi^2\right)\\
&+C\int_{\Sigma_\tau\cap\{r\geq \frac{\tau}{4}\}}\left(\underline{L}\hat\Phi\right)^2\\
\leq &C\int_{\Sigma_\tau\cap\{r\geq r^-_Y-\left(r^-_Y-r_+\right)/4\}}\left(v^2\left(L\hat\Phi\right)^2+u^2\left(\underline L\hat\Phi\right)^2+\left(u^2+v^2\right)|\nabb\hat\Phi|^2+\frac{u^2+v^2}{r^2}\hat\Phi^2\right)\\
&+C\int_{\Sigma_\tau}J^N_\mu\left(\hat\Phi\right)n^\mu_{\Sigma_\tau}\\
\leq &C\left(\int_{\Sigma_\tau\cap\{r\geq r^-_Y-\left(r^-_Y-r_+\right)/4\}}\mu\left(\left(vL\hat\Phi+u\underline L\hat\Phi\right)^2+\left(vL\hat\Phi-u\underline L\hat\Phi\right)^2\right)\right.\\
&\left.+\left(1-\mu\right)\left(\left(vL\hat\Phi+u\underline L\hat\Phi+\frac{r^*_S}{r}\hat\Phi\right)^2+\left(vL\hat\Phi-u\underline L\hat\Phi+\frac{t^*}{r}\hat\Phi\right)^2+2\left(u^2+v^2\right)|\nabb\hat\Phi|^2\right)\right)\\
&+C\int_{\Sigma_\tau}J^N_\mu\left(\hat\Phi\right)n^\mu_{\Sigma_\tau}.
\end{split}
\end{equation*}
Therefore, if $\epsilon$ is chosen to be small enough, then (\ref{Zcontrol}) implies that
\begin{equation}\label{Zhatcontrol}
\begin{split}
&\int_{\Sigma_\tau} J^{Z,w^Z}_\mu\left(\hat\Phi\right)n^\mu_{\Sigma_\tau}+\int_{\Sigma_\tau}J^N_\mu\left(\hat\Phi\right)n^\mu_{\Sigma_\tau}\\
\geq &c\int_{\Sigma_\tau\cap\{r\geq r^-_Y-\left(r^-_Y-r_+\right)/4\}}\left(v^2\left(L\hat\Phi\right)^2+u^2\left(\underline L\hat\Phi\right)^2+\left(u^2+v^2\right)|\nabb\hat\Phi|^2+\frac{u^2+v^2}{r^2}\hat\Phi^2\right).
\end{split}
\end{equation}
We note that $c$ here is independent of the choice of $r^-_Y$. With this bound we would like to estimate $\int_{\mathbb S^2} \Phi\left(\tau, r\right)^2$. Using (\ref{Zhatcontrol}), there exists a $\tilde{r}\in [r^-_Y,r^-_Y+1]$ such that
$$\int_{\mathbb S^2} \Phi\left(\tau, \tilde{r}\right)^2 =\int_{\mathbb S^2} \hat\Phi\left(\tau, \tilde{r}\right)^2\leq C\tau^{-2}\left(\int_{\Sigma_\tau} J^{Z,w^Z}_\mu\left(\hat\Phi\right)n^\mu_{\Sigma_\tau}+\int_{\Sigma_\tau}J^N_\mu\left(\hat\Phi\right)n^\mu_{\Sigma_\tau}\right).$$
Then for every $r\in [r_+,r^-_Y]$, since
$$\Phi\left(\tau, \tilde{r}\right)- \Phi\left(\tau, r\right)=\int_r^{\tilde{r}} \partial_r\Phi dr,$$
we have
\begin{equation}\label{0thorder}
\begin{split}
\int_{\mathbb S^2} \Phi\left(\tau, r\right)^2\leq &\int_{\mathbb S^2} \Phi\left(\tau, \tilde{r}\right)^2+\left(\tilde{r}-r\right)\int_{\Sigma_\tau\cap[r,\tilde{r}]} J^N_\mu\left(\Phi\right)n^\mu_{\Sigma_\tau}\\
\leq &C\int_{\Sigma_\tau\cap\{r\leq r^-_Y\}}J^N_\mu\left(\Phi\right)n^\mu_{\Sigma_\tau}+C\tau^{-2}\left(\int_{\Sigma_\tau} J^{Z,w^Z}_\mu\left(\hat\Phi\right)n^\mu_{\Sigma_\tau}+\int_{\Sigma_\tau}J^N_\mu\left(\hat\Phi\right)n^\mu_{\Sigma_\tau}\right)
\end{split}
\end{equation}
Now we need to obtain estimates for $\Phi$ from that for $\hat\Phi$. It is obvious that
\begin{equation*}
\begin{split}
&\int_{\Sigma_\tau\cap\{r\geq r^-_Y\}}\left(v^2\left(L\Phi\right)^2+u^2\left(\underline L\Phi\right)^2+\left(u^2+v^2\right)|\nabb\Phi|^2+\frac{u^2+v^2}{r^2}\Phi^2\right)\\
\leq &\int_{\Sigma_\tau\cap\{r\geq r^-_Y\}}\left(v^2\left(L\hat\Phi\right)^2+u^2\left(\underline L\hat\Phi\right)^2+\left(u^2+v^2\right)|\nabb\hat\Phi|^2+\frac{u^2+v^2}{r^2}\hat\Phi^2\right)\\
\leq &\int_{\Sigma_\tau\cap\{r\geq r^-_Y-\left(r^-_Y-r_+\right)/4\}}\left(v^2\left(L\hat\Phi\right)^2+u^2\left(\underline L\hat\Phi\right)^2+\left(u^2+v^2\right)|\nabb\hat\Phi|^2+\frac{u^2+v^2}{r^2}\hat\Phi^2\right),
\end{split}
\end{equation*}
and
\begin{equation}\label{r-closeZ}
\begin{split}
&\int_{\Sigma_\tau} J^{Z,w^Z}_\mu\left(\hat\Phi\right)n^\mu_{\Sigma_\tau}+\int_{\Sigma_\tau}J^N_\mu\left(\hat\Phi\right)n^\mu_{\Sigma_\tau}\\
\leq& \int_{\Sigma_\tau} J^{Z,w^Z}_\mu\left(\Phi\right)n^\mu_{\Sigma_\tau}+\int_{\Sigma_\tau}J^N_\mu\left(\Phi\right)n^\mu_{\Sigma_\tau}+C\tau^2\int_{\Sigma_\tau\cap\{r\leq r^-_Y\}}J^N_\mu\left(\Phi\right)n^\mu_{\Sigma_\tau}+C\tau^2\int_{\Sigma_\tau\cap\{r\leq r^-_Y\}}\Phi^2\\
&\quad\quad\mbox{where we have used Proposition \ref{Znearhorizon} to show that the $u^2$ factor comes with a factor of $1-\mu$}\\
\leq& \int_{\Sigma_\tau} J^{Z,w^Z}_\mu\left(\Phi\right)n^\mu_{\Sigma_\tau}+\int_{\Sigma_\tau}J^N_\mu\left(\Phi\right)n^\mu_{\Sigma_\tau}+C\tau^2\int_{\Sigma_\tau\cap\{r\leq r^-_Y\}}J^N_\mu\left(\Phi\right)n^\mu_{\Sigma_\tau}\\
&+C\left(r^-_Y-r_+\right)\left(\int_{\Sigma_\tau} J^{Z,w^Z}_\mu\left(\hat\Phi\right)n^\mu_{\Sigma_\tau}+\int_{\Sigma_\tau}J^N_\mu\left(\hat\Phi\right)n^\mu_{\Sigma_\tau}\right)\\
\leq& \int_{\Sigma_\tau} J^{Z,w^Z}_\mu\left(\Phi\right)n^\mu_{\Sigma_\tau}+\int_{\Sigma_\tau}J^N_\mu\left(\Phi\right)n^\mu_{\Sigma_\tau}+C\tau^2\int_{\Sigma_\tau\cap\{r\leq r^-_Y\}}J^N_\mu\left(\Phi\right)n^\mu_{\Sigma_\tau}\\
&+\frac{1}{2}\left(\int_{\Sigma_\tau} J^{Z,w^Z}_\mu\left(\hat\Phi\right)n^\mu_{\Sigma_\tau}+\int_{\Sigma_\tau}J^N_\mu\left(\hat\Phi\right)n^\mu_{\Sigma_\tau}\right),
\end{split}
\end{equation}
for $r^-_Y$ chosen to be sufficiently close to $r_+$.
Then
\begin{equation*}
\begin{split}
&\int_{\Sigma_\tau} J^{Z,w^Z}_\mu\left(\hat\Phi\right)n^\mu_{\Sigma_\tau}+\int_{\Sigma_\tau}J^N_\mu\left(\hat\Phi\right)n^\mu_{\Sigma_\tau}\\
\leq& \int_{\Sigma_\tau} J^{Z,w^Z}_\mu\left(\Phi\right)n^\mu_{\Sigma_\tau}+\int_{\Sigma_\tau}J^N_\mu\left(\Phi\right)n^\mu_{\Sigma_\tau}+C\tau^2\int_{\Sigma_\tau\cap\{r\leq r^-_Y\}}J^N_\mu\left(\Phi\right)n^\mu_{\Sigma_\tau}.
\end{split}
\end{equation*}
\end{proof}
\begin{remark}\label{r-fix}
From this point onward, we consider $r^-_Y$ to be fixed. We note again that $r^-_Y$ is chosen so that (\ref{definekappa}) and (\ref{r-closeZ}) hold.
\end{remark}
\begin{remark}\label{positivity}
The proof of the above proposition in particular shows that
$$\int_{\Sigma_\tau} J^{Z,w^Z}_\mu\left(\Phi\right)n^\mu_{\Sigma_\tau}+C \tau^2\int_{\Sigma_\tau\cap\{r\leq r^-_Y\}} J^{N}_\mu\left(\Phi\right)n^\mu_{\Sigma_\tau}\geq 0$$
\end{remark}

In order to use this Proposition, it is helpful to have a localized version of $\Phi$. This follows \cite{DRS}, \cite{DRL}. The idea is to use the finite speed of propagation and cutoff $\Phi$ outside the domain of dependence. Suppose we now focus on the time interval $[\tau',\tau]$. Take $\tilde{G}$ to be any smooth function agreeing with $G$ on the domain of dependence of the region $(t^*=\tau, r \leq \frac{\tau}{2})$. Let $\tilde{\Phi}\left(\tau'\right)=\chi\Phi\left(\tau'\right)$, $\partial_{t^*}\tilde{\Phi}\left(\tau'\right)=\chi\partial_{t^*}\Phi\left(\tau'\right)$, where $\chi$ is a cutoff function identically equals to 1 for $r\leq\frac{7\tau'}{10}$ and compactly supported in $r\leq\frac{9\tau'}{10}$. Notice that the region for which $\chi$ is one is inside the domain of dependence of the region $(t^*=\tau, r \leq \frac{\tau}{2})$ if $\tau'\leq\tau\leq (1.1)\tau'$. We solve for $\Box_{g_K}\tilde{\Phi}=\tilde{G}$.\\
With this definition of $\tilde{\Phi}$, we have two ways to estimate the non-degenerate energy of $\tilde{\Phi}$:
\begin{proposition}\label{cutoffbound}
\begin{equation*}
\begin{split}
\int_{\Sigma_{\tau'}} J^{N}_\mu\left(\tilde{\Phi}\right)n^\mu_{\Sigma_{\tau'}}\leq C\int_{\Sigma_{\tau'}}J^{N}_\mu\left(\Phi\right)n^\mu_{\Sigma_{\tau'}}.
\end{split}
\end{equation*}
\begin{equation*}
\begin{split}
&\int_{\Sigma_{\tau'}} J^{N}_\mu\left(\tilde{\Phi}\right)n^\mu_{\Sigma_{\tau'}}\\
\leq& C^2\int_{\Sigma_{\tau'}\cap\{r\leq r^-_Y\}}J^{N}_\mu\left(\Phi\right)n^\mu_{\Sigma_{\tau'}}+C(\tau')^{-2}\left(\int_{\Sigma_{\tau'}}J^{N}_\mu\left(\Phi\right)n^\mu_{\Sigma_{\tau'}}+\int_{\Sigma_{\tau'}}J^{Z,w^Z}_\mu\left(\Phi\right)n^\mu_{\Sigma_{\tau'}}\right).
\end{split}
\end{equation*}
\end{proposition}
\begin{proof}
The first part is an easy application of Proposition \ref{hi}
\begin{equation*}
\begin{split}
\int_{\Sigma_{\tau'}} J^{N}_\mu\left(\tilde{\Phi}\right)n^\mu_{\Sigma_{\tau'}}
\leq &C\int_{\Sigma_{\tau'}\cap\{R\leq r\leq \frac{9\tau'}{10}\}}\left(\left(D\Phi\right)^2+(\tau')^{-2}\Phi^2\right)\\
\leq &C\int_{\Sigma_{\tau'}\cap\{R\leq r\leq \frac{9\tau'}{10}\}}\left(\left(D\Phi\right)^2+r^{-2}\Phi^2\right)\\
\leq &C\int_{\Sigma_{\tau'}}J^{N}_\mu\left(\Phi\right)n^\mu_{\Sigma_{\tau'}}.
\end{split}
\end{equation*}
Following (\ref{0thorder}), we have
\begin{equation*}
\begin{split}
\int_{\Sigma_{\tau'}\cap\{r\leq r^-_Y\}}\Phi^2\leq C\left(\int_{\Sigma_{\tau'}\cap\{r\leq r^-_Y\}}J^{N}_\mu\left(\Phi\right)n^\mu_{\Sigma_{\tau'}}+\int_{\Sigma_{\tau'}\cap\{r^-_Y\leq r\leq r^+_Y\}} \Phi^2\right).
\end{split}
\end{equation*}
Using this and Proposition \ref{Zlowerbound}, we have
\begin{equation*}
\begin{split}
&\int_{\Sigma_{\tau'}} J^{N}_\mu\left(\tilde{\Phi}\right)n^\mu_{\Sigma_{\tau'}}\\
\leq &C\int_{\Sigma_{\tau'}\cap\{r\leq \frac{9\tau'}{10}\}}\left(\left(D\Phi\right)^2+(\tau')^{-2}\Phi^2\right)\\
\leq &C\int_{\Sigma_{\tau'}\cap\{r\leq r^-_Y\}}\left(\left(D\Phi\right)^2+\Phi^2\right)\\
&+C(\tau')^{-2}\int_{\Sigma_{\tau'}\cap\{r^-_Y\leq r\leq \frac{9\tau'}{10}\}}\left(u^2\left(\underline{L}\Phi\right)^2+v^2\left(L\Phi\right)^2+\left(u^2+v^2\right)|\nabb\Phi|^2+\left(\frac{u^2+v^2}{r^2}\right)\Phi^2\right)\\
\leq& C^2\int_{\Sigma_{\tau'}\cap\{r\leq r^-_Y\}}J^{N}_\mu\left(\Phi\right)n^\mu_{\Sigma_{\tau'}}+C(\tau')^{-2}\left(\int_{\Sigma_{\tau'}}J^{N}_\mu\left(\Phi\right)n^\mu_{\Sigma_{\tau'}}+\int_{\Sigma_{\tau'}}J^{Z,w^Z}_\mu\left(\Phi\right)n^\mu_{\Sigma_{\tau'}}\right).
\end{split}
\end{equation*}
\end{proof}

The cutoff procedure above would also allow us to localized the estimates for the bulk term:
\begin{proposition}\label{localization}
Let $G=G_1+G_2$ be any way to decompose the function $G$. Then for $\tau'\leq \tau\leq (1.1)\tau'$, we have
\begin{enumerate}
\item Localized Boundedness Estimate
\begin{equation*}
\begin{split}
&\int_{\Sigma_{\tau}\cap\{r\leq\frac{\tau}{2}\}} J^{N}_\mu\left(\Phi\right)n^\mu_{\Sigma_{\tau}} +\int_{\mathcal H(\tau',\tau)} J^{N}_\mu\left(\Phi\right)n^\mu_{\mathcal H^+} \\
&+\iint_{\mathcal R(\tau',\tau)\cap\{r\leq r^-_Y\}}K^{N}\left(\Phi\right)+\iint_{\mathcal R(\tau',\tau)\cap\{r\leq\frac{t^*}{2}\}}K^{X_0}\left(\Phi\right)\\
\leq &C\left(\int_{\Sigma_{\tau'}} J^{N}_\mu\left(\Phi\right)n^\mu_{\Sigma_{\tau'}}+\left(\int_{\tau'-1}^{\tau+1}\left(\int_{\Sigma_{t^*}\cap\{r\leq\frac{9t^*}{10}\}}  G_1^2\right)^{\frac{1}{2}}dt^*\right)^2\right.\\
&\left.+\iint_{\mathcal R(\tau'-1,\tau+1)\cap\{r\leq\frac{9t^*}{10}\}}  G_1^2+\sum_{m=0}^{1}\iint_{\mathcal R(\tau'-1,\tau+1)\cap\{r\leq\frac{9t^*}{10}\}}r^{1+\delta}\left(\partial_{t^*}^{m}{G_2}\right)^2\right.\\
&\left.+\sup_{t^*\in [\tau'-1,\tau+1]}\int_{\Sigma_{t^*}\cap\{|r-3M|\leq\frac{M}{8}\}\cap\{r\leq\frac{9t^*}{10}\}} {G_2}^2\right).
\end{split}
\end{equation*}
\item Localized Decay Estimate
\begin{equation*}
\begin{split}
&\int_{\Sigma_{\tau}\cap\{r\leq\frac{\tau}{2}\}} J^{N}_\mu\left(\Phi\right)n^\mu_{\Sigma_{\tau}} +\int_{\mathcal H(\tau',\tau)} J^{N}_\mu\left(\Phi\right)n^\mu_{\mathcal H^+} \\
&+\iint_{\mathcal R(\tau',\tau)\cap\{r\leq r^-_Y\}}K^{N}\left(\Phi\right)+\iint_{\mathcal R(\tau',\tau)\cap\{r\leq\frac{t^*}{2}\}}K^{X_0}\left(\Phi\right)\\
\leq& C\left(\tau^{-2}\int_{\Sigma_{\tau'}} J^{Z+N,w^Z}_\mu\left(\Phi\right)n^\mu_{\Sigma_{\tau'}}+ C\int_{\Sigma_{\tau'}\cap\{r\leq r^-_Y\}} J^{N}_\mu\left(\Phi\right)n^\mu_{\Sigma_{\tau'}}\right)\\
&+C\left(\left(\int_{\tau'-1}^{\tau+1}\left(\int_{\Sigma_{t^*}\cap\{r\leq\frac{9t^*}{10}\}}  G_1^2\right)^{\frac{1}{2}}dt^*\right)^2+\iint_{\mathcal R(\tau'-1,\tau+1)\cap\{r\leq\frac{9t^*}{10}\}}  G_1^2\right.\\
&\left.+\sum_{m=0}^{1}\iint_{\mathcal R(\tau'-1,\tau+1)\cap\{r\leq\frac{9t^*}{10}\}}r^{1+\delta}\left(\partial_{t^*}^{m}{G_2}\right)^2+\sup_{t^*\in [\tau'-1,\tau+1]}\int_{\Sigma_{t^*}\cap\{|r-3M|\leq\frac{M}{8}\}\cap\{r\leq\frac{9t^*}{10}\}} G_2^2\right).
\end{split}
\end{equation*}
\end{enumerate}
\end{proposition}
\begin{proof}
Applying Proposition \ref{bddcom} to the equation $\Box_{g_K}\tilde{\Phi}=\tilde{G}$, we have
\begin{equation*}
\begin{split}
&\int_{\Sigma_{\tau}\cap\{r\leq\frac{\tau}{2}\}} J^{N}_\mu\left(\tilde{\Phi}\right)n^\mu_{\Sigma_{\tau}} +\int_{\mathcal H(\tau',\tau)} J^{N}_\mu\left(\tilde{\Phi}\right)n^\mu_{\mathcal H^+} +\iint_{\mathcal R(\tau',\tau)\cap\{r\leq r^-_Y\}}K^{N}\left(\tilde{\Phi}\right)+\iint_{\mathcal R(\tau',\tau)\cap\{r\leq\frac{t^*}{2}\}}K^{X_0}\left(\tilde{\Phi}\right)\\
\leq &C\left(\int_{\Sigma_{\tau'}} J^{N}_\mu\left(\tilde{\Phi}\right)n^\mu_{\Sigma_{\tau'}}+\left(\int_{\tau'-1}^{\tau+1}\left(\int_{\Sigma_{t^*}} \tilde G_1^2\right)^{\frac{1}{2}}dt^*\right)^2+\iint_{\mathcal R(\tau'-1,\tau+1)}\tilde G_1^2\right.\\
&\left.+\sum_{m=0}^{1}\iint_{\mathcal R(\tau'-1,\tau+1)}r^{1+\delta}\left(\partial_{t^*}^{m}\tilde G_2\right)^2+\sup_{t^*\in [\tau'-1,\tau+1]}\int_{\Sigma_{t^*}\cap\{|r-3M|\leq\frac{M}{8}\}} \tilde G_2^2\right).
\end{split}
\end{equation*}
Since by the finite speed of propagation, $\tilde{\Phi}=\Phi$ in $\{r\leq \frac{t^*}{2}\}$, we have
\begin{equation*}
\begin{split}
&\int_{\Sigma_{\tau}\cap\{r\leq\frac{\tau}{2}\}} J^{N}_\mu\left(\Phi\right)n^\mu_{\Sigma_{\tau}} +\int_{\mathcal H(\tau',\tau)} J^{N}_\mu\left(\Phi\right)n^\mu_{\mathcal H^+} +\iint_{\mathcal R(\tau',\tau)\cap\{r\leq r^-_Y\}}K^{N}\left(\Phi\right)+\iint_{\mathcal R(\tau',\tau)\cap\{r\leq\frac{t^*}{2}\}}K^{X_0}\left(\Phi\right)\\
\leq &C\left(\int_{\Sigma_{\tau'}} J^{N}_\mu\left(\tilde\Phi\right)n^\mu_{\Sigma_{\tau'}}+\left(\int_{\tau'-1}^{\tau+1}\left(\int_{\Sigma_{t^*}} \tilde G_1^2\right)^{\frac{1}{2}}dt^*\right)^2+\iint_{\mathcal R(\tau'-1,\tau+1)}\tilde G_1^2\right.\\
&\left.+\sum_{m=0}^{1}\iint_{\mathcal R(\tau'-1,\tau+1)}r^{1+\delta}\left(\partial_{t^*}^{m}\tilde G_2\right)^2+\sup_{t^*\in [\tau'-1,\tau+1]}\int_{\Sigma_{t^*}\cap\{|r-3M|\leq\frac{M}{8}\}} \tilde G_2^2\right).
\end{split}
\end{equation*}
Now, we choose a particular $\tilde{G}$. Define $\tilde{G}$ to be $G$ for $r\leq\frac{7t^*}{10}$, and $0$ for $r\geq\frac{9t^*}{10}$. It can be easily shown that one can have the bounds $|\partial_{t^*}^m\tilde{G}|\leq C\displaystyle\sum_{k=0}^m|\left(\frac{r^*}{(t^*)^2}\right)^k\partial_{t^*}^{m-k}G|\leq C\displaystyle\sum_{k=0}^m|(t^*)^{-k}\partial_{t^*}^{m-k}G|$ for $\frac{7t^*}{10}\leq r\leq\frac{9t^*}{10}$. Therefore, we have
\begin{equation*}
\begin{split}
&\int_{\Sigma_{\tau}\cap\{r\leq\frac{\tau}{2}\}} J^{N}_\mu\left(\Phi\right)n^\mu_{\Sigma_{\tau}} +\int_{\mathcal H(\tau',\tau)} J^{N}_\mu\left(\Phi\right)n^\mu_{\mathcal H^+} +\iint_{\mathcal R(\tau',\tau)\cap\{r\leq r^-_Y\}}K^{N}\left(\Phi\right)+\iint_{\mathcal R(\tau',\tau)\cap\{r\leq\frac{t^*}{2}\}}K^{X_0}\left(\Phi\right)\\
\leq &C\left(\int_{\Sigma_{\tau'}} J^{N}_\mu\left(\tilde\Phi\right)n^\mu_{\Sigma_{\tau'}}+\left(\int_{\tau'-1}^{\tau+1}\left(\int_{\Sigma_{t^*}\cap\{r\leq\frac{9t^*}{10}\}}  G_1^2\right)^{\frac{1}{2}}dt^*\right)^2+\iint_{\mathcal R(\tau'-1,\tau+1)\cap\{r\leq\frac{9t^*}{10}\}} G_1^2\right.\\
&\left.+\sum_{m=0}^{1}\iint_{\mathcal R(\tau'-1,\tau+1)\cap\{r\leq\frac{9t^*}{10}\}}r^{1+\delta}\left(\partial_{t^*}^{m} G_2\right)^2+\sup_{t^*\in [\tau'-1,\tau+1]}\int_{\Sigma_{t^*}\cap\{|r-3M|\leq\frac{M}{8}\}}  G_2^2\right).
\end{split}
\end{equation*}
We can now conclude the Proposition using Proposition \ref{cutoffbound}.
\end{proof}

We can remove the degeneracy around $r\sim 3M$ using an extra derivative.
\begin{proposition}\label{localizationT}
Let $G=G_1+G_2$ be any way to decompose the function $G$. Then for $\tau'\leq \tau\leq (1.1)\tau'$, we have
\begin{enumerate}
\item Localized Boundedness Estimate
\begin{equation*}
\begin{split}
&\iint_{\mathcal R(\tau',\tau)\cap\{r\leq\frac{t^*}{2}\}}K^{X_1}\left(\Phi\right)\\
\leq &C\left(\sum_{m=0}^{1}\int_{\Sigma_{\tau'}} J^{N}_\mu\left(\partial_{t^*}^m\Phi\right)n^\mu_{\Sigma_{\tau'}}+\sum_{m=0}^{1}\left(\int_{\tau'-1}^{\tau+1}\left(\int_{\Sigma_{t^*}\cap\{r\leq\frac{9t^*}{10}\}}  \left(\partial_{t^*}^mG_1\right)^2\right)^{\frac{1}{2}}dt^*\right)^2\right.\\
&\left.+\sum_{m=0}^{1}\iint_{\mathcal R(\tau'-1,\tau+1)\cap\{r\leq\frac{9t^*}{10}\}}  \left(\partial_{t^*}^mG_1\right)^2+\sum_{m=0}^{2}\iint_{\mathcal R(\tau'-1,\tau+1)\cap\{r\leq\frac{9t^*}{10}\}}r^{1+\delta}\left(\partial_{t^*}^{m}{G_2}\right)^2\right.\\
&\left.+\sup_{t^*\in [\tau'-1,\tau+1]}\sum_{m=0}^{1}\int_{\Sigma_{t^*}\cap\{|r-3M|\leq\frac{M}{8}\}\cap\{r\leq\frac{9t^*}{10}\}} \left(\partial_{t^*}^m G_2\right)^2\right).
\end{split}
\end{equation*}
\item Localized Decay Estimate
\begin{equation*}
\begin{split}
&\iint_{\mathcal R(\tau',\tau)\cap\{r\leq\frac{t^*}{2}\}}K^{X_1}\left(\Phi\right)\\
\leq& C\left(\tau^{-2}\sum_{m=0}^{1}\int_{\Sigma_{\tau'}} J^{Z+N,w^Z}_\mu\left(\partial_{t^*}^m\Phi\right)n^\mu_{\Sigma_{\tau'}}+ C\sum_{m=0}^{1}\int_{\Sigma_{\tau'}\cap\{r\leq r^-_Y\}} J^{N}_\mu\left(\partial_{t^*}^m\Phi\right)n^\mu_{\Sigma_{\tau'}}\right)\\
&+C\left(\sum_{m=0}^{1}\left(\int_{\tau'-1}^{\tau+1}\left(\int_{\Sigma_{t^*}\cap\{r\leq\frac{9t^*}{10}\}}  \left(\partial_{t^*}^mG_1\right)^2\right)^{\frac{1}{2}}dt^*\right)^2+\sum_{m=0}^{1}\iint_{\mathcal R(\tau'-1,\tau+1)\cap\{r\leq\frac{9t^*}{10}\}}  \left(\partial_{t^*}^mG_1\right)^2\right.\\
&\left.+\sum_{m=0}^{2}\iint_{\mathcal R(\tau'-1,\tau+1)\cap\{r\leq\frac{9t^*}{10}\}}r^{1+\delta}\left(\partial_{t^*}^{m}{G_2}\right)^2\right.\\
&\left.+\sup_{t^*\in [\tau'-1,\tau+1]}\sum_{m=0}^{1}\int_{\Sigma_{t^*}\cap\{|r-3M|\leq\frac{M}{8}\}\cap\{r\leq\frac{9t^*}{10}\}} \left(\partial_{t^*}^m G_2\right)^2\right).
\end{split}
\end{equation*}
\end{enumerate}
\end{proposition}
\begin{proof}
We repeat the argument in Proposition \ref{localization}, using Proposition \ref{bddcom1} instead of \ref{bddcom}.
\end{proof}

When using the conservation law for $Z$, we can ignore the part of the bulk term that has a good sign.
\begin{definition}
Let $K^{Z,w^Z}_+\left(\Phi\right)=\max\{K^{Z,w^Z}\left(\Phi\right),0\}.$
\end{definition}

Using the conservation law for the modified vector field, we have a one-sided bound:
\begin{proposition}\label{Zid}
\begin{equation*}
\begin{split}
&\int_{\Sigma_\tau} J^{Z,w^Z}_\mu\left(\Phi\right)n^\mu_{\Sigma_\tau}+\int_{\mathcal H(\tau',\tau)} J^{Z,w^Z}_\mu\left(\Phi\right)n^\mu_{\mathcal H^+}\\
\leq& C(\tau')^2\int_{\Sigma_{\tau'}\cap\{r\leq r^-_Y\}} J^{N}_\mu\left(\Phi\right)n^\mu_{\Sigma_{\tau'}}+\int_{\Sigma_{\tau'}} J^{Z,w^Z}_\mu\left(\Phi\right)n^\mu_{\Sigma_{\tau'}}\\
&+\iint_{\mathcal R(\tau',\tau)} K^{Z,w^Z}_+\left(\Phi\right)+|\iint_{\mathcal R(\tau',\tau)}\left(u^2 L\Phi+v^2 \underline L\Phi-\frac{1}{4}w\Phi\right)G|.
\end{split}
\end{equation*}
\begin{remark}
In the above Proposition, the left hand side is not claimed to be positive. Note, however, that the right hand side is positive by Remark \ref{positivity}.
\end{remark}
\begin{remark}
We note also that $\int_{\mathcal H(\tau',\tau)} J^{Z,w^Z}_\mu\left(\Phi\right)n^\mu_{\mathcal H^+}\geq0$ because $Z$ and $n^\mu_{\mathcal H^+}$ are both null and future directed and $w^Z=0$ on the event horizon.
\end{remark}

\end{proposition}
To show that $\int_{\Sigma_\tau} J^{Z,w^Z}_\mu\left(\Phi\right)n^\mu_{\Sigma_\tau}$ is almost bounded, we would have to show that $\int_{\Sigma_{\tau'}\cap\{r\leq r^-_Y\}} J^{N}_\mu\left(\Phi\right)n^\mu_{\Sigma_{\tau'}}$ in fact decays. This is given by the following Proposition:
\begin{proposition}\label{decaynearhorizon}
\begin{equation*}
\begin{split}
&\int_{\Sigma_{\tau}\cap\{r\leq\frac{t^*}{2}\}} J^{N}_\mu\left(\Phi\right)n^\mu_{\Sigma_{\tau}}\\
\leq&C^2\tau^{-2}\int_{\Sigma_{(1.1)^{-2}\tau}\cap\{r\leq r^-_Y\}} J^{N}_\mu\left(\Phi\right)n^\mu_{\Sigma_{(1.1)^{-2}\tau}}+C\tau^{-2}\int_{\Sigma_{(1.1)^{-2}\tau}} J^{Z,w^Z}_\mu\left(\Phi\right)n^\mu_{\Sigma_{(1.1)^{-1}\tau}}\\
&+C\tau^{-2}\iint_{\mathcal R((1.1)^{-2}\tau,\tau)} K^{Z,w^Z}_+\left(\Phi\right)+C\tau^{-2}|\iint_{\mathcal R((1.1)^{-2}\tau,\tau)}\left(u^2 L\Phi+v^2 \underline L\Phi-\frac{1}{4}w\Phi\right)G|\\
&+C\left(\left(\int_{(1.1)^{-2}\tau-1}^{\tau+1}\left(\int_{\Sigma_{t^*}\cap\{r\leq\frac{9t^*}{10}\}}  G_1^2\right)^{\frac{1}{2}}dt^*\right)^2+\iint_{\mathcal R((1.1)^{-2}\tau-1,\tau+1)\cap\{r\leq\frac{9t^*}{10}\}}  G_1^2\right)\\
&+C\left(\sum_{m=0}^{1}\iint_{\mathcal R((1.1)^{-2}\tau-1,\tau+1)\cap\{r\leq\frac{9t^*}{10}\}}r^{1+\delta}\left(\partial_{t^*}^{m}{G_2}\right)^2+\sup_{t^*\in [(1.1)^{-2}\tau-1,\tau+1]}\int_{\Sigma_{t^*}\cap\{|r-3M|\leq\frac{M}{8}\}\cap\{r\leq\frac{9t^*}{10}\}} G_2^2\right).\\
\end{split}
\end{equation*}
\end{proposition}
\begin{proof}
By Proposition \ref{localization}.2 applied to the $t^*$ interval $[(1.1)^{-1}\tau,\tau]$, we have
\begin{equation*}
\begin{split}
&\iint_{\mathcal R((1.1)^{-1}\tau,\tau)\cap\{r\leq r^-_Y\}}K^{N}\left(\Phi\right)\\
\leq& C\tau^{-2}\int_{\Sigma_{(1.1)^{-1}\tau}} J^{Z+N,w^Z}_\mu\left(\Phi\right)n^\mu_{\Sigma_{(1.1)^{-1}\tau}}+ C^2\int_{\Sigma_{(1.1)^{-1}\tau}\cap\{r\leq r^-_Y\}} J^{N}_\mu\left(\Phi\right)n^\mu_{\Sigma_{(1.1)^{-1}\tau}}\\
&+C\left(\sum_{m=0}^{1}\iint_{\mathcal R((1.1)^{-1}\tau-1,\tau+1)\cap\{r\leq\frac{9t^*}{10}\}}r^{1+\delta}\left(\partial_{t^*}^{m}{G}\right)^2+\sup_{t^*\in [(1.1)^{-1}\tau-1,\tau+1]}\int_{\Sigma_{t^*}\cap\{|r-3M|\leq\frac{M}{8}\}} G^2\right).\\
\end{split}
\end{equation*}
By taking the infimum there exists $\tilde{\tau} \in [(1.1)^{-1}\tau,\tau]$ such that
$$\int_{\Sigma_{\tilde{\tau}}\cap\{r\leq r^-_Y\}} J^{N}_\mu\left(\Phi\right)n^\mu_{\Sigma_{\tilde{\tau}}}\leq C\tau^{-1} \iint_{\mathcal R((1.1)^{-1}\tau,\tau)\cap\{r\leq r^-_Y\}}K^{N}\left(\Phi\right).$$
Hence,
\begin{equation*}
\begin{split}
&\int_{\Sigma_{\tilde{\tau}}\cap\{r\leq r^-_Y\}} J^{N}_\mu\left(\Phi\right)n^\mu_{\Sigma_{\tilde{\tau}}}\\
\leq& C\tau^{-2}\int_{\Sigma_{(1.1)^{-1}\tau}} J^{Z+N,w^Z}_\mu\left(\Phi\right)n^\mu_{\Sigma_{(1.1)^{-1}\tau}}+ C^2\tau^{-1}\int_{\Sigma_{(1.1)^{-1}\tau}\cap\{r\leq r^-_Y\}} J^{N}_\mu\left(\Phi\right)n^\mu_{\Sigma_{(1.1)^{-1}\tau}}\\
&+C\tau^{-1}\left(\sum_{m=0}^{1}\iint_{\mathcal R((1.1)^{-1}\tau-1,\tau+1)\cap\{r\leq\frac{9t^*}{10}\}}r^{1+\delta}\left(\partial_{t^*}^{m}{G}\right)^2+\sup_{t^*\in [(1.1)^{-1}\tau-1,\tau+1]}\int_{\Sigma_{t^*}\cap\{|r-3M|\leq\frac{M}{8}\}} {G}^2\right).\\
\end{split}
\end{equation*}
Apply Proposition \ref{localization}.2 to the $t^*$ interval $[\tilde{\tau},\tau]$ and use Proposition \ref{bddcom} and \ref{Zid}
\begin{equation}\label{nearhorizon1}
\begin{split}
&\int_{\Sigma_{\tau}\cap\{r\leq\frac{t^*}{2}\}} J^{N}_\mu\left(\Phi\right)n^\mu_{\Sigma_{\tau}}\\
\leq& C\left(\tau^{-2}\int_{\Sigma_{\tilde{\tau}}} J^{Z+N,w^Z}_\mu\left(\Phi\right)n^\mu_{\Sigma_{\tilde{\tau}}}+ C\int_{\Sigma_{\tilde{\tau}}\cap\{r\leq r^-_Y\}} J^{N}_\mu\left(\Phi\right)n^\mu_{\Sigma_{\tilde{\tau}}}\right)\\
&+C\left(\sum_{m=0}^{1}\iint_{\mathcal R((1.1)^{-1}\tau-1,\tau+1)\cap\{r\leq\frac{9t^*}{10}\}}r^{1+\delta}\left(\partial_{t^*}^{m}{G}\right)^2+\sup_{t^*\in [(1.1)^{-1}\tau-1,\tau+1]}\int_{\Sigma_{t^*}\cap\{|r-3M|\leq\frac{M}{8}\}\cap\{r\leq\frac{9t^*}{10}\}} {G}^2\right)\\
\leq& C\tau^{-2}\left(\int_{\Sigma_{\tilde{\tau}}} J^{Z+N,w^Z}_\mu\left(\Phi\right)n^\mu_{\Sigma_{\tilde{\tau}}}+ \int_{\Sigma_{(1.1)^{-1}\tau}} J^{Z+N,w^Z}_\mu\left(\Phi\right)n^\mu_{\Sigma_{(1.1)^{-1}\tau}}\right)\\
&+C^2\tau^{-1}\int_{\Sigma_{(1.1)^{-1}\tau}\cap\{r\leq r^-_Y\}} J^{N}_\mu\left(\Phi\right)n^\mu_{\Sigma_{(1.1)^{-1}\tau}}\\
&+C\left(\sum_{m=0}^{1}\iint_{\mathcal R((1.1)^{-1}\tau-1,\tau+1)\cap\{r\leq\frac{9t^*}{10}\}}r^{1+\delta}\left(\partial_{t^*}^{m}{G}\right)^2+\sup_{t^*\in [(1.1)^{-1}\tau-1,\tau+1]}\int_{\Sigma_{t^*}\cap\{|r-3M|\leq\frac{M}{8}\}\cap\{r\leq\frac{9t^*}{10}\}} {G}^2\right)\\
\leq&C^2\tau^{-1}\int_{\Sigma_{(1.1)^{-1}\tau}\cap\{r\leq r^-_Y\}} J^{N}_\mu\left(\Phi\right)n^\mu_{\Sigma_{(1.1)^{-1}\tau}}+C\tau^{-2}\int_{\Sigma_{(1.1)^{-1}\tau}} J^{Z+N,w^Z}_\mu\left(\Phi\right)n^\mu_{\Sigma_{(1.1)^{-1}\tau}}\\
&+C\tau^{-2}\iint_{\mathcal R((1.1)^{-1}\tau,\tau)} K^{Z,w^Z}_+\left(\Phi\right)+C\tau^{-2}|\iint_{\mathcal R((1.1)^{-1}\tau,\tau)}\left(u^2 \underline{L}\Phi+v^2 L\Phi-\frac{1}{4}w\Phi\right)G|\\
&+C\left(\sum_{m=0}^{1}\iint_{\mathcal R((1.1)^{-1}\tau-1,\tau+1)\cap\{r\leq\frac{9t^*}{10}\}}r^{1+\delta}\left(\partial_{t^*}^{m}{G}\right)^2+\sup_{t^*\in [(1.1)^{-1}\tau-1,\tau+1]}\int_{\Sigma_{t^*}\cap\{|r-3M|\leq\frac{M}{8}\}\cap\{r\leq\frac{9t^*}{10}\}} {G}^2\right)\\
\end{split}
\end{equation}
Replacing $[(1.1)^{-1}\tau,\tau]$ with $[(1.1)^{-2}\tau,(1.1)^{-1}\tau]$, we get also
\begin{equation}\label{nearhorizon2}
\begin{split}
&\int_{\Sigma_{(1.1)^{-1}\tau}\cap\{r\leq\frac{t^*}{2}\}} J^{N}_\mu\left(\Phi\right)n^\mu_{\Sigma_{\tau}}\\
\leq&C^2\tau^{-1}\int_{\Sigma_{(1.1)^{-2}\tau}\cap\{r\leq r^-_Y\}} J^{N}_\mu\left(\Phi\right)n^\mu_{\Sigma_{(1.1)^{-2}\tau}}+C\tau^{-2}\int_{\Sigma_{(1.1)^{-2}\tau}} J^{Z+N,w^Z}_\mu\left(\Phi\right)n^\mu_{\Sigma_{(1.1)^{-2}\tau}}\\
&+C\tau^{-2}\iint_{\mathcal R((1.1)^{-2}\tau,(1.1)^{-1}\tau)} K^{Z,w^Z}_+C\left(\Phi\right)+C\tau^{-2}|\iint_{\mathcal R((1.1)^{-2}\tau,(1.1)^{-1}\tau)}\left(u^2 L\Phi+v^2 \underline L\Phi-\frac{1}{4}w\Phi\right)G|\\
&+C\left(\sum_{m=0}^{1}\iint_{\mathcal R((1.1)^{-2}\tau-1,(1.1)^{-1}\tau+1)\cap\{r\leq\frac{9t^*}{10}\}}r^{1+\delta}\left(\partial_{t^*}^{m}{G}\right)^2\right.\\
&\left.+\sup_{t^*\in [(1.1)^{-2}\tau-1,(1.1)^{-1}\tau+1]}\int_{\Sigma_{t^*}\cap\{|r-3M|\leq\frac{M}{8}\}\cap\{r\leq\frac{9t^*}{10}\}} {G}^2\right)\\
\end{split}
\end{equation}
Therefore, plugging (\ref{nearhorizon2}) into (\ref{nearhorizon1}),
\begin{equation*}
\begin{split}
&\int_{\Sigma_{\tau}\cap\{r\leq\frac{t^*}{2}\}} J^{N}_\mu\left(\Phi\right)n^\mu_{\Sigma_{\tau}}\\
\leq&C^2\tau^{-2}\int_{\Sigma_{(1.1)^{-2}\tau}\cap\{r\leq r^-_Y\}} J^{N}_\mu\left(\Phi\right)n^\mu_{\Sigma_{(1.1)^{-2}\tau}}+C\tau^{-2}\int_{\Sigma_{(1.1)^{-2}\tau}} J^{Z,w^Z}_\mu\left(\Phi\right)n^\mu_{\Sigma_{(1.1)^{-1}\tau}}\\
&+C\tau^{-2}\iint_{\mathcal R((1.1)^{-2}\tau,\tau)} K^{Z,w^Z}_+\left(\Phi\right)+C\tau^{-2}|\iint_{\mathcal R((1.1)^{-2}\tau,\tau)}\left(u^2 L\Phi+v^2 \underline L\Phi-\frac{1}{4}w\Phi\right)G|\\
&+C\left(\sum_{m=0}^{1}\iint_{\mathcal R((1.1)^{-2}\tau-1,\tau+1)\cap\{r\leq\frac{9t^*}{10}\}}r^{1+\delta}\left(\partial_{t^*}^{m}{G}\right)^2+\sup_{t^*\in [(1.1)^{-2}\tau-1,\tau+1]}\int_{\Sigma_{t^*}\cap\{|r-3M|\leq\frac{M}{8}\}\cap\{r\leq\frac{9t^*}{10}\}} {G}^2\right)\\
\end{split}
\end{equation*}
\end{proof}
Proposition \ref{decaynearhorizon} immediately gives control over the non-degenerate energy and conformal energy using Propositions \ref{Zlowerbound} and \ref{Zid} respectively:
\begin{corollary}\label{decaywitherror}
For any $\gamma<1$,
\begin{equation*}
\begin{split}
&\int_{\Sigma_{\tau}} J^{Z,w^Z}_\mu\left(\Phi\right)n^\mu_{\Sigma_{{\tau}}}+C\tau^2\int_{\Sigma_{\tau}\cap\{r\leq \gamma\tau\}} J^{N}_\mu\left(\Phi\right)n^\mu_{\Sigma_{\tau}}\\
\leq& C\left(\int_{\Sigma_{\tau_0}} J^{Z,w^Z}_\mu\left(\Phi\right)n^\mu_{\Sigma_{{\tau_0}}}
+C\int_{\Sigma_{\tau_0}} J^{N}_\mu\left(\Phi\right)n^\mu_{\Sigma_{\tau_0}}\right.\\
&\left.+\iint_{\mathcal R(\tau_0,\tau)} K^{Z,w^Z}_+\left(\Phi\right)+|\iint_{\mathcal R(\tau_0,\tau)}\left(u^2 \underline{L}\Phi+v^2 L\Phi-\frac{1}{4}w\Phi\right)G|\right)\\
&+C\left(\sum_{m=0}^{1}\iint_{\mathcal R(\tau_0-1,\tau+1)\cap\{r\leq\frac{9t^*}{10}\}}(t^*)^2 r^{1+\delta}\left(\partial_{t^*}^{m}G\right)^2+\sup_{t^*\in [\tau_0-1,\tau+1]}\int_{\Sigma_{t^*}\cap\{|r-3M|\leq\frac{M}{8}\}\cap\{r\leq\frac{9t^*}{10}\}} (t^*)^2G^2\right).\\
\end{split}
\end{equation*}
\end{corollary}
\begin{proof}
By Proposition \ref{Zlowerbound},
\begin{equation*}
\begin{split}
\tau^2\int_{\Sigma_{\tau}\cap\{r\leq \gamma\tau\}} J^{N}_\mu\left(\Phi\right)n^\mu_{\Sigma_{\tau}}
\leq &C\int_{\Sigma_{\tau}} J^{Z+N,w^Z}_\mu\left(\Phi\right)n^\mu_{\Sigma_{{\tau}}}+C^2\tau^2\int_{\Sigma_{\tau}\cap\{r\leq r^-_Y\}} J^{N}_\mu\left(\Phi\right)n^\mu_{\Sigma_{\tau}}.
\end{split}
\end{equation*}
Therefore, by Propositions \ref{Zid} and \ref{decaynearhorizon},
\begin{equation*}
\begin{split}
&\int_{\Sigma_{\tau}} J^{Z,w^Z}_\mu\left(\Phi\right)n^\mu_{\Sigma_{{\tau}}}+C\tau^2\int_{\Sigma_{\tau}\cap\{r\leq \gamma\tau\}} J^{N}_\mu\left(\Phi\right)n^\mu_{\Sigma_{\tau}}\\
\leq& C\left(\int_{\Sigma_{(1.1)^{-2}\tau}} J^{N}_\mu\left(\Phi\right)n^\mu_{\Sigma_{(1.1)^{-2}{\tau}}}+\int_{\Sigma_{(1.1)^{-2}\tau}} J^{Z,w^Z}_\mu\left(\Phi\right)n^\mu_{\Sigma_{(1.1)^{-2}{\tau}}}+\iint_{\mathcal R((1.1)^{-2}\tau,\tau)} K^{Z,w^Z}_+\left(\Phi\right)\right.\\
&\left.\quad\quad\quad+|\iint_{\mathcal R((1.1)^{-2}\tau,\tau)}\left(u^2 \underline{L}\Phi+v^2 L\Phi-\frac{1}{4}w\Phi\right)G|\right)\\
&+C\tau^2\left(\sum_{m=0}^{1}\iint_{\mathcal R((1.1)^{-2}\tau-1,\tau+1)}r^{1+\delta}\left(\partial_{t^*}^{m}G\right)^2+\sup_{t^*\in [(1.1)^{-2}\tau-1,\tau+1]}\int_{\Sigma_{t^*}\cap\{|r-3M|\leq\frac{M}{8}\}\cap\{r\leq\frac{9t^*}{10}\}} G^2\right)\\
\end{split}
\end{equation*}
We then use the same estimate for $[(1.1)^{-4}\tau,(1.1)^{-2}\tau], [(1.1)^{-6}\tau,(1.1)^{-4}\tau], ...$. 
\end{proof}

The term $\iint_{\mathcal R(\tau_0,\tau)} K^{Z,w^Z}_+\left(\Phi\right)$ can be controlled. Here is where the control of the logarithmic divergences from the red-shift vector field is crucially used.
\begin{proposition}
\begin{equation*}
\begin{split}
&\iint_{\mathcal R(\tau',\tau)} K^{Z,w^Z}_+\left(\Phi\right)\\
\leq& C\iint_{\mathcal R(\tau',\tau)\cap\{r\geq r^-_Y\}} t^*\left(r^{-2} J^N_\mu\left(\Phi\right)n^\mu_{\Sigma_{\bar\tau}}+r^{-4}\Phi\right)+\epsilon\iint_{\mathcal R(\tau',\tau)\cap\{r\leq r^-_Y\}} \left(t^*\right)^2K^N\left(\Phi\right).
\end{split}
\end{equation*}
\begin{proof}
See \cite{DRL}.
\end{proof}
\end{proposition}

The bulk term arising from the inhomogeneous term $G$ can also be controlled. 
\begin{proposition}\label{Zinho}
\begin{equation*}
\begin{split} 
&|\iint_{\mathcal R(\tau_0,\tau)}\left(u^2 \underline{L}\Phi+v^2 L\Phi-\frac{1}{4}w\Phi\right)G|\\
\leq& \delta'\iint_{\mathcal R(\tau_0,\tau)\cap\{r\leq \frac{t^*}{2}\}} (t^*)^2K^{X_0}\left(\Phi\right)+\delta'\iint_{\mathcal R(\tau_0,\tau)\cap\{r\leq r^-_Y\}} (t^*)^2K^{N}\left(\Phi\right)\\
&+\delta'\sup_{t^*\in [\tau_0,\tau]}\left(\int_{\Sigma_{t^*}\cap\{r\geq \frac{t^*}{2}\}} J^{Z+N,w^Z}_\mu\left(\Phi\right)n^\mu_{\Sigma_{t^*}}+(t^*)^2\int_{\Sigma_{t^*}\cap\{r\leq \frac{23M}{8}\}} J^{N}_\mu\left(\Phi\right)n^\mu_{\Sigma_{t^*}}\right)\\
&+C(\delta')^{-1}\sum_{m=0}^1\iint_{\mathcal R(\tau_0,\tau)\cap\{r\leq\frac{t^*}{2}\}} (t^*)^2r^{1+\delta}\left(\partial_{t^*}^m G\right)^2+C(\delta')^{-1}\left(\int_{\tau_0}^{\tau}\left(\int_{\Sigma_{t^*}\cap\{r\geq \frac{t^*}{2}\}}r^2 G^2 \right)^{\frac{1}{2}}dt^*\right)^2\\
&+C(\delta')^{-1}\sup_{t^*\in [\tau_0,\tau]}\int_{\Sigma_{t^*}\cap\{r^-_Y\leq r\leq \frac{25M}{8}\}} (t^*)^2 G^2.
\end{split}
\end{equation*}
\end{proposition}
\begin{proof}
Two regions require particular care to deal with. The first is the region $\{r\leq r^-_Y\}$, since the coefficients of the vector field $Z$ are not bounded as $r\to r_+$. The other is the region $\{|r-3M|\leq \frac{M}{8}\}$. This is where trapping occurs and where the integrated decay estimate degenerates or loses derivatives. We first look at the region $\{r\leq r_+\}$ using the null frame:
\begin{equation*}
\begin{split}
&|\iint_{\mathcal R(\tau_0,\tau)\cap\{r\leq r^-_Y\}}\left(u^2 \underline{L}\Phi+v^2 L\Phi-\frac{1}{4}w\Phi\right)G|\\
\leq& C\iint_{\mathcal R(\tau_0,\tau)\cap\{r\leq r^-_Y\}} \left((t^*)^2+(r^*_S)^2\right)\left(|\nabla_{\hat{V}}\Phi G|+(1-\mu )|\nabla_{\hat{Y}}\Phi G|+(1-\mu )\sum_{A}|\nabla_{E_A}\Phi G|\right)\\
&\quad\quad\mbox{using Proposition \ref{Znearhorizon}}\\
\leq& C\iint_{\mathcal R(\tau_0,\tau)\cap\{r\leq r^-_Y\}} (t^*)^2\left(|\log|r-r_+||^2|\nabla_{\hat{V}}\Phi G|+|\nabla_{\hat{Y}}\Phi G|+\sum_{A}|\nabla_{E_A}\Phi G|\right)\\
\leq& \delta'\iint_{\mathcal R(\tau_0,\tau)\cap\{r\leq r^-_Y\}} (t^*)^2\left(|\log|r-r_+||^4\left(\nabla_{\hat{V}}\Phi \right)^2+\left(\nabla_{\hat{Y}}\Phi \right)^2+\sum_{A}\left(\nabla_{E_A}\Phi\right)^2\right)\\
&+ C(\delta')^{-1}\iint_{\mathcal R(\tau_0,\tau)\cap\{r\leq r^-_Y\}} (t^*)^2 G^2\\
\leq& \delta'\iint_{\mathcal R(\tau_0,\tau)\cap\{r\leq r^-_Y\}} (t^*)^2K^N\left(\Phi\right)+ C(\delta')^{-1}\iint_{\mathcal R(\tau_0,\tau)\cap\{r\leq r^-_Y\}} (t^*)^2 G^2.\\
\end{split}
\end{equation*} For the region $\{r^-_Y\leq r\leq \frac{25M}{8}\}$, where trapping occurs, we integrate by parts in $t^*$ so that the bulk term does not have $\partial_{t^*}\Phi$, which cannot be controlled by the integrated decay estimate.
\begin{equation*}
\begin{split}
&|\iint_{\mathcal R(\tau_0,\tau)\cap\{r^-_Y\leq r\leq \frac{25M}{8}\}}\left(u^2 \underline{L}\Phi+v^2 L\Phi-\frac{1}{4}w\Phi\right)G|\\
\leq&C\iint_{\mathcal R(\tau_0,\tau)\cap\{r^-_Y\leq r\leq \frac{25M}{8}\}}(t^*)^2 |\partial_r\Phi G|+(t^*)^2 |\Phi \partial_{t^*}G| +t^*|\Phi G|+\int_{\Sigma_\tau} \tau^2|\Phi G|+\int_{\Sigma_{\tau_0}} \tau_0^2|\Phi G|\\
\leq&C\left(\iint_{\mathcal R(\tau_0,\tau)\cap\{ r\leq \frac{25M}{8}\}}(t^*)^2 \left(\Phi^2+\left(\partial_r\Phi\right)^2\right)\right)^{\frac{1}{2}}\left(\sum_{m=0}^1\iint_{\mathcal R(\tau_0,\tau)\cap\{r^-_Y\leq r\leq \frac{25M}{8}\}}(t^*)^2\left(\partial_{t^*}^m G\right)^2\right)^{\frac{1}{2}}\\
&+\delta'\int_{\Sigma_\tau\cap\{r^-_Y\leq r\leq r^-_Y\leq\frac{25M}{8}\}} \tau^2 J^N_\mu\left(\Phi\right)n^\mu_{\Sigma_\tau}+\delta'\int_{\Sigma_{\tau_0}\cap\{r^-_Y\leq r\leq \frac{25M}{8}\}} \tau_0^2 J^N_\mu\left(\Phi\right)n^\mu_{\Sigma_{\tau_0}}\\
&+C(\delta')^{-1}\sup_{t^*\in [\tau_0,\tau]}\int_{\Sigma_{t^*}\cap\{r^-_Y\leq r\leq \frac{25M}{8}\}}  (t^*)^2 G^2,\\
\end{split}
\end{equation*}
using Proposition \ref{hi}.
We then move to the region $\{\frac{25M}{8}\leq r\leq\frac{t^*}{2}\}$:
\begin{equation*}
\begin{split}
&|\iint_{\mathcal R(\tau_0,\tau)\cap\{\frac{25M}{8}\leq r\leq \frac{t^*}{2}\}}\left(u^2 \underline{L}\Phi+v^2 L\Phi-\frac{1}{4}w\Phi\right)G|\\
\leq&C\iint_{\mathcal R(\tau_0,\tau)\cap\{\frac{25M}{8}\leq r\leq \frac{t^*}{2}\}}\left((t^*)^2 |\partial\Phi|+t^*|\Phi|\right)|G|\\
\leq&C\left(\iint_{\mathcal R(\tau_0,\tau)\cap\{\frac{25M}{8}\leq r\leq \frac{t^*}{2}\}}(t^*)^2\left(r^{-3-\delta}\Phi^2+ r^{-1-\delta} J^N_\mu\left(\Phi\right)n^\mu_{\Sigma_{t^*}}\right)\right)^{\frac{1}{2}}\\
&\times\left(\iint_{\mathcal R(\tau_0,\tau)\cap\{\frac{25M}{8}\leq r\leq \frac{t^*}{2}\}}\left(r^{3+\delta}+(t^*)^2 r^{1+\delta}\right)G^2\right)^{\frac{1}{2}}\\
\leq&C\left(\iint_{\mathcal R(\tau_0,\tau)\cap\{\frac{25M}{8}\leq r\leq \frac{t^*}{2}\}}(t^*)^2\left(r^{-3-\delta}\Phi^2+ r^{-1-\delta} J^N_\mu\left(\Phi\right)n^\mu_{\Sigma_{t^*}}\right)\right)^{\frac{1}{2}}\left(\iint_{\mathcal R(\tau_0,\tau)\cap\{\frac{25M}{8}\leq r\leq \frac{t^*}{2}\}}(t^*)^2 r^{1+\delta}G^2\right)^{\frac{1}{2}}\\
\end{split}
\end{equation*}
Finally, we estimate in the region $\{r\geq\frac{t^*}{2}\}$:
\begin{equation*}
\begin{split}
&|\iint_{\mathcal R(\tau_0,\tau)\cap\{r\geq \frac{t^*}{2}\}}\left(u^2 \underline{L}\Phi+v^2 L\Phi-\frac{1}{4}w\Phi\right)G|\\
\leq&C\sup_{t^*\in [\tau_0,\tau]}\left(\int_{\Sigma_{t^*}\cap\{r\geq \frac{t^*}{2}\}} J^{Z,w^Z}_\mu\left(\Phi\right)n^\mu_{\Sigma_{t^*}}+(t^*)^2\int_{\Sigma_{t^*}\cap\{r\leq r^-_Y\}} J^{N}_\mu\left(\Phi\right)n^\mu_{\Sigma_{t^*}}\right)^{\frac{1}{2}}\int_{\tau_0}^{\tau}\left(\int_{\Sigma_{t^*}\cap\{r\geq \frac{t^*}{2}\}}r^2 G^2 \right)^{\frac{1}{2}}dt^*,
\end{split}
\end{equation*}
where we have used Proposition \ref{Zlowerbound}.The Proposition follows from Cauchy-Schwarz.
\end{proof}

We have therefore proved the following decay result associated to the vector field $Z$.
\begin{proposition}\label{energydecay}
For $\delta, \delta'>0$ sufficiently small and $0\leq\gamma< 1$, there exist $c=c(\delta,\gamma)$ and $C=C(\delta,\gamma)$ such that the following estimate holds for any solution to $\Box_{g_K}\Phi=G$:
\begin{equation*}
\begin{split}
&c\int_{\Sigma_{\tau}} J^{Z,w^Z}_\mu\left(\Phi\right)n^\mu_{\Sigma_{{\tau}}}+\tau^2\int_{\Sigma_{\tau}\cap\{r\leq \gamma\tau\}} J^{N}_\mu\left(\Phi\right)n^\mu_{\Sigma_{\tau}}\\
\leq &C\int_{\Sigma_{\tau_0}} J^{Z+CN,w^Z}_\mu\left(\Phi\right)n^\mu_{\Sigma_{{\tau_0}}}+C\iint_{\mathcal R(\tau_0,\tau)} t^*r^{-1+\delta}K^{X_1}\left(\Phi\right)\\
& +C\delta'\iint_{\mathcal R(\tau_0,\tau)\cap\{r\leq \frac{t^*}{2}\}} (t^*)^2K^{X_0}\left(\Phi\right)+C\left(\delta'+\epsilon\right)\iint_{\mathcal R(\tau_0,\tau)\cap\{r\leq r^-_Y\}}(t^*)^2K^N\left(\Phi\right)\\
&+C(\delta')^{-1}\left(\int_{\tau_0}^{\tau}\left(\int_{\Sigma_{t^*}\cap\{r\geq \frac{t^*}{2}\}}r^2 G^2 \right)^{\frac{1}{2}}dt^*\right)^2+C(\delta')^{-1}\sum_{m=0}^1\iint_{\mathcal R(\tau_0,\tau)\cap\{r\leq\frac{9t^*}{10}\}} (t^*)^2r^{1+\delta}\left(\partial_{t^*}^m G\right)^2\\
&+C(\delta')^{-1}\sup_{t^*\in [\tau_0,\tau]}\int_{\Sigma_{t^*}\cap\{r^-_Y\leq r\leq \frac{25M}{8}\}} (t^*)^2 G^2.
\end{split}
\end{equation*}
\end{proposition}

\section{Estimates for Solutions to $\Box_{g_K}\Phi=0$}\label{sectionhomo}
From this point onwards, we consider $\Box_{g_K}\Phi=0$. In this section, we write down the energy estimates derived by Dafermos-Rodnianski \cite{DRL}. These will be used in later sections.
\begin{proposition}\label{0energydecay2}
\begin{equation*}
\begin{split}
&\tau^2\int_{\Sigma_\tau\cap\{r\leq \frac{\tau}{2}\}} J^{N}_\mu\left(\Phi\right)n^\mu_{\Sigma_\tau}+c\int_{\Sigma_\tau} J^{Z+N,w^Z}_\mu\left(\Phi\right)n^\mu_{\Sigma_\tau}\\
\leq& C\tau^{\eta}\sum_{m=0}^2\left(\int_{\Sigma_{\tau_0}} J^{Z,w^Z}_\mu\left(\partial_{t^*}^m\Phi\right)n^\mu_{\Sigma_{{\tau_0}}}+\int_{\Sigma_{\tau_0}} J^{N}_\mu\left(\partial_{t^*}^m\Phi\right)n^\mu_{\Sigma_{{\tau_0}}}\right).\\
\end{split}
\end{equation*}
\end{proposition}
\begin{proof}
We introduce the bootstrap assumptions:
\begin{equation}\label{bootstrap1}
\begin{split}
&\tau^2\int_{\Sigma_\tau\cap\{r\leq \frac{\tau}{2}\}} J^{N}_\mu\left(\Phi\right)n^\mu_{\Sigma_\tau}+c\int_{\Sigma_\tau} J^{Z+N,w^Z}_\mu\left(\Phi\right)n^\mu_{\Sigma_\tau}\\
\leq& A^2\tau^{\eta}\sum_{m=0}^2\int_{\Sigma_{\tau_0}} J^{Z+CN,w^Z}_\mu\left(\partial_{t^*}^m\Phi\right)n^\mu_{\Sigma_{{\tau_0}}}.\\
\end{split}
\end{equation}
\begin{equation}\label{bootstrap2}
\begin{split}
&\tau^2\int_{\Sigma_\tau\cap\{r\leq \frac{\tau}{2}\}} J^{N}_\mu\left(\partial_{t^*}\Phi\right)n^\mu_{\Sigma_\tau}+c\int_{\Sigma_\tau} J^{Z+N,w^Z}_\mu\left(\Phi\right)n^\mu_{\Sigma_\tau}\\
\leq& A\tau^{1+\eta}\sum_{m=0}^2\int_{\Sigma_{\tau_0}} J^{Z+CN,w^Z}_\mu\left(\partial_{t^*}^m\Phi\right)n^\mu_{\Sigma_{{\tau_0}}}.\\
\end{split}
\end{equation}
Here we think of $\eta$ as a small positive number.
We divide the interval $[\tau_0,\tau]$ dyadically into $\tau_0\leq\tau_1\leq ...\leq \tau_{n-1}\leq\tau_n=\tau$ with $\tau_{i+1}\leq (1.1)\tau_i$ and $n$ the smallest integer for doing such division. We then have $n\sim\log|\tau-\tau'|$. We can now apply Proposition \ref{localization} on the intervals $[\tau_{i-1},\tau_i]$ and use the bootstrap assumption (\ref{bootstrap1}):
\begin{equation*}
\begin{split}
&\iint_{\mathcal R(\tau_{i-1},\tau_i)\cap\{r\leq \frac{t^*}{2}\}}K^{X_0}\left(\Phi\right)+\iint_{\mathcal R(\tau_{i-1},\tau_i)\cap\{r\leq r^-_Y\}}K^{N}\left(\Phi\right)\\
\leq& C\left(\tau_i^{-2}\int_{\Sigma_{\tau_{i-1}}} J^{Z,w^Z}_\mu\left(\Phi\right)n^\mu_{\Sigma_{\tau_{i-1}}}+ C\int_{\Sigma_{\tau_{i-1}}\cap\{r\leq r^-_Y\}} J^{N}_\mu\left(\Phi\right)n^\mu_{\Sigma_{\tau_{i-1}}}\right)\\
\leq& CA^2\tau_i^{-2+\eta}\sum_{m=0}^2\int_{\Sigma_{\tau_0}} J^{Z+CN,w^Z}_\mu\left(\partial_{t^*}^m\Phi\right)n^\mu_{\Sigma_{{\tau_0}}}.
\end{split}
\end{equation*}
Similarly, we can apply Proposition \ref{localization} on the intervals $[\tau_{i-1},\tau_i]$ for $\partial_{t^*}\Phi$ and use the bootstrap assumption (\ref{bootstrap2}):
\begin{equation*}
\begin{split}
&\iint_{\mathcal R(\tau_{i-1},\tau_i)\cap\{r\leq \frac{t^*}{2}\}}K^{X_0}\left(\partial_{t^*}\Phi\right)+\iint_{\mathcal R(\tau_{i-1},\tau_i)\cap\{r\leq r^-_Y\}}K^{N}\left(\partial_{t^*}\Phi\right)\\
\leq& C\left(\tau_i^{-2}\int_{\Sigma_{\tau_{i-1}}} J^{Z,w^Z}_\mu\left(\partial_{t^*}\Phi\right)n^\mu_{\Sigma_{\tau_{i-1}}}+ C\int_{\Sigma_{\tau_{i-1}}\cap\{r\leq r^-_Y\}} J^{N}_\mu\left(\partial_{t^*}\Phi\right)n^\mu_{\Sigma_{\tau_{i-1}}}\right)\\
\leq& CA\tau_i^{-1+\eta}\sum_{m=0}^2\int_{\Sigma_{\tau_0}} J^{Z+CN,w^Z}_\mu\left(\partial_{t^*}^m\Phi\right)n^\mu_{\Sigma_{{\tau_0}}}.
\end{split}
\end{equation*}
By Proposition \ref{bddcom1}, we have
\begin{equation*}
\begin{split}
&\iint_{\mathcal R(\tau_{i-1},\tau_i)}r^{-1+\delta}K^{X_1}\left(\partial_{t^*}\Phi\right)\leq C\sum_{m=0}^1\int_{\Sigma_{\tau_0}} J^{N}_\mu\left(\partial_{t^*}^m\Phi\right)n^\mu_{\Sigma_{{\tau_0}}}.
\end{split}
\end{equation*}
By Propositions \ref{bddcom1} and \ref{localizationT}, we have
\begin{equation*}
\begin{split}
&\iint_{\mathcal R(\tau_{i-1},\tau_i)}r^{-1+\delta}K^{X_1}\left(\Phi\right)\\
\leq&C\iint_{\mathcal R(\tau_{i-1},\tau_i)\cap\{r\leq\frac{t^*}{2}\}}K^{X_1}\left(\Phi\right)+C\tau_i^{-1+\delta}\iint_{\mathcal R(\tau_{i-1},\tau_i)\cap\{r\geq\frac{t^*}{2}\}}K^{X_1}\left(\Phi\right)\\
\leq& CA\tau_i^{-1+\eta}\sum_{m=0}^2\int_{\Sigma_{\tau_0}} J^{Z+CN,w^Z}_\mu\left(\partial_{t^*}^m\Phi\right)n^\mu_{\Sigma_{{\tau_0}}}.
\end{split}
\end{equation*}
Apply Proposition \ref{energydecay}, we get
\begin{equation*}
\begin{split}
&c\int_{\Sigma_{\tau}} J^{Z+N,w^Z}_\mu\left(\Phi\right)n^\mu_{\Sigma_{{\tau}}}+\tau^2\int_{\Sigma_{\tau}\cap\{r\leq \gamma\tau\}} J^{N}_\mu\left(\Phi\right)n^\mu_{\Sigma_{\tau}}\\
\leq &C\int_{\Sigma_{\tau_0}} J^{Z+CN,w^Z}_\mu\left(\Phi\right)n^\mu_{\Sigma_{{\tau_0}}}+C\iint_{\mathcal R(\tau_0,\tau)} t^*r^{-1+\delta}K^{X_1}\left(\Phi\right)\\
& +C\delta'\iint_{\mathcal R(\tau_0,\tau)\cap\{r\leq \frac{t^*}{2}\}} (t^*)^2K^{X_0}\left(\Phi\right)+C\left(\delta'+\epsilon\right)\iint_{\mathcal R(\tau_0,\tau)\cap\{r\leq r^-_Y\}}(t^*)^2K^N\left(\Phi\right) \\
\leq &\left(C+\left(C+CA+CA^2(2\delta'+\epsilon)\right)\sum_{i=0}^{n-1}\tau_i^\eta\right)\sum_{m=0}^2\int_{\Sigma_{\tau_0}} J^{Z+CN,w^Z}_\mu\left(\partial_{t^*}^m\Phi\right)n^\mu_{\Sigma_{{\tau_0}}} \\
\leq&\left(C+\eta^{-1}\left(C+CA+CA^2(2\delta'+\epsilon)\right)\tau^\eta\right)\sum_{m=0}^2\int_{\Sigma_{\tau_0}} J^{Z+CN,w^Z}_\mu\left(\partial_{t^*}^m\Phi\right)n^\mu_{\Sigma_{{\tau_0}}}.
\end{split}
\end{equation*}
Now take $A$ large, $\epsilon=\frac{\eta}{4C}$ and $\delta'=\frac{\epsilon}{2}$, we improve (\ref{bootstrap1}). Apply Proposition \ref{energydecay} again, this time to $\partial_{t^*}\Phi$, we have
\begin{equation*}
\begin{split}
&c\int_{\Sigma_{\tau}} J^{Z,w^Z}_\mu\left(\partial_{t^*}\Phi\right)n^\mu_{\Sigma_{{\tau}}}+\tau^2\int_{\Sigma_{\tau}\cap\{r\leq \gamma\tau\}} J^{N}_\mu\left(\partial_{t^*}\Phi\right)n^\mu_{\Sigma_{\tau}}\\
\leq &C\int_{\Sigma_{\tau_0}} J^{Z+CN,w^Z}_\mu\left(\partial_{t^*}\Phi\right)n^\mu_{\Sigma_{{\tau_0}}}+C\iint_{\mathcal R(\tau_0,\tau)} t^*r^{-1+\delta}K^{X_1}\left(\partial_{t^*}\Phi\right)\\
& +C\delta'\iint_{\mathcal R(\tau_0,\tau)\cap\{r\leq \frac{t^*}{2}\}} (t^*)^2K^{X_0}\left(\partial_{t^*}\Phi\right)+C\left(\delta'+\epsilon\right)\iint_{\mathcal R(\tau_0,\tau)\cap\{r\leq r^-_Y\}}(t^*)^2K^N\left(\partial_{t^*}\Phi\right)\\
\leq&\left(C+C\sum_{i=0}^{n-1}\tau_i+CA(2\delta'+\epsilon)\sum_{i=0}^{n-1}\tau_i^{1+\eta}\right)\sum_{m=0}^2\int_{\Sigma_{\tau_0}} J^{Z+CN,w^Z}_\mu\left(\partial_{t^*}^m\Phi\right)n^\mu_{\Sigma_{{\tau_0}}} \\
\leq&\left(C+C\tau+CA(2\delta'+\epsilon)\tau^{1+\eta}\right)\sum_{m=0}^2\int_{\Sigma_{\tau_0}} J^{Z+CN,w^Z}_\mu\left(\partial_{t^*}^m\Phi\right)n^\mu_{\Sigma_{{\tau_0}}}.
\end{split}
\end{equation*}
Now take $A$ large, $\delta'={\epsilon}$ and $\epsilon$ sufficiently small, we also improve (\ref{bootstrap2}).
\end{proof}

In particular, the Theorem of Dafermos-Rodnianski \cite{DRL} is retrieved.
\begin{corollary}[Dafermos-Rodnianski]\label{DRDecay}
Suppose $\Box_{g_K}\Phi=0$. Then for all $\eta>0$ and all $M>0$ there exists $a_0$ such that the following estimates hold on Kerr spacetimes with $(M,a)$ for which $a\leq a_0$:
\begin{enumerate}
\item Boundedness of Non-degenerate Energy
\begin{equation*}
\begin{split}
&\int_{\Sigma_{\tau}} J^{N}_\mu\left(\Phi\right)n^\mu_{\Sigma_{\tau}} +\int_{\mathcal H(\tau_0,\tau)} J^{N}_\mu\left(\Phi\right)n^\mu_{\mathcal H^+} +\iint_{\mathcal R(\tau_0,\tau)\cap\{r\leq r^-_Y\}}K^{N}\left(\Phi\right)+\iint_{\mathcal R(\tau',\tau)}K^{X_0}\left(\Phi\right)\\
\leq &C\int_{\Sigma_{\tau_0}} J^{N}_\mu\left(\Phi\right)n^\mu_{\Sigma_{\tau_0}}.
\end{split}
\end{equation*}
 \item Decay of Non-degenerate Energy
\begin{equation*}
\begin{split}
&\tau^2\int_{\Sigma_\tau\cap\{r\leq \gamma\tau\}} J^{N}_\mu\left(\Phi\right)n^\mu_{\Sigma_\tau}+c\int_{\Sigma_\tau} J^{Z+N,w^Z}_\mu\left(\Phi\right)n^\mu_{\Sigma_\tau}\leq C\tau^{1+\eta}\sum_{m=0}^1\int_{\Sigma_{\tau_0}} J^{Z+CN,w^Z}_\mu\left(\partial_{t^*}^m\Phi\right)n^\mu_{\Sigma_{{\tau_0}}}.
\end{split}
\end{equation*}
and
\begin{equation*}
\begin{split}
&\tau^2\int_{\Sigma_\tau\cap\{r\leq \gamma\tau\}} J^{N}_\mu\left(\Phi\right)n^\mu_{\Sigma_\tau}+c\int_{\Sigma_\tau} J^{Z+N,w^Z}_\mu\left(\Phi\right)n^\mu_{\Sigma_\tau}\leq C\tau^{\eta}\sum_{m=0}^2\int_{\Sigma_{\tau_0}} J^{Z+CN,w^Z}_\mu\left(\partial_{t^*}^m\Phi\right)n^\mu_{\Sigma_{{\tau_0}}}.
\end{split}
\end{equation*}
\item Decay of Local Integrated Energy\\
For $\tau'\leq\tau\leq (1.1)\tau'$,
\begin{equation*}
\begin{split}
&\iint_{\mathcal R(\tau',\tau)\cap\{r\leq \frac{t^*}{2}\}}K^{X_0}\left(\Phi\right)+\iint_{\mathcal R(\tau',\tau)\cap\{r\leq r^-_Y\}}K^{N}\left(\Phi\right)\leq C\tau^{-2+\eta}\sum_{m=0}^2\int_{\Sigma_{\tau_0}} J^{Z+CN,w^Z}_\mu\left(\partial_{t^*}^m\Phi\right)n^\mu_{\Sigma_{{\tau_0}}}.
\end{split}
\end{equation*}
and
\begin{equation*}
\begin{split}
&\iint_{\mathcal R(\tau',\tau)\cap\{r\leq \frac{t^*}{2}\}}K^{X_1}\left(\Phi\right)\leq C\tau^{-2+\eta}\sum_{m=0}^3\int_{\Sigma_{\tau_0}} J^{Z+CN,w^Z}_\mu\left(\partial_{t^*}^m\Phi\right)n^\mu_{\Sigma_{{\tau_0}}}.
\end{split}
\end{equation*}
\end{enumerate}
\end{corollary}
\begin{proof}
1 follows directly from Proposition \ref{Nbdd}. 2 contains two statements. The second one is a restatement of Proposition \ref{0energydecay2}. The first one is evident from the proof of Proposition \ref{0energydecay2}. 3 again has two statements. For the first statement, we revisit the proof of Proposition \ref{0energydecay2}. Notice that the bootstrap assumptions are true. Hence it holds. For the second statement, we note by comparing Propositions \ref{localization} and \ref{localizationT} that $K^{X_1}$ can be estimated in the same way as $K^{X_0}$ except for an extra derivative. The second statement in 3 can then be proved by re-running the argument in Proposition \ref{0energydecay2} with an extra derivative.
\end{proof}

\section{Estimates for $\hat{Y}\Phi$ and Elliptic Estimates}\label{sectioncommutatorY}

Away from the event horizon, we can control all higher order derivatives simply by commuting with $\partial_{t^*}$ and using standard elliptic estimates. We write down a general version of the estimates in which we have inhomogeneous terms.
\begin{proposition}\label{elliptic}
Suppose $\Box_{g_K}\Phi=G$. For $m\geq 1$ and for any $\alpha$,
\begin{enumerate}
\item Boundedness of Weighted Energy
\begin{equation*}
\begin{split}
\int_{\Sigma_{\tau}\cap\{r\geq r^-_Y\}} r^\alpha\left(D^m\Phi\right)^2 \leq C_{\alpha,m}\left(\sum_{j=0}^{m-1}\int_{\Sigma_{\tau}} r^\alpha J^N_\mu\left(\partial_{t^*}^j\Phi\right)n^\mu_{\Sigma_\tau}+\sum_{j=0}^{m-2}\int_{\Sigma_{\tau}} r^\alpha\left(D^jG\right)^2\right).
\end{split}
\end{equation*}
\item Boundedness of Local Energy\\
For any $0< \gamma<\gamma'$,
\begin{equation*}
\begin{split}
&\int_{\Sigma_{\tau}\cap\{r^-_Y\leq r\leq \gamma t^*\}} r^\alpha\left(D^m\Phi\right)^2 \\
 \leq &C_{\alpha,m,\gamma,\gamma'}\left(\sum_{j=0}^{m-1}\int_{\Sigma_{\tau}\cap\{r\leq \gamma' t^*\}} r^\alpha J^N_\mu\left(\partial_{t^*}^j\Phi\right)n^\mu_{\Sigma_\tau}+\tau^{\alpha-\beta-2}\int_{\Sigma_{\tau}} r^\beta J^N_\mu\left(\Phi\right)n^\mu_{\Sigma_\tau}\right. \\
&\left. \quad\quad+\sum_{j=0}^{m-2}\int_{\Sigma_{\tau}\cap\{r\leq \gamma' t^*\}} r^\alpha\left(D^jG\right)^2\right). \\
\end{split}
\end{equation*}
\end{enumerate}
\end{proposition}
\begin{proof}
This is obvious for $m=1$ (even without the restriction $r\geq r_Y^-$). We will proceed by induction. Take $\delta\ll\frac{r^-_Y-r_+}{4}$. Assume
\begin{equation*}
\begin{split}
\sum_{j=1}^{m-1}\int_{\Sigma_{\tau}\cap\{r\geq r^-_Y-2\delta\}} r^\alpha\left(D^j\Phi\right)^2 \leq C\left(\sum_{j=0}^{m-2}\int_{\Sigma_{\tau}} r^\alpha J^N_\mu\left(\partial_{t^*}^j\Phi\right)n^\mu_{\Sigma_\tau}+\sum_{j=0}^{m-3}\int_{\Sigma_{\tau}} r^\alpha\left(D^jG\right)^2\right).
\end{split}
\end{equation*}
We want to show
\begin{equation*}
\begin{split}
\int_{\Sigma_{\tau}\cap\{r\geq r^-_Y-\delta\}} r^\alpha\left(D^m\Phi\right)^2 \leq C\left(\sum_{j=0}^{m-1}\int_{\Sigma_{\tau}} r^\alpha J^N_\mu\left(\partial_{t^*}^j\Phi\right)n^\mu_{\Sigma_\tau}+\sum_{j=0}^{m-2}\int_{\Sigma_{\tau}} r^\alpha\left(D^jG\right)^2\right),
\end{split}
\end{equation*}
which would then imply the conclusion.
Denote by $\Delta_{g_K}$ the Laplace-Beltrami operator for the metric $g_K$ restricted on the spacelike hypersurface $t^*=\mbox{constant}$. Notice that since $\partial_{t^*}$ is Killing, the operator is defined independent of $t^*$. Then we have 
$$|[\Delta_{g_K},D^k]\Phi|\leq C\sum_{j=1}^{k+1}|D^j\Phi|.$$
Denote by $\nabla$ the spatial derivatives with respect to the spatial coordinate variables in the Schwarzschild $\left(t^*_S,r_S,x^1_S,x^2_S\right)$ coordinate system. On the set $\{r\geq r^-_Y-\frac{r^-_Y-r_+}{4}\}$, $\Delta_{g_K}$ is elliptic and therefore controls all spatial derivatives:
\begin{equation*}
\begin{split}
&\int_{\Sigma_{\tau}\cap\{r\geq r^-_Y-\delta\}} r^\alpha \left(D^m\Phi\right)^2\\
\leq &C\int_{\Sigma_{\tau}\cap\{r\geq r^-_Y-2\delta\}} r^\alpha \left(\left(\Delta_{g_K} D^{m-2}\Phi\right)^2+\left(D^{m-1}\Phi\right)^2+\left(\partial_{t^*}^{m-1}\nabla\Phi\right)^2+\left(\partial_{t^*}^m\Phi\right)^2\right)\\
\leq &C\int_{\Sigma_{\tau}\cap\{r\geq r^-_Y-2\delta\}} r^\alpha \left(\left( D^{m-2}\Delta_{g_K}\Phi\right)^2+\sum_{j=1}^{m-1}\left(D^{j}\Phi\right)^2+r^{-2}\Phi^2+\left(\partial_{t^*}^{m-1}\nabla\Phi\right)^2+\left(\partial_{t^*}^m\Phi\right)^2\right)
\end{split}
\end{equation*}
The last two terms are obviously bounded by $C\int_{\Sigma_{\tau}} J^N_\mu\left(\partial_{t^*}^{m-1}\Phi\right)n^\mu_{\Sigma_\tau}$. The second term can be bounded using the induction hypothesis. The third term can be bounded using the Hardy inequality in Proposition \ref{hi}. Finally, to estimate the first term we use the equation $\Box_{g_K}\Phi=G$. Then, by the form of the Kerr metric, $\Delta_{g_K}\Phi=G-g^{t^*t^*}\partial_{t^*}^2\Phi-2g^{t^*\phi^*}\partial_{t^*}\partial_{\phi^*}\Phi$. Therefore,
\begin{equation*}
\begin{split}
&\int_{\Sigma_{\tau}\cap\{r\geq r^-_Y-2\delta\}} r^\alpha \left(D^{m-2}\Delta_{g_K}\Phi\right)^2\\
\leq&C\int_{\Sigma_{\tau}\cap\{r\geq r^-_Y-2\delta\}}r^\alpha \left(\left(D^{m-1}\partial_{t^*}\Phi\right)^2+\left(D^{m-2}G\right)^2\right)\\
\leq&C\left(\sum_{j=0}^{m-1}\int_{\Sigma_{\tau}} r^\alpha J^N_\mu\left(\partial_{t^*}^j\Phi\right)n^\mu_{\Sigma_\tau}+\sum_{j=0}^{m-2}\int_{\Sigma_\tau}r^\alpha \left(D^jG\right)^2\right),
\end{split}
\end{equation*}
where at the last step we have used the induction hypothesis for $\partial_{t^*}\Phi$. We have thus proved the boundedness of weighted energy.
To prove the second part of the Proposition, consider the function $\chi(\frac{r}{\tau})\Phi\left(\tau\right)$ for a fixed time $t^*=\tau$, where $\chi:\mathbb R_{\geq 0} \to \mathbb R_{\geq 0}$ is supported in $\{x\leq \gamma'\}$ and is identically 1 in $\{x\leq \gamma\}$. Now $$\Box_{g_K}\Phi=\chi G+\tau^{-1}\tilde{\chi}\partial_r\Phi+\tau^{-2}\tilde{\tilde{\chi}}\Phi,$$
where $\tilde{\chi}$ and $\tilde{\tilde{\chi}}$ are supported in $\{\gamma\leq \frac{t^*}{r}\leq \gamma'\}$.
Thus, by the estimate just proved,
\begin{equation*}
\begin{split}
&\int_{\Sigma_{\tau}\cap\{r^-_Y\leq r\leq \gamma t^*\}} r^\alpha\left(D^m\Phi\right)^2 \\
 \leq &C_\alpha\left(\sum_{j=0}^{m-1}\int_{\Sigma_{\tau}\cap\{r\leq \gamma' t^*\}} r^\alpha J^N_\mu\left(\partial_{t^*}^j\Phi\right)n^\mu_{\Sigma_\tau}+\int_{\Sigma_{\tau}\cap\{\gamma t\leq r\leq \gamma' t^*\}} r^\alpha \tau^{-4}\Phi^2+\sum_{j=0}^{m-2}\int_{\Sigma_{\tau}\cap\{r\leq \gamma' t^*\}} r^\alpha\left(D^jG\right)^2\right) \\
 \leq &C_\alpha\left(\sum_{j=0}^{m-1}\int_{\Sigma_{\tau}\cap\{r\leq \gamma' t^*\}} r^\alpha J^N_\mu\left(\partial_{t^*}^j\Phi\right)n^\mu_{\Sigma_\tau}+\tau^{\alpha-\beta-2}\int_{\Sigma_{\tau}} r^\beta J^N_\mu\left(\Phi\right)n^\mu_{\Sigma_\tau} +\sum_{j=0}^{m-2}\int_{\Sigma_{\tau}\cap\{r\leq \gamma' t^*\}} r^\alpha\left(D^jG\right)^2\right), \\
\end{split}
\end{equation*}
by Hardy inequality in Proposition \ref{hi}.
\end{proof}
\begin{remark}
The boundedness of local energy should be seen as a decay result because for example for the homogeneous equation, the right hand side of the inequality decays.
\end{remark}

Near the event horizon, higher order derivatives can be controlled by commuting with the red-shift vector field as in \cite{DRK}, \cite{DRL}. The computation here will be completely local, i.e., only in the region $\{r\leq r^-_Y\}$.

We have the following estimate for higher order derivatives:
\begin{proposition}\label{elliptichorizon}
Suppose $\Box_{g_K}\Phi=G$. For every $m\geq 1$,
\begin{equation*}
\begin{split}
\int_{\Sigma_{\tau}\cap\{r\leq r^-_Y\}} \left(D^m\Phi\right)^2 \leq C\left(\sum_{j+k\leq m-1}\int_{\Sigma_{\tau}\cap\{r\leq r^-_Y\}} J^N_\mu\left(\partial_{t^*}^j\hat{Y}^k\Phi\right)n^\mu_{\Sigma_\tau}+\sum_{j=0}^{m-2}\int_{\Sigma_{\tau}\cap\{r\leq r^-_Y\}}\left(D^jG\right)^2\right).
\end{split}
\end{equation*}
\end{proposition}
\begin{proof}
This is obvious for $m=1$.  We will proceed by induction. Suppose, for some $m\geq 2$ that
\begin{equation}\label{elliptichorizonassumption}
\begin{split}
\sum_{j=0}^{m-1}\int_{\Sigma_{\tau}\cap\{r\leq r^-_Y\}} \left(D^j\Phi\right)^2 \leq C\left(\sum_{j+k\leq m-2}\int_{\Sigma_{\tau}\cap\{r\leq r^-_Y\}} J^N_\mu\left(\partial_{t^*}^j\hat{Y}^k\Phi\right)n^\mu_{\Sigma_\tau}+\sum_{j=0}^{m-3}\int_{\Sigma_{\tau}\cap\{r\leq r^-_Y\}}\left(D^jG\right)^2\right).
\end{split}
\end{equation}
Since $\Box_{g_K}\left(\partial_{t^*}\Phi\right)=\partial_{t^*}G$, this immediately implies
\begin{equation}\label{elliptichorizonT}
\begin{split}
\int_{\Sigma_{\tau}\cap\{r\leq r^-_Y\}} \left(\partial_{t^*}D^{m-1}\Phi\right)^2 \leq C\left(\sum_{j+k\leq m-1}\int_{\Sigma_{\tau}\cap\{r\leq r^-_Y\}} J^N_\mu\left(\partial_{t^*}^j\hat{Y}^k\Phi\right)n^\mu_{\Sigma_\tau}+\sum_{j=0}^{m-2}\int_{\Sigma_{\tau}\cap\{r\leq r^-_Y\}}\left(D^jG\right)^2\right).
\end{split}
\end{equation}
Since $\Box_{g_K}\Phi=G$, we have $\Box_{g_K}\left(\hat{Y}\Phi\right)=\hat{Y}G+O(1)\left(D^2\Phi+D\Phi\right)$.
Then using the induction hypothesis (\ref{elliptichorizonassumption}) (both on $\hat{Y}\Phi$ and $\Phi$), we have
\begin{equation}\label{ellipticY}
\begin{split}
&\sum_{j=0}^{m-1}\int_{\Sigma_{\tau}\cap\{r\leq r^-_Y\}} \left(D^j\hat{Y}\Phi\right)^2\\
\leq& C\left(\sum_{j+k\leq m-2}\int_{\Sigma_{\tau}\cap\{r\leq r^-_Y\}} J^N_\mu\left(\partial_{t^*}^j\hat{Y}^{k+1}\Phi\right)n^\mu_{\Sigma_\tau}+\sum_{j=0}^{m-3}\int_{\Sigma_{\tau}}\left(D^j\hat{Y}G\right)^2+\sum_{j=0}^{m-1}\int_{\Sigma_{\tau}\cap\{r\leq r^-_Y\}}\left(D^j\Phi\right)^2\right)\\
\leq&C\left(\sum_{j+k\leq m-1}\int_{\Sigma_{\tau}\cap\{r\leq r^-_Y\}} J^N_\mu\left(\partial_{t^*}^j\hat{Y}^{k}\Phi\right)n^\mu_{\Sigma_\tau}+\sum_{j=0}^{m-2}\int_{\Sigma_{\tau}\cap\{r\leq r^-_Y\}}\left(D^jG\right)^2\right).
\end{split}
\end{equation}
Using the null frame $\{\hat{V},\hat{Y}, E_1,E_2\}$,
\begin{equation*}
\begin{split}
\Box_{g_K}\left(D^{m-2}\Phi\right)=&-4\nabla_{\hat{Y}}\nabla_{\hat{V}}D^{m-2}\Phi+\lapp D^{m-2}\Phi+P_1D^{m-2}\Phi,\\
\end{split}
\end{equation*}
where $P_1$ denotes a first order differential operator.
Notice that we also have
$$|\Box_{g_K}\left(D^{m-2}\Phi\right)|=|[\Box_{g_K},D^{m-2}]\Phi+D^{m-2}G|\leq C\left(\sum_{j=0}^{m-1}|D^{j}\Phi|+|D^{m-2}G|\right).$$
Now using a standard $L^2$ elliptic estimate on the sphere,
\begin{equation*}
\begin{split}
\int_{\mathbb S^2}|\nabb^2D^{m-2}\Phi|^2 dA \leq C \int_{\mathbb S^2}\left(\left(D^{m-2}G\right)^2+\sum_{j=0}^{m-1}\left(D^{j}\Phi\right)^2+\left(D^{m-1}\nabla_{\hat{Y}}\Phi\right)^2\right)dA,
\end{split}
\end{equation*}
where we notice that the constant can be chosen uniformly because the metric on the sphere is everywhere close to that of the standard metric.
Therefore, after integrate over $\{r_+\leq r\leq r^-_Y\}$ and applying (\ref{elliptichorizonassumption}) and (\ref{ellipticY}), we have
\begin{equation}\label{elliptichorizonangular}
\begin{split}
&\int_{\Sigma_{\tau}\cap\{r\leq r^-_Y\}}|\nabb^2D^{m-2}\Phi|^2\\
\leq &C \int_{\Sigma_{\tau}\cap\{r\leq r^-_Y\}}\left(\left(D^{m-2}G\right)^2+\sum_{j=0}^{m-1}\left(D^{j}\Phi\right)^2+\left(D^{m-1}\nabla_{\hat{Y}}\Phi\right)^2\right)\\
\leq &C\left(\sum_{j+k\leq m-1}\int_{\Sigma_{\tau}\cap\{r\leq r^-_Y\}} J^N_\mu\left(\partial_{t^*}^j\hat{Y}^k\Phi\right)n^\mu_{\Sigma_\tau}+\sum_{j=0}^{m-2}\int_{\Sigma_{\tau}\cap\{r\leq r^-_Y\}}\left(D^jG\right)^2\right).
\end{split}
\end{equation}
Combining (\ref{elliptichorizonT}), (\ref{ellipticY}) and (\ref{elliptichorizonangular}), we have
\begin{equation*}
\begin{split}
\int_{\Sigma_{\tau}\cap\{r\leq r^-_Y\}} \left(D^{m}\Phi\right)^2 \leq C\left(\sum_{j+k\leq m}\int_{\Sigma_{\tau}\cap\{r\leq r^-_Y\}} J^N_\mu\left(\partial_{t^*}^j\hat{Y}^k\Phi\right)n^\mu_{\Sigma_\tau}+\sum_{j=0}^{m-2}\int_{\Sigma_{\tau}\cap\{r\leq r^-_Y\}}\left(D^jG\right)^2\right).
\end{split}
\end{equation*}
\end{proof}
We show that the currents associated to $\hat{Y}^k\Phi$ can actually be controlled. Again, in view of the nonlinear problem, we work in the setting of an inhomogeneous equation.

\begin{proposition}\label{commYcontrol}
Suppose $\Box_{g_K}\Phi=G$. For every $k\geq 0$,
\begin{equation*}
\begin{split}
&\int_{\Sigma_\tau\cap\{r\leq r^+_Y\}} J^{N}_\mu\left(\hat{Y}^{k}\Phi\right)n^\mu_{\Sigma_\tau} +\int_{\mathcal H(\tau',\tau)} J^{N}_\mu\left(\hat{Y}^{k}\Phi\right)n^\mu_{\Sigma_\tau} +\iint_{\mathcal R(\tau',\tau)\cap\{r\leq r^-_Y\}}K^{N}\left(\hat{Y}^{k}\Phi\right)\\
\leq &C\left(\sum_{j+m\leq k}\int_{\Sigma_{\tau'}\cap\{r\leq r^+_Y\}} J^N_\mu\left(\partial_{t^*}^j\hat{Y}^m\Phi\right)n^\mu_{\Sigma_{\tau'}}+\sum_{j=0}^k\int_{\Sigma_{\tau}\cap\{r\leq r^+_Y\}} J^N_\mu\left(\partial_{t^*}^j\Phi\right)n^\mu_{\Sigma_{\tau}}\right.\\
&\left.+\sum_{j=0}^k\iint_{\mathcal R(\tau',\tau)\cap\{ r\leq \frac{23M}{8}\}} \left(\Phi^2+J^N_\mu\left(\partial_{t^*}^j\Phi\right)n^\mu_{\Sigma_{t^*}}\right)+\sum_{j=0}^{k}\iint_{\mathcal R(\tau',\tau)\cap\{r\leq\frac{23M}{8}\}}\left(D^jG\right)^2\right).
\end{split}
\end{equation*}
\end{proposition}

\begin{proof}
We prove the Proposition by induction on $k$. The $k=0$ case is trivial because the right hand side simply contains more terms than the left hand side. We suppose the Proposition is true for $k\leq k_0-1$ for some $k_0\geq 1$. 
Commuting $\Box_{g_K}$ with $\hat{Y}$ for $k_0$ times, we get
\begin{equation*}
\begin{split}
\Box_{g_K}\hat{Y}^{k_0}\Phi=\kappa^{k_0}\hat{Y}^{k_0+1}\Phi+O(1)\hat{Y}^{k_0}\partial_{t^*}\Phi+O(\epsilon)D^{k_0+1}\Phi+O(1)\sum_{j=1}^{k_0}D^{j}\Phi+\hat{Y}^{k_0}G.
\end{split}
\end{equation*}
We now use the energy identity for the vector field $N$, i.e., Proposition \ref{Nid} for $\hat{Y}^{k}\Phi$. Notice that $\hat{Y}$ is supported in $\{r\leq r^+_Y\}$ and therefore each term is supported in the same set.
\begin{equation*}
\begin{split}
&\int_{\Sigma_\tau\cap\{r\leq r^+_Y\}} J^{N}_\mu\left(\hat{Y}^{k_0}\Phi\right)n^\mu_{\Sigma_\tau} +\int_{\mathcal H(\tau_0,\tau)} J^{N}_\mu\left(\hat{Y}^{k_0}\Phi\right)n^\mu_{\Sigma_\tau} +\iint_{\mathcal R(\tau_0,\tau)\cap\{r\leq r^-_Y\}}K^{N}\left(\hat{Y}^{k_0}\Phi\right)\\
=&\int_{\Sigma_{\tau_0}} J^{N}_\mu\left(\hat{Y}^{k_0}\Phi\right)n^\mu_{\Sigma_{\tau_0}}+e\iint_{\mathcal R(\tau_0,\tau)\cap\{r^-_Y\leq r\leq r^+_Y\}} K^Y\left(\hat{Y}^{k_0}\Phi\right)\\
&+\iint_{\mathcal R(\tau_0,\tau)\cap\{r\leq r^+_Y\}} \left(\partial_{t^*}\hat{Y}^{k_0}\Phi+e\hat{Y}^{k_0+1}\Phi\right) \left(-\kappa^{k_0}\hat{Y}^{k_0+1}\Phi+O(1)\hat{Y}^{k_0}\partial_{t^*}\Phi\right.\\
&\left.\quad\quad\quad\quad\quad\quad\quad\quad\quad\quad\quad\quad\quad\quad\quad\quad\quad\quad+O(\epsilon)D^{k_0+1}\Phi+O(1)\sum_{j=1}^{k_0}D^{j}\Phi+\hat{Y}^kG\right).
\end{split}
\end{equation*}
The crucial observation in \cite{DRK} is that one of the inhomogeneous terms has a good sign and thus gives
\begin{equation*}
\begin{split}
&\int_{\Sigma_\tau} J^{N}_\mu\left(\hat{Y}^{k_0}\Phi\right)n^\mu_{\Sigma_\tau} +\int_{\mathcal H(\tau',\tau)} J^{N}_\mu\left(\hat{Y}^{k_0}\Phi\right)n^\mu_{\Sigma_\tau} +\iint_{\mathcal R(\tau',\tau)\cap\{r\leq r^-_Y\}}K^{N}\left(\hat{Y}^{k_0}\Phi\right)\\
&+\iint_{\mathcal R(\tau',\tau)}\left(\hat{Y}^{k_0+1}\Phi\right)^2\\
\leq&C\left(\int_{\Sigma_{\tau'}\cap\{r\leq r^+_Y\}} J^{N}_\mu\left(\hat{Y}^{k_0}\Phi\right)n^\mu_{\Sigma_{\tau_0}}+\iint_{\mathcal R(\tau',\tau)\cap\{r^-_Y\leq r\leq r^+_Y\}} K^N\left(\hat{Y}^{k_0}\Phi\right)+\epsilon\iint_{\mathcal R(\tau',\tau)\cap\{r\leq r^+_Y\}}\left(D^{k_0+1}\Phi\right)^2\right.\\
&\left.+\iint_{\mathcal R(\tau',\tau)\cap\{r\leq r^+_Y\}} J^N_\mu\left(\partial_{t^*}\hat{Y}^{k_0-1}\Phi\right)n^\mu_{\Sigma_{t^*}}+\sum_{j=1}^{k_0}\iint_{\mathcal R(\tau',\tau)\cap\{r\leq r^+_Y\}}\left(D^{j}\Phi\right)^2+\iint_{\mathcal R(\tau',\tau)\cap\{r\leq r^+_Y\}}(\hat{Y}^{k_0}G)^2\right)\\
\leq&C\left(\int_{\Sigma_{\tau'}\cap\{r\leq r^+_Y\}} J^{N}_\mu\left(\hat{Y}^{k_0}\Phi\right)n^\mu_{\Sigma_{\tau_0}}+\sum_{j=1}^{k_0+1}\iint_{\mathcal R(\tau',\tau)\cap\{r^-_Y\leq r\leq r^+_Y\}} \left(D^{j}\Phi\right)^2+\epsilon\iint_{\mathcal R(\tau',\tau)\cap\{r\leq r^-_Y\}}\left(D^{k_0+1}\Phi\right)^2\right.\\
&\left.+\iint_{\mathcal R(\tau',\tau)\cap\{r\leq r^-_Y\}} J^N_\mu\left(\partial_{t^*}\hat{Y}^{k_0-1}\Phi\right)n^\mu_{\Sigma_{t^*}}+\sum_{j=1}^{k_0}\iint_{\mathcal R(\tau',\tau)\cap\{r\leq r^-_Y\}}\left(D^{j}\Phi\right)^2+\iint_{\mathcal R(\tau',\tau)\cap\{r\leq r^+_Y\}}(\hat{Y}^{k_0}G)^2\right).
\end{split}
\end{equation*}
Using Proposition \ref{elliptic} with an appropriate cutoff,
\begin{equation*}
\begin{split}
&\sum_{j=1}^{k_0+1}\iint_{\mathcal R(\tau',\tau)\cap\{r^-_Y\leq r\leq r^+_Y\}} \left(D^{j}\Phi\right)^2\\
\leq&C\left(\sum_{j=0}^{k_0}\iint_{\mathcal R(\tau',\tau)\cap\{r\leq \frac{23M}{8}\}} \left(\Phi^2+J^N_\mu\left(\partial_{t^*}^j\Phi\right)n^\mu_{\Sigma_{t^*}}\right)+\sum_{j=0}^{k_0-1}\iint_{\mathcal R(\tau',\tau)\cap\{r\leq \frac{23M}{8}\}}\left(D^jG\right)^2\right).
\end{split}
\end{equation*}
Using Proposition \ref{elliptichorizon}, 
\begin{equation*}
\begin{split}
&\sum_{j=1}^{k_0}\iint_{\mathcal R(\tau',\tau)\cap\{r\leq r^-_Y\}}\left(D^{j}\Phi\right)^2\\
\leq&C\left(\sum_{j+m\leq k_0-1}\iint_{\mathcal R(\tau',\tau)\cap\{r\leq r^-_Y\}} J^N_\mu\left(\partial_{t^*}^j\hat{Y}^m\Phi\right)n^\mu_{\Sigma_{t^*}}+\sum_{j=0}^{k_0-2}\iint_{\mathcal R(\tau',\tau)\cap\{r\leq r^-_Y\}}\left(D^jG\right)^2\right)\\
\leq &C\left(\sum_{j+m\leq k_0-1}\int_{\Sigma_{\tau'}\cap\{r\leq r^+_Y\}} J^N_\mu\left(\partial_{t^*}^j\hat{Y}^m\Phi\right)n^\mu_{\Sigma_{\tau'}}+\sum_{j=0}^{ k-1}\int_{\Sigma_{\tau}\cap\{r\leq r^+_Y\}} J^N_\mu\left(\partial_{t^*}^j\Phi\right)n^\mu_{\Sigma_{\tau}}\right.\\
+&\left.\sum_{j+m\leq k_0-1}\iint_{\mathcal R(\tau',\tau)\cap\{r\leq \frac{23M}{8}\}} J^N_\mu\left(\partial_{t^*}^j\hat{Y}^m\Phi\right)n^\mu_{\Sigma_{t^*}}+\sum_{j=0}^{k_0-1}\iint_{\mathcal R(\tau',\tau)\cap\{r\leq\frac{23M}{8}\}}\left(D^jG\right)^2\right),
\end{split}
\end{equation*}
using the induction hypothesis (on $\partial_{t^*}^m\Phi$ instead of $\Phi$) at the last step.
Similarly, using Proposition \ref{elliptichorizon},
\begin{equation*}
\begin{split}
&\epsilon\iint_{\mathcal R(\tau',\tau)\cap\{r\leq r^-_Y\}}\left(D^{k_0+1}\Phi\right)^2\\
\leq&C\epsilon\left(\sum_{j+m\leq k_0}\iint_{\mathcal R(\tau',\tau)\cap\{r\leq r^-_Y\}} J^N_\mu\left(\partial_{t^*}^j\hat{Y}^m\Phi\right)n^\mu_{\Sigma_{t^*}} +\sum_{j=0}^{k-1}\iint_{\mathcal R(\tau',\tau)\cap\{r\leq r^-_Y\}}\left(D^jG\right)^2\right)\\
\leq&C\epsilon\left(\iint_{\mathcal R(\tau',\tau)\cap\{r\leq r^+_Y\}} J^N_\mu\left(\hat{Y}^{k_0}\Phi\right)n^\mu_{\Sigma_{t^*}}+\sum_{j=0}^{k_0}\int_{\Sigma_{\tau'}\cap\{r\leq r^+_Y\}} J^N_\mu\left(\partial_{t^*}^j\Phi\right)n^\mu_{\Sigma_{\tau'}}\right.\\
&\left.+\sum_{j=0}^{ k_0}\iint_{\mathcal R(\tau',\tau)\cap\{r\leq \frac{23M}{8}\}} J^N_\mu\left(\partial_{t^*}^j\Phi\right)n^\mu_{\Sigma_{t^*}}+\sum_{j=0}^{k_0}\iint_{\mathcal R(\tau',\tau)\cap\{r\leq\frac{23M}{8}\}}\left(D^jG\right)^2\right),
\end{split}
\end{equation*}
where again the induction hypotheses is used at the last step.
All these together give
\begin{equation*}
\begin{split}
&\int_{\Sigma_\tau} J^{N}_\mu\left(\hat{Y}^{k_0}\Phi\right)n^\mu_{\Sigma_\tau} +\int_{\mathcal H(\tau_0,\tau)} J^{N}_\mu\left(\hat{Y}^{k_0}\Phi\right)n^\mu_{\Sigma_\tau} +\iint_{\mathcal R(\tau_0,\tau)\cap\{r\leq r^-_Y\}}K^{N}\left(\hat{Y}^{k_0}\Phi\right)\\
\leq &C\left(\sum_{j+m\leq k_0}\int_{\Sigma_{\tau'}\cap\{r\leq r^+_Y\}} J^N_\mu\left(\partial_{t^*}^j\hat{Y}^m\Phi\right)n^\mu_{\Sigma_{\tau'}}+\sum_{j=0}^{k_0}\int_{\Sigma_{\tau'}\cap\{r\leq r^+_Y\}} J^N_\mu\left(\partial_{t^*}^j\Phi\right)n^\mu_{\Sigma_{\tau'}}\right.\\
&\left.+\sum_{j=0}^{k_0}\iint_{\mathcal R(\tau',\tau)\cap\{r\leq \frac{23M}{8}\}}\left(\Phi^2+ J^N_\mu\left(\partial_{t^*}^j\Phi\right)n^\mu_{\Sigma_{t^*}}\right)+\sum_{j=0}^{k_0}\int_{\mathcal R(\tau',\tau)\cap\{r\leq\frac{23M}{8}\}}\left(D^jG\right)^2\right.\\
&\left.+\epsilon\iint_{\mathcal R(\tau',\tau)\cap\{r\leq r^-_Y\}} J^N_\mu\left(\hat{Y}^{k_0}\Phi\right)n^\mu_{\Sigma_{t^*}}\right).
\end{split}
\end{equation*}
The Proposition can be proved by noticing that
$$\iint_{\mathcal R(\tau',\tau)\cap\{r\leq r^-_Y\}} J^N_\mu\left(\hat{Y}^{k_0}\Phi\right)n^\mu_{\Sigma_{t^*}}\leq C\iint_{\mathcal R(\tau_0,\tau)\cap\{r\leq r^-_Y\}}K^{N}\left(\hat{Y}^{k_0}\Phi\right).$$
and absorbing the small term to the left hand side.
\end{proof}

We now specialize to the case $\Box_{g_K}\Phi=0$. The above Proposition implies that the behavior of $\hat{Y}^k\Phi$ is determined by the behavior of $\partial_{t^*}^m\Phi$ in the region $\{r\leq \frac{23M}{8}\}$. 
\begin{proposition}\label{Yhomo}
Fix $k\geq0$. Suppose $\Box_{g_K}\Phi=0$ and suppose for some constants $\alpha, B>0$ (independents of $\tau$),
$$\sum_{j=0}^{k}\int_{\Sigma_{\tau}\cap\{r\leq \frac{23M}{8}\}} \left(\Phi^2+J^N_\mu\left(\partial_{t^*}^j\Phi\right)n^\mu_{\Sigma_{\tau}}\right)\leq CB\tau^{-\alpha}.$$
Then
\begin{equation*}
\begin{split}
\sum_{j+m\leq k}\int_{\Sigma_\tau} J^{N}_\mu\left(\partial_{t^*}^j\hat{Y}^{m}\Phi\right)n^\mu_{\Sigma_\tau}
\leq &C\tau^{-\alpha}\left(\sum_{j+m\leq k}\int_{\Sigma_{\tau_0}} J^N_\mu\left(\partial_{t^*}^j\hat{Y}^m\Phi\right)n^\mu_{\Sigma_{\tau_0}}+B\right),\\
\end{split}
\end{equation*}
and
\begin{equation*}
\begin{split}
\sum_{j+m\leq k}\iint_{\mathcal R(\tau',\tau)\cap\{r\leq r^-_Y\}}K^{N}\left(\partial_{t^*}^j\hat{Y}^{m}\Phi\right)
\leq &C(\tau')^{-\alpha}\left(\sum_{j+m\leq k}\int_{\Sigma_{\tau_0}} J^N_\mu\left(\partial_{t^*}^j\hat{Y}^m\Phi\right)n^\mu_{\Sigma_{\tau_0}}+B\right).\\
\end{split}
\end{equation*}

\end{proposition}
\begin{remark}
In the applications, we will apply this Proposition with $B$ being some energy quantity of the initial condition. 
\end{remark}
\begin{proof}
We will proof this with a bootstrap argument. Suppose for all $\tau$ that 
\begin{equation}\label{bootstrapcommY}
\begin{split}
&\sum_{j+m\leq k}\int_{\Sigma_\tau\cap\{r\leq r^+_Y\}} J^{N}_\mu\left(\partial_{t^*}^j\hat{Y}^{m}\Phi\right)n^\mu_{\Sigma_\tau} 
\leq A\tau^{-\alpha}\left(\sum_{j+m\leq k}\int_{\Sigma_{\tau_0}} J^N_\mu\left(\partial_{t^*}^j\hat{Y}^m\Phi\right)n^\mu_{\Sigma_{\tau_0}}+B\right).\\
\end{split}
\end{equation}
This obviously holds initially for any $A\geq 1$ (and in particular independent of $\Phi$). 
By taking $\tau'=\tau-K,$ for some (large and to be chosen) constant $K$ and $\tau \geq 2K$, Proposition \ref{commYcontrol} implies
\begin{equation*}
\begin{split}
&\sum_{j+m\leq k}\left(\int_{\Sigma_{\tau}\cap\{r\leq r^+_Y\}} J^{N}_\mu\left(\partial_{t^*}^j\hat{Y}^{m}\Phi\right)n^\mu_{\Sigma_{\tau}}  +\iint_{\mathcal R(\tau-K,\tau)\cap\{r\leq r^-_Y\}}J^{N}_\mu\left(\partial_{t^*}^j\hat{Y}^{m}\Phi\right)n^\mu_{\Sigma_{t^*}}\right)\\
\leq &C\left(\sum_{j+m\leq k}\int_{\Sigma_{\tau-K}\cap\{r\leq r^+_Y\}} J^N_\mu\left(\partial_{t^*}^j\hat{Y}^m\Phi\right)n^\mu_{\Sigma_{\tau-K}}+\sum_{j=0}^k\int_{\Sigma_{\tau}\cap\{r\leq r^+_Y\}} J^N_\mu\left(\partial_{t^*}^j\Phi\right)n^\mu_{\Sigma_{\tau}}\right.\\
&\left.+\sum_{j=0}^k\iint_{\mathcal R(\tau-K,\tau)\cap\{ r\leq \frac{23M}{8}\}} \left(\Phi^2+J^N_\mu\left(\partial_{t^*}^j\Phi\right)n^\mu_{\Sigma_{t^*}}\right)\right)\\
\leq &C\left(\sum_{j+m\leq k}\int_{\Sigma_{\tau-K}\cap\{r\leq r^+_Y\}} J^N_\mu\left(\partial_{t^*}^j\hat{Y}^m\Phi\right)n^\mu_{\Sigma_{\tau-K}}+KB\tau^{-\alpha}\right)\\
&\quad\quad\mbox{using the assumption of the Proposition and using Proposition \ref{elliptic}}\\
\leq &C\left(A\left(\tau-K\right)^{-\alpha}\left(\sum_{j+m\leq k}\int_{\Sigma_{\tau_0}} J^N_\mu\left(\partial_{t^*}^j\hat{Y}^m\Phi\right)n^\mu_{\Sigma_{\tau_0}}+B\right)+KB\tau^{-\alpha}\right)\\
&\quad\quad\mbox{using the bootstrap assumption}\\
\leq &C\tau^{-\alpha}\left(\sum_{j+m\leq k}A\int_{\Sigma_{\tau_0}} J^N_\mu\left(\partial_{t^*}^j\hat{Y}^m\Phi\right)n^\mu_{\Sigma_{\tau_0}}+AB+KB\right).\\
\end{split}
\end{equation*}
Notice that $C$ is independent of $K$. By selecting a $t^*$ slice, we have that for some $\tilde{\tau}$,
\begin{equation*}
\begin{split}
&\int_{\Sigma_{\tilde{\tau}}\cap\{r\leq r^-_Y\}} J^{N}_\mu\left(\hat{Y}^{k}\Phi\right)n^\mu_{\Sigma_{\tilde{\tau}}}\\
\leq &CK^{-1}\tau^{-\alpha}\left(\sum_{j+m\leq k}A\int_{\Sigma_{\tau_0}} J^N_\mu\left(\partial_{t^*}^j\hat{Y}^m\Phi\right)n^\mu_{\Sigma_{\tau_0}}+AB+KB\right).\\
\end{split}
\end{equation*}
Now apply Proposition \ref{commYcontrol} on $[\tilde{\tau},\tau]$ to get
\begin{equation*}
\begin{split}
&\sum_{j+m\leq k}\int_{\Sigma_{\tau}\cap\{r\leq r^+_Y\}} J^{N}_\mu\left(\partial_{t^*}^j\hat{Y}^{m}\Phi\right)n^\mu_{\Sigma_{\tau}}\\
\leq &C\left(\sum_{j+m\leq k}\int_{\Sigma_{\tilde{\tau}}\cap\{r\leq r^+_Y\}} J^N_\mu\left(\partial_{t^*}^j\hat{Y}^m\Phi\right)n^\mu_{\Sigma_{\tilde{\tau}}}+\sum_{j=0}^k\int_{\Sigma_{\tau}\cap\{r\leq r^+_Y\}} J^N_\mu\left(\partial_{t^*}^j\Phi\right)n^\mu_{\Sigma_{\tau}}\right.\\
&\left.+\sum_{j=0}^k\iint_{\mathcal R(\tilde{\tau},\tau)\cap\{ r\leq \frac{23M}{8}\}} \left(\Phi^2+J^N_\mu\left(\partial_{t^*}^j\Phi\right)n^\mu_{\Sigma_{t^*}}\right)\right)\\
\leq &CK^{-1}\tau^{-\alpha}\left(\sum_{j+m\leq k}A\int_{\Sigma_{\tau_0}} J^N_\mu\left(\partial_{t^*}^j\hat{Y}^m\Phi\right)n^\mu_{\Sigma_{\tau_0}}+AB+BK\right)+CB(K+1)\tau^{-\alpha}\\
\leq &CAK^{-1}\tau^{-\alpha}\sum_{j+m\leq k}\int_{\Sigma_{\tau_0}} J^N_\mu\left(\partial_{t^*}^j\hat{Y}^m\Phi\right)n^\mu_{\Sigma_{\tau_0}}+\left(CAK^{-1}+CK+C\right)B\tau^{-\alpha}.\\
\end{split}
\end{equation*}
This would improve (\ref{bootstrapcommY}) if we choose $K=4C$ and $A$ sufficiently large.
Hence we have proved 
\begin{equation*}
\begin{split}
&\sum_{j+m\leq k}\int_{\Sigma_\tau\cap\{r\leq r^+_Y\}} J^{N}_\mu\left(\partial_{t^*}^j\hat{Y}^{m}\Phi\right)n^\mu_{\Sigma_\tau} 
\leq C\tau^{-\alpha}\left(\sum_{j+m\leq k}\int_{\Sigma_{\tau_0}} J^N_\mu\left(\partial_{t^*}^j\hat{Y}^m\Phi\right)n^\mu_{\Sigma_{\tau_0}}+B\right).\\
\end{split}
\end{equation*}
To prove the second statement of the Proposition, we simply use the first statement and Proposition \ref{commYcontrol}.
\end{proof}
We can use Corollary \ref{DRDecay} to show the decay of $\hat{Y}^k\Phi$.
\begin{corollary}\label{YDecay}
Suppose $\Box_{g_K}\Phi=0$. Then for $\tau'\leq\tau\leq (1.1)\tau'$,
\begin{equation*}
\begin{split}
&\sum_{j+m\leq k}\int_{\Sigma_\tau\cap\{r\leq r^+_Y\}} J^{N}_\mu\left(\partial_{t^*}^j\hat{Y}^{m}\Phi\right)n^\mu_{\Sigma_\tau}\\
\leq &C\tau^{-2+\eta}\left(\sum_{j+m\leq k}\int_{\Sigma_{\tau_0}} J^N_\mu\left(\partial_{t^*}^j\hat{Y}^m\Phi\right)n^\mu_{\Sigma_{\tau_0}}+\sum_{j=0}^{k+2}\int_{\Sigma_{\tau_0}} J^{Z+CN}_\mu\left(\partial_{t^*}^j\Phi\right)n^\mu_{\Sigma_{\tau_0}}\right),\\
\end{split}
\end{equation*}
and
\begin{equation*}
\begin{split}
&\sum_{j+m\leq k}\iint_{\mathcal R(\tau',\tau)\cap\{r\leq r^-_Y\}}K^{N}\left(\partial_{t^*}^j\hat{Y}^{m}\Phi\right)\\
\leq &C\tau^{-2+\eta}\left(\sum_{j+m\leq k}\int_{\Sigma_{\tau_0}} J^N_\mu\left(\partial_{t^*}^j\hat{Y}^m\Phi\right)n^\mu_{\Sigma_{\tau_0}}+\sum_{j=0}^{k+2}\int_{\Sigma_{\tau_0}} J^{Z+CN}_\mu\left(\partial_{t^*}^j\Phi\right)n^\mu_{\Sigma_{\tau_0}}\right).\\
\end{split}
\end{equation*}

\end{corollary}

\section{Estimates for $\tilde{\Omega}\Phi$}\label{sectioncommutatoromega}
In this section, we would like to prove estimates for $\tilde{\Omega}^\ell\Phi$. The estimates for $\tilde{\Omega}\Phi$ are useful to provide an extra factor of $r$ in the energy estimates.
\begin{proposition}\label{improveomega}
\begin{equation*}
\begin{split}
\int_{\Sigma_{\tau}\cap\{r\geq r^-_Y\}} r^2|\nabb^2\Phi|^2 \leq C\int_{\Sigma_{\tau}} J^N_\mu\left(\Phi,\partial_{t^*}\Phi,\tilde{\Omega}\Phi\right)n^\mu_{\Sigma_\tau}.
\end{split}
\end{equation*}
\begin{equation*}
\begin{split}
\iint_{\mathcal R(\tau',\tau)\cap\{r\geq r^-_Y\}} r^{1-\delta}|\nabb^2\Phi|^2 \leq C\iint_{\mathcal R(\tau',\tau)} r^{-1-\delta}J^N_\mu\left(\Phi,\partial_{t^*}\Phi,\tilde{\Omega}\Phi\right)n^\mu_{\Sigma_{t^*}}.
\end{split}
\end{equation*}
\begin{equation*}
\begin{split}
\iint_{\mathcal R(\tau',\tau)\cap\{r\geq r^-_Y\}} r^{1-\delta}|\nabb^2\Phi|^2 \leq C\iint_{\mathcal R(\tau',\tau)} K^{X_1}\left(\Phi,\partial_{t^*}\Phi\right)+K^{X_0}\left(\tilde{\Omega}\Phi\right).
\end{split}
\end{equation*}
\end{proposition}
In order to prove such estimates, we commute $\Box_{g_K}$ with $\tilde{\Omega}$. Recall from Section \ref{commOmega} that 
\begin{equation*}
|[\Box_{g_K}, \tilde{\Omega}]\Phi|\leq Cr^{-2}\left(|D^2\Phi|+|D\Phi|\right) \quad \mbox{everywhere}, \mbox{ and}
\end{equation*}
\begin{equation*}
[\Box_{g_K}, \tilde{\Omega}]\Phi=0, \quad \mbox{for } r<R_\Omega.
\end{equation*}
Now suppose $\Box_{g_K}\Phi=0$. We have 
$$\Box_{g_K}\left(\tilde{\Omega}^\ell\Phi\right)=\sum_{j=0}^ {\ell-1}\tilde{\Omega}^j[\Box_{g_K}, \tilde{\Omega}]\tilde{\Omega}^{\ell-j-1}\Phi=:G_{\Omega,\ell}.$$
Since $[D,\tilde{\Omega}]=D$, we have 
$$|G_{\Omega,\ell}|\leq Cr^{-2}\left(\sum_{j=0}^{\ell-1}\left(|D^2\tilde{\Omega}^j\Phi|+|D\tilde{\Omega}^j\Phi|\right)+\sum_{j=0}^{\ell+1}|D^j\Phi|\right),$$
and $G_{\Omega,\ell}$ is supported in $\{r\geq R_\Omega\}$.
\begin{definition}
$$E_{\Omega,\ell}=\iint_{\mathcal R(\tau'-1,\tau+1)}r^{1+\delta}G_{\Omega,\ell}^2+\sup_{t^*\in [\tau'-1,\tau+1]}\int_{\Sigma_{t^*}\cap\{|r-3M|\leq\frac{M}{8}\}} {G_{\Omega,\ell}}^2.$$
\end{definition}
This is the error term for the energy estimates for $\tilde{\Omega}^\ell\Phi$. We show that this can be controlled.
\begin{proposition}\label{omegaerror}
\begin{equation*}
\begin{split}
E_{\Omega,\ell}\leq &C\sum_{m=0}^1\sum_{j=0}^{\ell-1}\iint_{\mathcal R(\tau'-1,\tau+1)\cap\{r\geq R_\Omega\}}K^{X_0}\left(\partial_{t^*}^m\tilde{\Omega}^j\Phi\right)+C\sum_{m=0}^{\ell}\iint_{\mathcal R(\tau'-1,\tau+1)\cap\{r\geq R_\Omega\}}K^{X_0}\left(\partial_{t^*}^m\Phi\right).
\end{split}
\end{equation*}

\end{proposition}
\begin{proof}
\begin{equation*}
\begin{split}
&\iint_{\mathcal R(\tau'-1,\tau+1)}r^{1+\delta}G_{\Omega,\ell}^2\\
\leq&C\sum_{j=0}^{\ell-1}\iint_{\mathcal R(\tau'-1,\tau+1)\cap\{r\geq R_\Omega\}}r^{-3+\delta}\left(\left(D^2\tilde{\Omega}^j\Phi\right)^2+\left(D\tilde{\Omega}^j\Phi\right)^2\right)\\
&+C\sum_{j=0}^{\ell+1}\iint_{\mathcal R(\tau'-1,\tau+1)\cap\{r\geq R_\Omega\}}r^{-3+\delta}\left(D^j\Phi\right)^2\\
\leq&C\sum_{m=0}^1\sum_{j=0}^{\ell-1}\iint_{\mathcal R(\tau'-1,\tau+1)\cap\{r\geq R_\Omega\}}r^{-3+\delta}J^N_\mu\left(\partial_{t^*}^m\tilde{\Omega}^j\Phi\right)n^\mu_{\Sigma{t^*}}\\
&+C\sum_{m=0}^{\ell}\iint_{\mathcal R(\tau'-1,\tau+1)\cap\{r\geq R_\Omega\}}r^{-3+\delta}J^N_\mu\left(\partial_{t^*}^m\Phi\right)n^\mu_{\Sigma{t^*}}\\
\leq&C\sum_{m=0}^1\sum_{j=0}^{\ell-1}\iint_{\mathcal R(\tau'-1,\tau+1)\cap\{r\geq R_\Omega\}}K^{X_0}\left(\partial_{t^*}^m\tilde{\Omega}^j\Phi\right)+C\sum_{m=0}^{\ell}\iint_{\mathcal R(\tau'-1,\tau+1)\cap\{r\geq R_\Omega\}}K^{X_0}\left(\partial_{t^*}^m\Phi\right).\\
\end{split}
\end{equation*}
By choosing $R_\Omega$ sufficiently large, the second term of $E_{\Omega,\ell}$ vanishes.
\end{proof}
We can show that the non-degenerate energy of $\tilde{\Omega}^\ell\Phi$ is almost bounded.
\begin{proposition}\label{bddomega}
Suppose $\Box_{g_K}\Phi=0$. Then
\begin{equation*}
\begin{split}
&\int_{\Sigma_{\tau}} J^{N}_\mu\left(\tilde{\Omega}^\ell\Phi\right)n^\mu_{\Sigma_{\tau}} +\int_{\mathcal H(\tau_0,\tau)} J^{N}_\mu\left(\tilde{\Omega}^\ell\Phi\right)n^\mu_{\mathcal H^+} +\iint_{\mathcal R(\tau_0,\tau)\cap\{r\leq r^-_Y\}}K^{N}\left(\tilde{\Omega}^\ell\Phi\right)+\iint_{\mathcal R(\tau_0,\tau)}K^{X_0}\left(\tilde{\Omega}^\ell\Phi\right)\\
\leq &C\sum_{i+j\leq\ell}\int_{\Sigma_{\tau_0}} J^{N}_\mu\left(\partial_{t^*}^i\tilde{\Omega}^j\Phi\right)n^\mu_{\Sigma_{\tau_0}}.\\
\end{split}
\end{equation*}
\end{proposition}
\begin{proof}
We prove this by induction on $\ell$. The $\ell=0$ case is true by setting $G=0$ in Proposition \ref{bddcom}. We assume that the Proposition is true for $\ell\leq \ell_0-1$. This in particular implies, after a commutations with the Killing vector field $\partial_{t^*}$, that
\begin{equation*}
\begin{split}
&\sum_{j=0}^{\ell_0-1}\iint_{\mathcal R(\tau_0,\tau)}K^{X_0}\left(\partial_{t^*}^m\tilde{\Omega}^j\Phi\right)\leq C\sum_{i+j\leq m+\ell_0-1}\int_{\Sigma_{\tau_0}} J^{N}_\mu\left(\partial_{t^*}^i\tilde{\Omega}^j\Phi\right)n^\mu_{\Sigma_{\tau_0}}.\\
\end{split}
\end{equation*}
By Propositions \ref{bddnoloss} and \ref{omegaerror},
\begin{equation*}
\begin{split}
&\int_{\Sigma_{\tau}}J^N_\mu\left(\tilde{\Omega}^{\ell_0}\Phi\right)n^\mu_{\Sigma_{\tau}}+\int_{\mathcal H(\tau_0,\tau)}J^N_\mu\left(\tilde{\Omega}^{\ell_0}\Phi\right)n^\mu_{\mathcal H^+}+\iint_{\mathcal R(\tau_0,\tau)\cap\{r\leq r^-_Y\}}K^{N}\left(\tilde{\Omega}^{\ell_0}\Phi\right)+\iint_{\mathcal R(\tau_0,\tau)}K^{X_0}\left(\tilde{\Omega}^{\ell_0}\Phi\right)\\
\leq &C\left(\int_{\Sigma_{\tau_0}} J^{N}_\mu\left(\tilde{\Omega}^{\ell_0}\Phi\right)n^\mu_{\Sigma_{\tau_0}}+\iint_{\mathcal R(\tau'-1,\tau+1)}r^{1+\delta}G_{\Omega,\ell_0}^2+\sup_{t^*\in [\tau'-1,\tau+1]}\int_{\Sigma_{t^*}\cap\{|r-3M|\leq\frac{M}{8}\}} G_{\Omega,{\ell_0}}^2\right)\\
\leq &C\left(\int_{\Sigma_{\tau_0}} J^{N}_\mu\left(\tilde{\Omega}^{\ell_0}\Phi\right)n^\mu_{\Sigma_{\tau_0}}+C\sum_{m=0}^1\sum_{j=0}^{{\ell_0}-1}\iint_{\mathcal R(\tau_0-1,\tau+1)\cap\{r\geq R_\Omega\}}K^{X_0}\left(\partial_{t^*}^m\tilde{\Omega}^j\Phi\right)\right.\\
&\left.+C\sum_{m=0}^{\ell_0}\iint_{\mathcal R(\tau_0-1,\tau+1)\cap\{r\geq R_\Omega\}}K^{X_0}\left(\partial_{t^*}^m\Phi\right)\right)\\
\leq &C\sum_{i+j\leq\ell_0}\int_{\Sigma_{\tau_0}} J^{N}_\mu\left(\partial_{t^*}^i\tilde{\Omega}^{j}\Phi\right)n^\mu_{\Sigma_{\tau_0}}.
\end{split}
\end{equation*}

\end{proof}
\begin{remark}
Only the $\ell=1$ case will be used.
\end{remark}

\section{Estimates for $S\Phi$}\label{sectioncommutatorS}
We will now use the energy estimates that we have derived to control $S\Phi$. In particular, we would like to prove a local integrated decay estimate for $S\Phi$. This will be used in the next section where we prove our main theorem. Recall from Section \ref{commS} that for $r$ large
$$|[\Box_{g_K},S]\Phi-\left(2+\frac{r^*\mu}{r}\right)\Box_{g_K} \Phi -\frac{2}{r}\left(\frac{r^*}{r}-1-\frac{2r^*\mu}{r}\right)\partial_{r^*}\Phi-2\left(\left(\frac{r^*}{r}-1\right)-\frac{3r^*\mu}{2r}\right)\lapp\Phi|\leq C\epsilon r^{-2}(\sum_{k=1}^2|\partial^k\Phi|),$$
and that for $r\leq R$, we have 
$$|[\Box_{g_K},S]\Phi|\leq C(\sum_{k=1}^2|D^k\Phi|).$$

From now on we will prove estimates for $S\Phi$ by considering the wave equation that it satisfies. We will assume, as before, $\Box_{g_K}\Phi=0$ and let $G$ denote the commutator term, i.e., $\Box_{g_K}(S\Phi)=G$. If we look at our estimates in the previous sections, we will need to control $G$ in three different norms. We now consider them separately.

\begin{proposition}\label{Sinho1}
Let $\tau'\leq\tau\leq (1.1)\tau'$. Then
\begin{equation*}
\begin{split}
&\sum_{m=0}^{\ell}\iint_{\mathcal R(\tau'-1,\tau+1)} r^{1+\delta}\left(\partial_{t^*}^{m}G\right)^2\\
\leq&C\tau^{-1+\eta}\sum_{m+k+j\leq \ell+3} \left(\int_{\Sigma_{\tau_0}} J^{Z,w^Z}_\mu\left(\partial_{t^*}^m\hat{Y}^k\tilde{\Omega}^j\Phi\right)n^\mu_{\Sigma_{{\tau_0}}}+C\int_{\Sigma_{\tau_0}} J^{N}_\mu\left(\partial_{t^*}^m\hat{Y}^k\tilde{\Omega}^j\Phi\right)n^\mu_{\Sigma_{{\tau_0}}}\right).\\
\end{split}
\end{equation*}
and
\begin{equation*}
\begin{split}
&\sum_{m=0}^{\ell}\iint_{\mathcal R(\tau'-1,\tau+1)\cap\{r\leq\frac{9t^*}{10}\}} r^{1+\delta}\left(\partial_{t^*}^{m}G\right)^2\\
\leq&C\tau^{-2+\eta}\sum_{m+k+j\leq \ell+4} \left(\int_{\Sigma_{\tau_0}} J^{Z,w^Z}_\mu\left(\partial_{t^*}^m\hat{Y}^k\tilde{\Omega}^j\Phi\right)n^\mu_{\Sigma_{{\tau_0}}}+C\int_{\Sigma_{\tau_0}} J^{N}_\mu\left(\partial_{t^*}^m\hat{Y}^k\tilde{\Omega}^j\Phi\right)n^\mu_{\Sigma_{{\tau_0}}}\right).\\
\end{split}
\end{equation*}
In other words, we can get more decay if we localize and allow an extra derivative.
\end{proposition}

\begin{proof}
\begin{equation*}
\begin{split}
&\sum_{m=0}^{\ell}\iint_{\mathcal R(\tau'-1,\tau+1)}r^{1+\delta}\left(\partial_{t^*}^{m}G\right)^2\\
\leq&C\sum_{m=0}^{\ell}\iint_{\mathcal R(\tau'-1,\tau+1)}r^{-3+\delta}\left(\left(\partial_{t^*}^{m}D^2\Phi\right)^2+\left(\partial_{t^*}^{m}D\Phi\right)^2+\left(r\partial_{t^*}^m\lapp\Phi\right)^2\right)\\
&\quad\quad\mbox{noting that the }\delta\mbox{ in the two lines are different}\\
\leq&C\sum_{m+k\leq \ell+1}\left(\iint_{\mathcal R(\tau'-1,\tau+1)\cap\{r\leq\frac{t^*}{2}\}}r^{-1+\delta}J^N_\mu\left(\partial_{t^*}^m\hat{Y}^k\Phi\right)n^\mu_{\Sigma_t^*}\right.\\
+&\left.\iint_{\mathcal R(\tau'-1,\tau+1)\cap\{ r\geq\frac{t^*}{2}}r^{-3+\delta}J^N_\mu\left(\partial_{t^*}^m\tilde{\Omega}^k\Phi\right)n^\mu_{\Sigma_t^*}\right)\\
&\quad\quad\mbox{by Proposition \ref{elliptic}, \ref{elliptichorizon} and \ref{improveomega}}\\
\leq&C\sum_{m+k\leq \ell+3} \tau^{-1+\eta}\left(\int_{\Sigma_{\tau_0}} J^{Z,w^Z}_\mu\left(\partial_{t^*}^m\Phi\right)n^\mu_{\Sigma_{{\tau_0}}}+C\int_{\Sigma_{\tau_0}} J^{N}_\mu\left(\partial_{t^*}^m\hat{Y}^k\Phi\right)n^\mu_{\Sigma_{{\tau_0}}}\right)\\
&+C\sum_{m+k\leq \ell+3} \int_{\tau'-1}^{\tau+1}(t^*)^{-3+\delta}\left(\int_{\Sigma_{\tau_0}} J^{Z,w^Z}_\mu\left(\partial_{t^*}^m\Phi\right)n^\mu_{\Sigma_{{\tau_0}}}+C\int_{\Sigma_{\tau_0}} J^{N}_\mu\left(\partial_{t^*}^m\hat{Y}^k\Phi\right)n^\mu_{\Sigma_{{\tau_0}}}\right)dt^*\\
&+C\sum_{m\leq \ell+1} \int_{\tau'-1}^{\tau+1}(t^*)^{-3+\delta}\left(\int_{\Sigma_{\tau_0}} J^{N}_\mu\left(\partial_{t^*}^m\tilde{\Omega}\Phi\right)n^\mu_{\Sigma_{{\tau_0}}}\right)dt^*\\
&\quad\quad\mbox{using Corollaries \ref{DRDecay}, \ref{YDecay} and Proposition \ref{bddomega}}\\
\leq&C\tau^{-1+\eta}\sum_{m+k+j\leq \ell+3} \left(\int_{\Sigma_{\tau_0}} J^{Z,w^Z}_\mu\left(\partial_{t^*}^m\hat{Y}^k\tilde{\Omega}^j\Phi\right)n^\mu_{\Sigma_{{\tau_0}}}+C\int_{\Sigma_{\tau_0}} J^{N}_\mu\left(\partial_{t^*}^m\hat{Y}^k\tilde{\Omega}^j\Phi\right)n^\mu_{\Sigma_{{\tau_0}}}\right).\\
\end{split}
\end{equation*}
We then move on to the localized version:
\begin{equation*}
\begin{split}
&\sum_{m=0}^{\ell}\iint_{\mathcal R(\tau'-1,\tau+1)\cap\{r\leq\frac{9t^*}{10}\}}r^{1+\delta}\left(\partial_{t^*}^{m}G\right)^2\\
\leq&C\sum_{m=0}^{\ell}\iint_{\mathcal R(\tau'-1,\tau+1)\cap\{r\leq\frac{9t^*}{10}\}}r^{-3+\delta}\left(\left(\partial_{t^*}^{m}D^2\Phi\right)^2+\left(\partial_{t^*}^{m}D\Phi\right)^2+\left(r\partial_{t^*}^m\lapp\Phi\right)^2\right)\\
\leq&C\sum_{m+k\leq \ell+1}\left(\iint_{\mathcal R(\tau'-1,\tau+1)\cap\{r\leq\frac{t^*}{2}\}}r^{-1+\delta}J^N_\mu\left(\partial_{t^*}^m\hat{Y}\Phi\right)n^\mu_{\Sigma_t^*}\right.\\
&\left.+\iint_{\mathcal R(\tau'-1,\tau+1)\cap\{ \frac{t^*}{2}\leq r\leq\frac{19t^*}{20}\}}r^{-3+\delta}J^N_\mu\left(\partial_{t^*}^m\hat{Y}\Phi\right)n^\mu_{\Sigma_t^*}\right)\\
+&C\sum_{m=0}^1\left(\iint_{\mathcal R(\tau'-1,\tau+1)\cap\{ \frac{t^*}{2}\leq r\leq\frac{19t^*}{20}\}}r^{-3+\delta}J^N_\mu\left(\partial_{t^*}^m\tilde{\Omega}\Phi\right)n^\mu_{\Sigma_t^*}\right)\\
&\quad\quad\mbox{by Proposition \ref{elliptic}, \ref{elliptichorizon} and \ref{improveomega}}\\
\leq&C\sum_{m+k\leq \ell+4} \tau^{-2+\eta}\left(\int_{\Sigma_{\tau_0}} J^{Z,w^Z}_\mu\left(\partial_{t^*}^m\Phi\right)n^\mu_{\Sigma_{{\tau_0}}}+C\int_{\Sigma_{\tau_0}} J^{N}_\mu\left(\partial_{t^*}^m\hat{Y}\Phi\right)n^\mu_{\Sigma_{{\tau_0}}}\right)\\
&+C\sum_{m+k\leq \ell+3} \int_{\tau'-1}^{\tau+1}(t^*)^{-3+\delta}\left(\int_{\Sigma_{\tau_0}} J^{Z,w^Z}_\mu\left(\partial_{t^*}^m\Phi\right)n^\mu_{\Sigma_{{\tau_0}}}+C\int_{\Sigma_{\tau_0}} J^{N}_\mu\left(\partial_{t^*}^m\hat{Y}\Phi\right)n^\mu_{\Sigma_{{\tau_0}}}\right)dt^*\\
&+C\sum_{m=0}^\ell \int_{\tau'-1}^{\tau+1}(t^*)^{-3+\delta}\left(\int_{\Sigma_{\tau_0}} J^{N}_\mu\left(\partial_{t^*}^m\tilde{\Omega}\Phi\right)n^\mu_{\Sigma_{{\tau_0}}}\right)dt^*\\
&\quad\quad\mbox{using Corollaries \ref{DRDecay}, \ref{YDecay} and Proposition \ref{bddomega}}\\
\leq&C\tau^{-2+\eta}\sum_{m+k+j\leq \ell+4} \left(\int_{\Sigma_{\tau_0}} J^{Z,w^Z}_\mu\left(\partial_{t^*}^m\hat{Y}^k\tilde{\Omega}^j\Phi\right)n^\mu_{\Sigma_{{\tau_0}}}+C\int_{\Sigma_{\tau_0}} J^{N}_\mu\left(\partial_{t^*}^m\hat{Y}^k\tilde{\Omega}^j\Phi\right)n^\mu_{\Sigma_{{\tau_0}}}\right).\\
\end{split}
\end{equation*}
\end{proof}

To estimate the inhomogeneous term in the region $r\leq\frac{t^*}{2}$, we would also need to estimate a term not integrated over $t^*$ which arises from the integration by parts.

\begin{proposition}\label{Sinho2}
For $\tau'\leq\tau\leq(1.1)\tau'$,
\begin{equation*}
\begin{split}
&\sup_{t^*\in [\tau'-1,\tau+1]}\sum_{m=0}^\ell\int_{\Sigma_{t^*}\cap\{|r-3M|\leq\frac{M}{8}\}} (\partial_{t^*}^m G)^2\\
\leq&C\tau^{-2+\eta}\sum_{m+j\leq \ell+3}\left(\int_{\Sigma_{\tau_0}} J^{Z,w^Z}_\mu\left(\partial_{t^*}^m\tilde{\Omega}^j\Phi\right)n^\mu_{\Sigma_{{\tau_0}}}+C\int_{\Sigma_{\tau_0}} J^{N}_\mu\left(\partial_{t^*}^m\tilde{\Omega}^j\Phi\right)n^\mu_{\Sigma_{{\tau_0}}}\right).
\end{split}
\end{equation*}
\end{proposition}

\begin{proof}
\begin{equation*}
\begin{split}
&\sup_{t^*\in [\tau'-1,\tau+1]}\sum_{m=0}^\ell\int_{\Sigma_{t^*}\cap\{|r-3M|\leq\frac{M}{8}\}} (\partial_{t^*}^mG)^2\\
\leq&C \sup_{t^*\in [\tau'-1,\tau+1]}\sum_{m=0}^\ell\int_{\Sigma_{t^*}\cap\{r^-_Y\leq r\leq \frac{25M}{8}\}}\left(\left(D^2\partial_{t^*}^m\Phi\right)^2+\left(D\partial_{t^*}^m\Phi\right)^2+\left(r\lapp\partial_{t^*}^m\Phi\right)^2\right)\\
\leq&C\sup_{t^*\in [\tau'-1,\tau+1]}\left(\sum_{m=0}^{\ell+1}\int_{\Sigma_{t^*}\cap\{r^-_Y\leq r\leq \frac{25M}{8}\}}J^N_\mu\left(\partial_{t^*}^m\Phi\right)n^\mu_{\Sigma_t^*}+\sum_{m=0}^\ell\int_{\Sigma_{t^*}\cap\{r^-_Y\leq r\leq \frac{25M}{8}\}}J^N_\mu\left(\tilde{\Omega}\partial_{t^*}^m\Phi\right)n^\mu_{\Sigma_t^*}\right)\\
&\quad\quad\mbox{by Proposition \ref{elliptic} and \ref{improveomega}}\\
\leq&C\tau^{-2+\eta}\sum_{m+j\leq \ell+3}\left(\int_{\Sigma_{\tau_0}} J^{Z,w^Z}_\mu\left(\partial_{t^*}^m\tilde{\Omega}^j\Phi\right)n^\mu_{\Sigma_{{\tau_0}}}+C\int_{\Sigma_{\tau_0}} J^{N}_\mu\left(\partial_{t^*}^m\tilde{\Omega}^j\Phi\right)n^\mu_{\Sigma_{{\tau_0}}}\right),\\
&\quad\quad\mbox{using Corollary \ref{DRDecay} and Proposition \ref{bddomega}}\\
\end{split}
\end{equation*}
\end{proof}

Finally, we estimate the third norm:

\begin{proposition}\label{Sinho3}
\begin{equation*}
\begin{split}
&\sum_{m=0}^\ell\left(\int_{\tau_0}^{\tau}\left(\int_{\Sigma_{t^*}\cap\{r\geq \frac{t^*}{2}\}}r^2 (\partial_{t^*}^mG)^2 \right)^{\frac{1}{2}}dt^*\right)^2
\leq C\tau^{\eta}\sum_{m+j\leq \ell+3}\int_{\Sigma_{\tau_0}} J^N_\mu\left(\partial_{t^*}^m\tilde{\Omega}^j\Phi\right)n^\mu_{\Sigma_{\tau_0}}. \\
\end{split}
\end{equation*}
\end{proposition}
\begin{proof}
\begin{equation*}
\begin{split}
&\sum_{m=0}^\ell\left(\int_{\tau_0}^{\tau}\left(\int_{\Sigma_{t^*}\cap\{r\geq \frac{t^*}{2}\}}r^2 (\partial_{t^*}^mG)^2 \right)^{\frac{1}{2}}dt^*\right)^2\\
\leq &C\sum_{m=0}^\ell\left(\int_{\tau_0}^{\tau}(t^*)^{-1+\delta}\left(\int_{\Sigma_{t^*}\cap\{r\geq \frac{t^*}{2}\}} \left(\left(D^2\Phi\right)^2+\left(D\partial_{t^*}^m\Phi\right)^2+\left(r\lapp\partial_{t^*}^m\Phi\right)^2\right) \right)^{\frac{1}{2}}dt^*\right)^2\\
\leq &C\left(\int_{\tau_0}^{\tau}(t^*)^{-1+\delta}\left(\sum_{m=0}^{\ell+1}\int_{\Sigma_{t^*}} J^N_\mu\left(\partial_{t^*}^m\Phi\right)n^\mu_{\Sigma_{t^*}} +\sum_{m=0}^\ell\int_{\Sigma_{t^*}} J^N_\mu\left(\tilde{\Omega}\partial_{t^*}^m\Phi\right)n^\mu_{\Sigma_{t^*}}\right)^{\frac{1}{2}}dt^*\right)^2\\
\leq &C\tau^{\eta}\sum_{m+j\leq \ell+3}\int_{\Sigma_{\tau_0}} J^N_\mu\left(\partial_{t^*}^m\tilde{\Omega}^j\Phi\right)n^\mu_{\Sigma_{\tau_0}}. \\
\end{split}
\end{equation*}
\end{proof}

Now that we have control of the inhomogeneous terms in the equation $\Box_{g_K}\Phi=G$, we can prove the decay of $S\Phi$. To this end, we will introduce the bootstrap assumptions:
\begin{equation}\label{bootstrapS1}
\begin{split}
&c\int_{\Sigma_{\tau}} J^{Z,w^Z}_\mu\left(\partial_{t^*}S\Phi\right)n^\mu_{\Sigma_{{\tau}}}+\tau^2\int_{\Sigma_{\tau}\cap\{r\leq \gamma\tau\}} J^{N}_\mu\left(\partial_{t^*}S\Phi\right)n^\mu_{\Sigma_{\tau}}\\
\leq &A\tau\sum_{m=1}^2\int_{\Sigma_{\tau_0}} J^{Z+CN,w^Z}_\mu\left(\partial_{t^*}^mS\Phi\right)n^\mu_{\Sigma_{{\tau_0}}}+A\tau^{1+\eta}\sum_{m+k+j\leq 5} \int_{\Sigma_{\tau_0}} J^{Z+CN,w^Z}_\mu\left(\partial_{t^*}^m\hat{Y}^k\tilde{\Omega}^j\Phi\right)n^\mu_{\Sigma_{{\tau_0}}}.\\
\end{split}
\end{equation}
\begin{equation}\label{bootstrapS2}
\begin{split}
&c\int_{\Sigma_{\tau}} J^{Z,w^Z}_\mu\left(S\Phi\right)n^\mu_{\Sigma_{{\tau}}}+\tau^2\int_{\Sigma_{\tau}\cap\{r\leq \gamma\tau\}} J^{N}_\mu\left(S\Phi\right)n^\mu_{\Sigma_{\tau}}\\
\leq &A^2\tau^\eta\left(\sum_{m=0}^2\int_{\Sigma_{\tau_0}} J^{Z+CN,w^Z}_\mu\left(\partial_{t^*}^mS\Phi\right)n^\mu_{\Sigma_{{\tau_0}}}+\sum_{m+k+j\leq 5} \int_{\Sigma_{\tau_0}} J^{Z+CN,w^Z}_\mu\left(\partial_{t^*}^m\hat{Y}^k\tilde{\Omega}^j\Phi\right)n^\mu_{\Sigma_{{\tau_0}}}\right).\\
\end{split}
\end{equation}
We think of $A$ as some large constant to be chosen. We will improve the constants $A$ and $A^2$ in the above assumptions. Under these two assumptions, we will get the following three estimates for the bulk terms:
\begin{proposition}
 \begin{equation*}
 \begin{split}
 &\iint_{\mathcal R(\tau_0,\tau)}K^{X_1}\left(\partial_{t^*} S\Phi\right)\\
\leq& C\left(\sum_{m=1}^2\int_{\Sigma_{\tau_0}} J^{N}_\mu\left(\partial_{t^*}^mS\Phi\right)n^\mu_{\Sigma_{\tau_0}}+\sum_{m+k+j\leq 5} \int_{\Sigma_{\tau_0}} J^{Z+CN,w^Z}_\mu\left(\partial_{t^*}^m\hat{Y}\tilde{\Omega}^j\Phi\right)n^\mu_{\Sigma_{{\tau_0}}}\right).\\
 \end{split}
\end{equation*}
\end{proposition}
\begin{proof}
By Proposition \ref{bddcom1} for the equation $\Box_{g_K}\left(\partial_{t^*}S\Phi\right)=\partial_{t^*}G$, taking $\tau'=\tau_0$ and $G_1=0$, $G_2=\partial_{t^*}G$. Then use Propositions \ref{Sinho1} and \ref{Sinho2} to estimate the terms with $G$.
\end{proof}

\begin{proposition}
For $\tau'\leq\tau\leq (1.1)\tau'$,
 \begin{equation*}
 \begin{split}
 &\iint_{\mathcal R(\tau',\tau)\cap\{r\leq r^-_Y\}}K^{N}\left(\partial_{t^*}S\Phi\right)+\iint_{\mathcal R(\tau',\tau)\cap\{r\leq\frac{t^*}{2}\}}K^{X_0}\left(\partial_{t^*}S\Phi\right)+\iint_{\mathcal R(\tau_0,\tau)}r^{-1+\delta}K^{X_1}\left( S\Phi\right)\\
\leq& CA\left(\tau^{-2}\int_{\Sigma_{\tau'}} J^{Z,w^Z}_\mu\left(S\Phi\right)n^\mu_{\Sigma_{\tau'}}+ C\int_{\Sigma_{\tau'}\cap\{r\leq r^-_Y\}} J^{N}_\mu\left(S\Phi\right)n^\mu_{\Sigma_{\tau'}}\right)\\
&+CA\tau^{-1+\eta}\sum_{m+k+j\leq 5} \int_{\Sigma_{\tau_0}} J^{Z+CN,w^Z}_\mu\left(\partial_{t^*}^m\hat{Y}\tilde{\Omega}^j\Phi\right)n^\mu_{\Sigma_{{\tau_0}}}.\\
 \end{split}
\end{equation*}
\end{proposition}
\begin{proof}
By Propositions \ref{localization} and \ref{localizationT}, taking $G_1=0$ and $G_2=G$, and using Propositions \ref{Sinho1} and \ref{Sinho2} to estimate the terms with $G$, we have
\begin{equation*}
 \begin{split}
 &\iint_{\mathcal R(\tau',\tau)\cap\{r\leq r^-_Y\}}K^{N}\left(\partial_{t^*}S\Phi\right)+\iint_{\mathcal R(\tau',\tau)\cap\{r\leq\frac{t^*}{2}\}}K^{X_0}\left(\partial_{t^*}S\Phi\right)+\iint_{\mathcal R(\tau_0,\tau)\cap\{r\leq\frac{t^*}{2}\}}K^{X_1}\left( S\Phi\right)\\
\leq& CA\left(\tau^{-2}\int_{\Sigma_{\tau'}} J^{Z,w^Z}_\mu\left(S\Phi\right)n^\mu_{\Sigma_{\tau'}}+ C\int_{\Sigma_{\tau'}\cap\{r\leq r^-_Y\}} J^{N}_\mu\left(S\Phi\right)n^\mu_{\Sigma_{\tau'}}\right)\\
&+CA\tau^{-1+\eta}\sum_{m+k+j\leq 5} \int_{\Sigma_{\tau_0}} J^{Z+CN,w^Z}_\mu\left(\partial_{t^*}^m\hat{Y}\tilde{\Omega}^j\Phi\right)n^\mu_{\Sigma_{{\tau_0}}}.\\
 \end{split}
\end{equation*}
It now remains to estimate $r^{-1+\delta}K^{X_1}$ in the region $r\geq\frac{t^*}{2}$. Here, we will use crucially the decay in $r$. Clearly,
$$ \iint_{\mathcal R(\tau_0,\tau)\cap\{r\geq\frac{t^*}{2}\}}r^{-1+\delta}K^{X_1}\left( S\Phi\right)\leq C\tau^{-1+\delta}\iint_{\mathcal R(\tau_0,\tau)}K^{X_1}\left( S\Phi\right).$$
Then we can estimate the right hand side by Proposition \ref{bddcom1}, taking $\tau'=\tau_0$ and $G_1=0$, $G_2=\partial_{t^*}G$. Then use Propositions \ref{Sinho1} and \ref{Sinho2} to estimate the terms with $G$.
\end{proof}

\begin{proposition}\label{Sint}
For $\tau'\leq\tau\leq (1.1)\tau'$,
 \begin{equation*}
 \begin{split}
 &\iint_{\mathcal R(\tau',\tau)\cap\{r\leq r^-_Y\}}K^{N}\left(S\Phi\right)+\iint_{\mathcal R(\tau',\tau)\cap\{r\leq\frac{t^*}{2}\}}K^{X_0}\left(S\Phi\right)\\
\leq& CA^2\left(\tau^{-2}\int_{\Sigma_{\tau'}} J^{Z,w^Z}_\mu\left(S\Phi\right)n^\mu_{\Sigma_{\tau'}}+ C\int_{\Sigma_{\tau'}\cap\{r\leq r^-_Y\}} J^{N}_\mu\left(S\Phi\right)n^\mu_{\Sigma_{\tau'}}\right)\\
&+CA^2\tau^{-2+\eta}\sum_{m+k+j\leq 4} \int_{\Sigma_{\tau_0}} J^{Z+CN,w^Z}_\mu\left(\partial_{t^*}^m\hat{Y}\tilde{\Omega}^j\Phi\right)n^\mu_{\Sigma_{{\tau_0}}}.\\
 \end{split}
\end{equation*}

\end{proposition}
\begin{proof}
This follows from using Proposition \ref{localization}, taking $G_1=0$ and $G_2=G$, and using Propositions \ref{Sinho1} and \ref{Sinho2} to estimate the terms with $G$.
\end{proof}
We are now ready to retrieve the bootstrap assumptions. First, we retrieve the bootstrap assumption \ref{bootstrapS1}:
\begin{proposition}
\begin{equation*}
\begin{split}
&c\int_{\Sigma_{\tau}} J^{Z,w^Z}_\mu\left(\partial_{t^*}S\Phi\right)n^\mu_{\Sigma_{{\tau}}}+\tau^2\int_{\Sigma_{\tau}\cap\{r\leq \gamma\tau\}} J^{N}_\mu\left(\partial_{t^*}S\Phi\right)n^\mu_{\Sigma_{\tau}}\\
\leq &\frac{A}{2}\tau\sum_{m=1}^2\int_{\Sigma_{\tau_0}} J^{Z+CN,w^Z}_\mu\left(\partial_{t^*}^mS\Phi\right)n^\mu_{\Sigma_{{\tau_0}}}+\frac{A}{2}\tau^{1+\eta}\sum_{m+k+j\leq 5} \int_{\Sigma_{\tau_0}} J^{Z+CN,w^Z}_\mu\left(\partial_{t^*}^m\hat{Y}^k\tilde{\Omega}^j\Phi\right)n^\mu_{\Sigma_{{\tau_0}}}.\\
\end{split}
\end{equation*}
\end{proposition}
\begin{proof}
By Proposition \ref{energydecay},
\begin{equation*}
\begin{split}
&c\int_{\Sigma_{\tau}} J^{Z,w^Z}_\mu\left(\partial_{t^*}S\Phi\right)n^\mu_{\Sigma_{{\tau}}}+\tau^2\int_{\Sigma_{\tau}\cap\{r\leq \gamma\tau\}} J^{N}_\mu\left(\partial_{t^*}S\Phi\right)n^\mu_{\Sigma_{\tau}}\\
\leq &C\int_{\Sigma_{\tau_0}} J^{Z+CN,w^Z}_\mu\left(\partial_{t^*}S\Phi\right)n^\mu_{\Sigma_{{\tau_0}}}+C\iint_{\mathcal R(\tau_0,\tau)} t^*r^{-1+\delta}K^{X_1}\left(\partial_{t^*}S\Phi\right)\\
& +C\delta'\iint_{\mathcal R(\tau_0,\tau)\cap\{r\leq \frac{t^*}{2}\}} (t^*)^2K^{X_0}\left(\partial_{t^*}S\Phi\right)+C\left(\delta'+\epsilon\right)\iint_{\mathcal R(\tau_0,\tau)\cap\{r\leq r^-_Y\}}(t^*)^2K^N\left(\partial_{t^*}S\Phi\right)\\
&+C(\delta')^{-1}\left(\int_{\tau_0}^{\tau}\left(\int_{\Sigma_{t^*}\cap\{r\geq \frac{t^*}{2}\}}r^2 (\partial_{t^*}G)^2 \right)^{\frac{1}{2}}dt^*\right)^2+C(\delta')^{-1}\sum_{m=1}^2\iint_{\mathcal R(\tau_0,\tau)\cap\{r\leq\frac{9t^*}{10}\}} (t^*)^2r^{1+\delta}\left(\partial_{t^*}^m G\right)^2\\
&+C(\delta')^{-1}\sup_{t^*\in [\tau_0,\tau]}\int_{\Sigma_{t^*}\cap\{r^-_Y\leq r\leq \frac{25M}{8}\}} (t^*)^2 (\partial_{t^*}G)^2.
\end{split}
\end{equation*}
It suffices to check that by Propositions \ref{Sinho1}, \ref{Sinho2} and \ref{Sinho3}, all terms are acceptable.
\end{proof}
We can now retrieve the bootstrap assumption (\ref{bootstrapS2}).
\begin{proposition}
\begin{equation*}
\begin{split}
&c\int_{\Sigma_{\tau}} J^{Z,w^Z}_\mu\left(S\Phi\right)n^\mu_{\Sigma_{{\tau}}}+\tau^2\int_{\Sigma_{\tau}\cap\{r\leq \gamma\tau\}} J^{N}_\mu\left(S\Phi\right)n^\mu_{\Sigma_{\tau}}\\
\leq &A^2\tau^\eta\left(\sum_{m=0}^2\int_{\Sigma_{\tau_0}} J^{Z+CN,w^Z}_\mu\left(\partial_{t^*}^mS\Phi\right)n^\mu_{\Sigma_{{\tau_0}}}+\sum_{m+k+j\leq 5} \int_{\Sigma_{\tau_0}} J^{Z+CN,w^Z}_\mu\left(\partial_{t^*}^m\hat{Y}^k\tilde{\Omega}^j\Phi\right)n^\mu_{\Sigma_{{\tau_0}}}\right).\\
\end{split}
\end{equation*}
\end{proposition}
\begin{proof}
By Proposition \ref{energydecay},
\begin{equation*}
\begin{split}
&c\int_{\Sigma_{\tau}} J^{Z,w^Z}_\mu\left(S\Phi\right)n^\mu_{\Sigma_{{\tau}}}+\tau^2\int_{\Sigma_{\tau}\cap\{r\leq \gamma\tau\}} J^{N}_\mu\left(S\Phi\right)n^\mu_{\Sigma_{\tau}}\\
\leq &C\int_{\Sigma_{\tau_0}} J^{Z+CN,w^Z}_\mu\left(S\Phi\right)n^\mu_{\Sigma_{{\tau_0}}}+C\iint_{\mathcal R(\tau_0,\tau)} t^*r^{-1+\delta}K^{X_1}\left(S\Phi\right)\\
& +C\delta'\iint_{\mathcal R(\tau_0,\tau)\cap\{r\leq \frac{t^*}{2}\}} (t^*)^2K^{X_0}\left(S\Phi\right)+C\left(\delta'+\epsilon\right)\iint_{\mathcal R(\tau_0,\tau)\cap\{r\leq r^-_Y\}}(t^*)^2K^N\left(S\Phi\right)\\
&+C(\delta')^{-1}\left(\int_{\tau_0}^{\tau}\left(\int_{\Sigma_{t^*}\cap\{r\geq \frac{t^*}{2}\}}r^2 G^2 \right)^{\frac{1}{2}}dt^*\right)^2+C(\delta')^{-1}\sum_{m=0}^1\iint_{\mathcal R(\tau_0,\tau)\cap\{r\leq\frac{9t^*}{10}\}} (t^*)^2r^{1+\delta}\left(\partial_{t^*}^m G\right)^2\\
&+C(\delta')^{-1}\sup_{t^*\in [\tau_0,\tau]}\int_{\Sigma_{t^*}\cap\{r^-_Y\leq r\leq \frac{25M}{8}\}} (t^*)^2 G^2.
\end{split}
\end{equation*}
It suffices to check that by Propositions \ref{Sinho1}, \ref{Sinho2} and \ref{Sinho3}, all terms are acceptable.
\end{proof}

We have thus showed the following:
\begin{proposition}\label{decayS}
For all $\eta>0$, there exists $\epsilon>0$ small enough such that for Kerr spacetimes satisfying (\ref{smallness}), the following estimates hold:
\begin{equation*}
\begin{split}
&c\int_{\Sigma_{\tau}} J^{Z,w^Z}_\mu\left(S\Phi\right)n^\mu_{\Sigma_{{\tau}}}+\tau^2\int_{\Sigma_{\tau}\cap\{r\leq \gamma\tau\}} J^{N}_\mu\left(S\Phi\right)n^\mu_{\Sigma_{\tau}}\\
\leq &C\tau^\eta\sum_{m=0}^2\int_{\Sigma_{\tau_0}} J^{Z+CN,w^Z}_\mu\left(\partial_{t^*}^mS\Phi\right)n^\mu_{\Sigma_{{\tau_0}}}+C\tau^\eta\sum_{m+k+j\leq 5} \int_{\Sigma_{\tau_0}} J^{Z+CN,w^Z}_\mu\left(\partial_{t^*}^m\hat{Y}\tilde{\Omega}^j\Phi\right)n^\mu_{\Sigma_{{\tau_0}}}.\\
\end{split}
\end{equation*}
Moreover, for $\tau'\leq \tau\leq (1.1)\tau'$,
\begin{equation*}
\begin{split}
&\iint_{\mathcal R(\tau',\tau)\cap\{r\leq r^-_Y\}}K^{N}\left(S\Phi\right)+\iint_{\mathcal R(\tau',\tau)\cap\{r\leq\frac{t^*}{2}\}}K^{X_0}\left(S\Phi\right)\\
\leq &C\tau^{-2+\eta}\sum_{m=0}^2\int_{\Sigma_{\tau_0}} J^{Z+CN,w^Z}_\mu\left(\partial_{t^*}^mS\Phi\right)n^\mu_{\Sigma_{{\tau_0}}}+C\tau^{-2+\eta}\sum_{m+k+j\leq 5}\int_{\Sigma_{\tau_0}} J^{Z+CN,w^Z}_\mu\left(\partial_{t^*}^m\hat{Y}\tilde{\Omega}^j\Phi\right)n^\mu_{\Sigma_{{\tau_0}}}.\\
\end{split}
\end{equation*}
and
\begin{equation*}
\begin{split}
&\iint_{\mathcal R(\tau',\tau)\cap\{r\leq\frac{t^*}{2}\}}K^{X_1}\left(S\Phi\right)\\
\leq &C\tau^{-2+\eta}\sum_{m=0}^3\int_{\Sigma_{\tau_0}} J^{Z+CN,w^Z}_\mu\left(\partial_{t^*}^mS\Phi\right)n^\mu_{\Sigma_{{\tau_0}}}+C\tau^{-2+\eta}\sum_{m+k+j\leq 6}\int_{\Sigma_{\tau_0}} J^{Z+CN,w^Z}_\mu\left(\partial_{t^*}^m\hat{Y}\tilde{\Omega}^j\Phi\right)n^\mu_{\Sigma_{{\tau_0}}}.\\
\end{split}
\end{equation*}
\end{proposition}

\begin{proof}
The first statement is proved by the bootstrap above. Since the bootstrap assumptions are true, the conclusion in Proposition \ref{Sint} is also true, hence the second statement is true. The third statement makes use of the fact that $K^{X_1}$ can be estimated in the same way as $K^{X_0}$ with an extra derivative.
\end{proof}

\section{Improved Decay for the Linear Homogeneous Wave Equation}\label{sectionproof}
To use the estimates for $S\Phi$, we need to integrate along integral curves of $S$. We first find the integral curves by solving the ordinary differential equation
$$\frac{dr_S}{dt^*_S}=\frac{h(r_S)}{t^*_S}$$
where $h(r_S)$ is as in the definition of $S$. Hence the integral curves are given by
$$\frac{\exp\left(\int_{(r_S)_0}^{r_S} \frac{dr'_S}{h(r'_S)}\right)}{t^*_S}=\mbox{constant},$$
where $r_0>2M$ can be chosen arbitrarily.
Let $\sigma=t^*$, $\rho=\frac{\exp\left(\int_{(r_S)_0}^{r_S} \frac{dr_S'}{h(r_S')}\right)}{t_S^*}$ and consider $(\sigma,\rho, x^A, x^B)$ as a new system of coordinates. Notice that
$$\partial_\sigma=\frac{h(r_S)}{t^*}\partial_{r_S}+\partial_{t_S^*}=\frac{1}{t^*}S.$$
Now for each fixed $\rho$, we have
$$\Phi^2(\tau)\leq \Phi^2(\tau')+|\int_{\tau'}^\tau \frac{1}{\sigma}S(\Phi^2) d\sigma|.$$
Integrating along a finite region of $\rho$, we get:
$$\int_{\rho_1}^{\rho_2} \Phi^2(\tau) d\rho\leq \int_{\rho_1}^{\rho_2} \Phi^2(\tau') d\rho+\int_{\rho_1}^{\rho_2} \int_{\tau'}^\tau |\frac{2}{\sigma}\Phi S\Phi |d\sigma d\rho.$$
We would like to change coordinates back to $(t^*_S,r_S, x_S^A, x_S^B)$. Notice that since $h(r_S)$ is everywhere positive, $(\rho,\tau)$ would correspond to a point with a larger value of $r$ than $(\rho,\tau')$. Therefore,
\begin{equation*}
\begin{split}
&\int_{r_+}^{r_2} \Phi^2(\tau)\frac{\exp\left(\int_{(r_S)_0}^{r_S} \frac{dr_S'}{h(r_S')}\right)}{\tau h(r_S)}dr\\
\leq & \int_{r_+}^{r_2} \Phi^2(\tau')\frac{\exp\left(\int_{(r_S)_0}^{r_S} \frac{dr_S'}{h(r_S')}\right)}{\tau' h(r_S)}dr+\int_{\tau'}^\tau \int_{r_+}^{r_2} |\frac{2}{\sigma}\Phi S\Phi |\frac{\exp\left(\int_{(r_S)_0}^{r_S} \frac{dr_S'}{h(r_S')}\right)}{t_S^* h(r_S)}dr dt^*.
\end{split}
\end{equation*}
We have to compare $\frac{\exp\left(\int_{(r_S)_0}^{r_S} \frac{dr'_S}{h(r'_S)}\right)}{h(r_S)}$ with the volume form. Very close to the horizon, $h(r_S)=r_S-2M$. Hence 
$$\frac{\exp\left(\int_{(r_S)_0}^{r_S} \frac{dr'_S}{h(r'_S)}\right)}{h(r_S)}=e^{\int_{(r_S)_0}^{r_S}\frac{dr'_S}{h(r'_S)}}\left(\frac{1}{r_S-2M}\right)\sim 1.$$
The corresponding expression on the compact set $[r^-_Y,R]$ is obviously bounded. Hence we have
\begin{equation}\label{mainlemma}
\begin{split}
&\int_{\Sigma_{\tau}\cap\{r <r_2\}} \frac{\Phi^2(\tau)}{\tau}\leq C\left(\int_{\Sigma_{\tau'}\cap\{r <r_2\}} \frac{\Phi^2(\tau')}{\tau'}+\iint_{\mathcal R(\tau',\tau)\cap\{r <r_2\}} |\frac{2}{(t^*)^2}\Phi S\Phi |\right).
\end{split}
\end{equation}
This easily implies the following improved decay for the non-degenerate energy:
\begin{proposition}\label{mainprop}
\begin{equation*}
\begin{split}
&\int_{\Sigma_{\tau}\cap\{r <R  \}} \Phi^2\leq C_R\tau^{-1}\left(\iint_{\mathcal R((1.1)^{-1}\tau,\tau)\cap\{r < R\}} \Phi^2 +\iint_{\mathcal R((1.1)^{-1}\tau,\tau)\cap\{r < R \}}  \left(S\Phi\right)^2\right).
\end{split}
\end{equation*}
\end{proposition}
\begin{proof}
By choosing an appropriate $\tilde{\tau}\in[(1.1)^{-1}\tau,\tau]$, we have
$$\int_{\Sigma_{\tilde{\tau}}\cap\{r < R\}} \Phi^2\leq C\tau^{-1}\iint_{\mathcal R((1.1)^{-1}\tau,\tau)\cap\{r < R \}} \Phi^2.$$
Now, apply (\ref{mainlemma}) with $\tau'=\tilde{\tau}$, we have 
\begin{equation*}
\begin{split}
&\int_{\Sigma_{\tau}\cap\{r < R\}} \Phi^2\\
\leq &C\tau\left(\int_{\Sigma_{\tilde{\tau}}\cap\{r < R \}} \frac{\Phi^2}{\tilde{\tau}}+\iint_{\mathcal R(\tilde{\tau},\tau)\cap\{r < R \}} |\frac{2}{(t^*)^2}\Phi S\Phi |\right)\\
\leq &C\tau^{-1}\left(\iint_{\mathcal R((1.1)^{-1}\tau,\tau)\cap\{r < R \}} \Phi^2 +\iint_{\mathcal R((1.1)^{-1}\tau,\tau)\cap\{r < R \}}  \left(S\Phi\right)^2\right),
\end{split}
\end{equation*}
using Cauchy-Schwarz for the second term.
\end{proof}
We can now conclude with the improved decay for solutions to the homogeneous wave equation.\\
\begin{proof}[Proof of Main Theorem \ref{mainmaintheorem}]

By Proposition \ref{DRDecay}, \ref{mainprop} and \ref{decayS}, we have
\begin{equation*}
\begin{split}
&\int_{\Sigma_{\tau}\cap\{r <R  \}} \Phi^2\\
\leq &C_R\tau^{-1}\left(\iint_{\mathcal R((1.1)^{-1}\tau,\tau)\cap\{r < R\}} \Phi^2 +\iint_{\mathcal R((1.1)^{-1}\tau,\tau)\cap\{r < R \}}  \left(S\Phi\right)^2\right)\\
\leq &C_R\tau^{-1}\iint_{\mathcal R((1.1)^{-1}\tau,\tau)\cap\{r < R\}}\left( K^{X_0}\left(\Phi\right)+K^{X_0}\left(S\Phi\right)\right)\\
\leq&C_R\tau^{-3+\eta}\left(\sum_{m=0}^2\int_{\Sigma_{\tau_0}} J^{Z+CN,w^Z}_\mu\left(\partial_{t^*}^mS\Phi\right)n^\mu_{\Sigma_{{\tau_0}}}+\sum_{m+k+j\leq 5}\int_{\Sigma_{\tau_0}} J^{Z+CN,w^Z}_\mu\left(\partial_{t^*}^m\hat{Y}\tilde{\Omega}^j\Phi\right)n^\mu_{\Sigma_{{\tau_0}}}\right)
\end{split}
\end{equation*}
Similarly we can use Proposition \ref{mainprop} for the derivatives of $\Phi$. By Proposition \ref{DRDecay}, \ref{mainprop} and \ref{decayS}, we have
\begin{equation*}
\begin{split}
&\int_{\Sigma_{\tau}\cap\{r <R  \}} \left(D\Phi\right)^2\\
\leq &C_R\tau^{-1}\left(\iint_{\mathcal R((1.1)^{-1}\tau,\tau)\cap\{r < R\}} (D\Phi)^2 +\iint_{\mathcal R((1.1)^{-1}\tau,\tau)\cap\{r < R \}}  \left(SD\Phi\right)^2\right)\\
\leq &C_R\tau^{-1}\iint_{\mathcal R((1.1)^{-1}\tau,\tau)\cap\{r < R\}}\left( K^{X_1}\left(\Phi\right)+K^{X_1}\left(S\Phi\right)\right)\\
&\quad\quad\mbox{since we have the commutation } [D,S]=D\\
\leq&C_R\tau^{-3+\eta}\left(\sum_{m=0}^3\int_{\Sigma_{\tau_0}} J^{Z+CN,w^Z}_\mu\left(\partial_{t^*}^mS\Phi\right)n^\mu_{\Sigma_{{\tau_0}}}+\sum_{m+k+j\leq 6}\int_{\Sigma_{\tau_0}} J^{Z+CN,w^Z}_\mu\left(\partial_{t^*}^m\hat{Y}\tilde{\Omega}^j\Phi\right)n^\mu_{\Sigma_{{\tau_0}}}\right)
\end{split}
\end{equation*}
By commuting with $\partial_{t^*}$, we get
\begin{equation*}
\begin{split}
&\int_{\Sigma_{\tau}\cap\{r <R  \}} \left(D\partial_{t^*}^\ell\Phi\right)^2\\
\leq&C_R\tau^{-3+\eta}\left(\sum_{m=0}^{\ell+3}\int_{\Sigma_{\tau_0}} J^{Z+CN,w^Z}_\mu\left(\partial_{t^*}^mS\Phi\right)n^\mu_{\Sigma_{{\tau_0}}}+\sum_{m+k+j\leq \ell+6}\int_{\Sigma_{\tau_0}} J^{Z+CN,w^Z}_\mu\left(\partial_{t^*}^m\hat{Y}\tilde{\Omega}^j\Phi\right)n^\mu_{\Sigma_{{\tau_0}}}\right).
\end{split}
\end{equation*}
Without loss of generality, we can take $R>\frac{23M}{8}$. Then, by Proposition \ref{Yhomo},
\begin{equation*}
\begin{split}
&\sum_{j+m\leq \ell}\int_{\Sigma_\tau\cap\{r\leq r^+_Y\}} J^{N}_\mu\left(\partial_{t^*}^j\hat{Y}^{m}\Phi\right)n^\mu_{\Sigma_\tau}\\
\leq &C_R\tau^{-3+\eta}\left(\sum_{m=0}^{\ell+3}\int_{\Sigma_{\tau_0}} J^{Z+CN,w^Z}_\mu\left(\partial_{t^*}^mS\Phi\right)n^\mu_{\Sigma_{{\tau_0}}}+\sum_{m+k+j\leq \ell+6}\int_{\Sigma_{\tau_0}} J^{Z+CN,w^Z}_\mu\left(\partial_{t^*}^m\hat{Y}\tilde{\Omega}^j\Phi\right)n^\mu_{\Sigma_{{\tau_0}}}\right).\\
\end{split}
\end{equation*}
Hence, by Proposition \ref{elliptic} and \ref{elliptichorizon},
\begin{equation*}
\begin{split}
&\sum_{j=0}^\ell\int_{\Sigma_\tau\cap\{r\leq R\}} \left(D^j\Phi\right)^2\\
\leq &C_R\tau^{-3+\eta}\left(\sum_{m=0}^{\ell+2}\int_{\Sigma_{\tau_0}} J^{Z+CN,w^Z}_\mu\left(\partial_{t^*}^mS\Phi\right)n^\mu_{\Sigma_{{\tau_0}}}+\sum_{m+k+j\leq \ell+5}\int_{\Sigma_{\tau_0}} J^{Z+CN,w^Z}_\mu\left(\partial_{t^*}^m\hat{Y}\tilde{\Omega}^j\Phi\right)n^\mu_{\Sigma_{{\tau_0}}}\right).\\
\end{split}
\end{equation*}
The pointwise decay statement follows from standard Sobolev Embedding.
\end{proof}

\section{Discussion}
The Theorem that we proved in this paper holds in the set $\{r_+\leq r\leq R\}$ for any fixed $R$. It is however interesting also to derive the same estimates, for example, in the set $\{r_+\leq r\leq \frac{t^*}{2}\}$. This can be achieved by proving the full decay result when we commuted the equation with $\tilde{\Omega}^\ell$. Using this we can prove (with more loss in derivatives) that 
$$|\Phi|\leq CE(t^*)^{-\frac{3}{2}+\eta}r^{\eta},\quad |D\Phi|\leq CE(t^*)^{-\frac{3}{2}+\eta}r^{-\frac{1}{2}+\eta},$$
for $r\leq \frac{t^*}{2}$. This will be useful in studying nonlinear problems. This decay rate will be proved as a corollary in our forthcoming paper on the null condition.
\section{Acknowledgments}
The author thanks his advisor Igor Rodnianski for his continual support and encouragement and for many enlightening discussions. He thanks Gustav Holzegel for very helpful comments on the manuscript. He thanks also an anonymous referee for many suggestions to improve the manuscript.
\bibliographystyle{hplain}
\bibliography{Kerr}
\end{document}

%% file: penrose.pstex_t
\begin{picture}(0,0)%
\includegraphics{penrose.pstex}%
\end{picture}%
\setlength{\unitlength}{2289sp}%
\begingroup\makeatletter\ifx\SetFigFont\undefined%
\gdef\SetFigFont#1#2#3#4#5{%
  \reset@font\fontsize{#1}{#2pt}%
  \fontfamily{#3}\fontseries{#4}\fontshape{#5}%
  \selectfont}%
\fi\endgroup%
\begin{picture}(4824,4824)(3589,-6373)
\put(6601,-4186){\makebox(0,0)[lb]{\smash{{\SetFigFont{7}{8.4}{\rmdefault}{\mddefault}{\updefault}{\color[rgb]{0,0,0}$\Sigma_{\tau}$}%
}}}}
\put(6076,-3436){\makebox(0,0)[lb]{\smash{{\SetFigFont{7}{8.4}{\rmdefault}{\mddefault}{\updefault}{\color[rgb]{0,0,0}$\mathcal H^+$}%
}}}}
\put(7801,-3286){\makebox(0,0)[lb]{\smash{{\SetFigFont{7}{8.4}{\rmdefault}{\mddefault}{\updefault}{\color[rgb]{0,0,0}$\mathcal I^+$}%
}}}}
\end{picture}%